\documentclass[11pt]{amsart}
\usepackage{latexsym,amsmath,amssymb,amscd}
\usepackage{enumerate}
\def\today{\ifcase \month \or
   January \or February \or March \or April \or
   May \or June \or July \or August \or
   September \or October \or November \or December \fi
   \space\number\day , \number\year}

\begin{document}

\makeatletter
\@addtoreset{figure}{section}
\def\thefigure{\thesection.\@arabic\c@figure}
\def\fps@figure{h,t}
\@addtoreset{table}{bsection}

\def\thetable{\thesection.\@arabic\c@table}
\def\fps@table{h, t}
\@addtoreset{equation}{section}
\def\theequation{
\arabic{equation}}
\makeatother

\newcommand{\bfi}{\bfseries\itshape}

\newtheorem{theorem}{Theorem}
\newtheorem{acknowledgment}{Acknowledgment}
\newtheorem{claim}[theorem]{Claim}
\newtheorem{conclusion}[theorem]{Conclusion}
\newtheorem{condition}[theorem]{Condition}
\newtheorem{conjecture}[theorem]{Conjecture}
\newtheorem{construction}[theorem]{Construction}
\newtheorem{corollary}[theorem]{Corollary}
\newtheorem{criterion}[theorem]{Criterion}
\newtheorem{definition}[theorem]{Definition}
\newtheorem{example}[theorem]{Example}
\newtheorem{lemma}[theorem]{Lemma}
\newtheorem{notation}[theorem]{Notation}
\newtheorem{problem}[theorem]{Problem}
\newtheorem{proposition}[theorem]{Proposition}
\newtheorem{question}[theorem]{Question}
\newtheorem{remark}[theorem]{Remark}

\numberwithin{theorem}{section}
\numberwithin{equation}{section}

\newcommand{\todo}[1]{\vspace{5 mm}\par \noindent
\framebox{\begin{minipage}[c]{0.95 \textwidth}
#1 \end{minipage}}\vspace{5 mm}\par}

\newcommand{\1}{{\bf 1}}

\newcommand{\hotimes}{\widehat\otimes}

\newcommand{\Alt}{{\rm Alt}\,}
\newcommand{\Ci}{{\mathcal C}^\infty}
\newcommand{\comp}{\circ}
\newcommand{\D}{\text{\bf D}}
\newcommand{\ev}{{\rm ev}}
\newcommand{\id}{{\rm id}}
\newcommand{\ad}{{\rm ad}}
\newcommand{\de}{{\rm d}}
\newcommand{\dist}{{\rm dist}}
\newcommand{\Ad}{{\rm Ad}}
\newcommand{\ie}{{\rm i}}
\newcommand{\End}{{\rm End}\,}
\newcommand{\Gr}{{\rm Gr}_{\rm res}}
\newcommand{\Grr}{{\rm Gr}}
\newcommand{\Hom}{{\rm Hom}\,}
\newcommand{\Ker}{{\rm Ker}\,}
\newcommand{\Lie}{\text{\bf L}}
\newcommand{\lf}{{\rm l}}
\newcommand{\res}{{\rm res}}
\newcommand{\Ran}{{\rm Ran}\,}
\newcommand{\spann}{{\rm span}}
\newcommand{\Tr}{{\rm Tr}\,}

\newcommand{\G}{{\rm G}}
\newcommand{\U}{{\rm U}}
\newcommand{\Z}{{\rm Z}}
\newcommand{\VB}{{\rm VB}}

\newcommand{\Ac}{{\mathcal A}}
\newcommand{\Bc}{{\mathcal B}}
\newcommand{\Cc}{{\mathcal C}}
\newcommand{\Dc}{{\mathcal D}}
\newcommand{\Fc}{{\mathcal F}}
\newcommand{\Hc}{{\mathcal H}}
\newcommand{\Oc}{{\mathcal O}}
\newcommand{\Pc}{{\mathcal P}}
\newcommand{\Qc}{{\mathcal Q}}
\newcommand{\Wc}{{\mathcal W}}
\newcommand{\Xc}{{\mathcal X}}

\newcommand{\Bg}{{\mathfrak B}}
\newcommand{\Jg}{{\mathfrak J}}
\newcommand{\Lg}{{\mathfrak L}}
\newcommand{\Sg}{{\mathfrak S}}
\newcommand{\Xg}{{\mathfrak X}}
\newcommand{\ug}{{\mathfrak u}}
\newcommand{\g}{{\mathfrak g}}
\newcommand{\bg}{{\mathfrak b}}
\newcommand{\hg}{{\mathfrak h}}
\newcommand{\kg}{{\mathfrak k}}

\def\R{{\mathbb{R}}}
\def\C{{\mathbb{C}}}
\def\D{{\mathbb{D}}}
\def\N{{\mathbb{N}}}
\pagestyle{myheadings} \markboth{}{}


\date
\makeatletter
\title[Bruhat-Poisson structure of the restricted Grassmannian]{Banach Poisson--Lie groups and Bruhat-Poisson structure of the restricted Grassmannian}

\date{\today}
\author{A. B. Tumpach}

\makeatother
\maketitle

\begin{center}
\textit{Dedicated to T.S.~Ratiu for his 70's birthday and for the contagious enthusiasm he has in doing Mathematics.}
\end{center}
\vspace{.5cm}

\begin{center}
Laboratoire Paul Painlev\'e\\ Universit\'e de Lille\\Cit\'e scientifique\\59650 Villeneuve d'Ascq, \\France\\
\vspace{0.5cm}
\texttt{alice-barbora.tumpach@univ-lille.fr}\\
Orcid Identifier~: 0000-0002-7771-6758.

\end{center}

\begin{abstract}
The first part of this paper is devoted to the theory of Poisson--Lie groups in the Banach setting. Our starting point is the straightforward adaptation of the notion of Manin triples to the Banach context. The difference with the finite-dimensional case lies in the fact that a duality pairing between two non-reflexive Banach spaces is necessary weak  (as opposed to a strong pairing where one Banach space can be identified with the dual space of the other). The notion of generalized Banach Poisson manifolds introduced in this paper is compatible with weak duality pairings between the tangent space and a subspace of the dual. We investigate related notion like Banach Lie bialgebras and Banach Poisson--Lie groups, suitably generalized to the non-reflexive Banach context. 

The second part of the paper is devoted to the treatment of particular examples of Banach Poisson--Lie groups related to the restricted Grassmannian and the KdV hierarchy. More precisely, we construct a Banach Poisson--Lie group structure on the unitary restricted Banach Lie group which acts transitively on the restricted Grassmannian. A``dual'' Banach Lie group consisting of (a class of) upper triangular bounded operators admits also a Banach Poisson--Lie group structure of the same kind. We show that the restricted Grassmannian inherits a generalized Banach Poisson structure from the unitary Banach Lie group, called Bruhat-Poisson structure. Moreover the action of the triangular Banach Poisson--Lie group on it is a Poisson map. This action generates the KdV hierarchy, and its orbits are the Schubert cells of the restricted Grassmannian. 
\\

\noindent{\it Keywords:} restricted Grassmannian; Bruhat decomposition; Poisson manifold;
coadjoint orbits; Poisson maps; Poisson--Lie groups, Lie bialgebras, Lie--Poisson spaces, Schatten ideals.

\end{abstract}

\tableofcontents

\newpage

\section*{Introduction}
%
%
Poisson--Lie groups  and Lie bialgebras were  introduced by Drinfel'd in \cite{Dr83}. From this starting point, these notions and their relations to integrable systems were extensively studied. We refer the readers to the very well documented papers \cite{Ko04}, \cite{STS91}, \cite{Lu90} and the references therein. For a more algebraic approach to Poisson--Lie groups and their relation to quantum groups we refer to \cite{BHRS11}. For more details about dual pairs of Poisson manifolds we refer to \cite{We83}, applications to the study of equations coming from fluid dynamics were given in \cite{GBV12}, \cite{GBV15} and \cite{GBV152},  and applications to geometric quantization can be found in \cite{BW12}.
The motivation to write the present paper comes mainly from the reading of \cite{LW90}, \cite{SW85} and \cite{PS88}. In \cite{LW90}, the Bruhat-Poisson structure of finite-dimensional Grassmannians were studied. In \cite{SW85}, the relation between the infinite-dimensional restricted Grassmannian and equations of the KdV hierarchy was established. In \cite{PS88}, the Schubert cells of the restricted Grassmannian were shown to be homogeneous spaces with respect to the natural action of some triangular group, which appears to contain the group  that generates the KdV hierarchy in \cite{SW85}. It is therefore natural to ask the following questions~:
\begin{question} Does the restricted Grassmannian carry a Bruhat-Poisson structure? Can the KdV hierarchy be related to a Poisson action of a Poisson--Lie group on the restricted Grassmannian?
\end{question}
The difficulties to answer these questions come mainly from the following facts
\begin{itemize}
\item taking the upper triangular part of some infinite-dimensional matrix does not preserve the Banach space of bounded operators, nor the Banach space of trace-class operators.
\item Iwasawa decompositions may not exist in the context of infinite-dimensional Banach Lie groups (see however \cite{Bel06} and \cite{BN10} where some Iwasawa type factorisations where established).
\end{itemize}
Related papers on Poisson geometry in the infinite-dimensional setting are  \cite{DGR15},  \cite{NST14}, \cite{OR03} and \cite{Z94} (see Section~3). 
Let us mention that a hierarchy of commuting Hamiltonian equations related to the restricted Grassmannian was described in \cite{GO10}. In the aforementionned paper, the method of F.~Magri was used to generate the integrals of motions. It would be interesting to explore the link between equations studied in \cite{GO10} and the Bruhat-Poisson structure of the restricted Grassmannian introduced in the present paper. Some integrable systems on subspaces of Hilbert-Schmidt operators were also introduced in \cite{DO11}. There, the coinduction method suggested in \cite{OR08} was used to construct
Banach Lie--Poisson spaces obtained from the ideal of real Hilbert-Schmidt operators, and Hamiltonian systems related to the $k$-diagonal Toda lattice were presented. 
Last but not least, the relation between the Bruhat-Poisson structure on the restricted Grassmannian constructed in the present paper and the Poisson--Lie group of Pseudo-Differential symbols considered in \cite{KZ95} in relation to the Korteweg-de Vries hierarchy needs further study, and the link with the Poisson--Lie Grassmannian introduced in \cite{Z94} has to be clarified.

The present paper just approaches some aspects of the theory of Banach Poisson--Lie groups, and a more systematic study of the  infinite-dimensional theory would be interesting. It is written to be as self-contained as possible, and we hope that our exposition enables functional-analysts, geometers and physicists to read it. However the notions of Banach manifold and fiber bundles over Banach manifolds will not be recalled and we refer the readers to \cite{La01} for more introductory exposition. 

The paper is organized as follows. 
Part~1 is devoted to the general theory of Banach Poisson--Lie groups and related structures.
The exposition goes in the opposite direction of the usual exposition in the finite-dimensional setting, where the notion of finite-dimensional Poisson--Lie groups is introduced first, followed by the notion of Lie bialgebra (which is the structure that a Lie algebra of a Poisson--Lie group inherits), and at last the notion of Manin triples. Here we start with the notion of Banach Manin triples, since it is a notion of linear algebra that is easy to adapted to the Banach context, and which provides a good entry point into the theory of (Banach) Poisson--Lie groups.  This point of view allows us to introduce little by little notation and notions that are fundamental for the paper~:  the notion of duality pairing is recalled in  Section~1.1, the notion of coadjoint action on bounded multilinear maps on subspaces of the dual is defined in Section~2.3, and the notion of $1$-cocycles on a Banach Lie group or a Banach Lie algebra is  explained in Section~2.5. Generalized Banach Poisson manifolds are defined in Section~3.1. In Section~3.2 we show that weak symplectic Banach manifolds are examples of generalized Banach Poisson manifolds. In Section~3.3, we adapt the notion of Banach Lie--Poisson spaces introduced in \cite{OR03} to the case of an arbitrary duality pairing between two Banach Lie algebras, and show that they are generalized Banach Poisson manifolds (Theorem~\ref{4.2}). The notion of Banach Lie bialgebras is introduced in Section~4, and its relation to the notion of Banach Manin triples is given by the following Theorem~:
\begin{theorem}[Theorem \ref{bialgebra_to_manin}]
Consider two Banach Lie algebras 
 $\left(\mathfrak{g}_+, [\cdot,\cdot]_{\mathfrak{g}_+}\right)$ and  $\left(\mathfrak{g}_-, [\cdot,\cdot]_{\mathfrak{g}_-}\right)$ in duality.  Denote by $\mathfrak{g}$ the Banach space $\mathfrak{g} = \mathfrak{g}_+\oplus \mathfrak{g}_-$ with norm $\|\cdot\|_{\mathfrak{g}} = \|\cdot\|_{\mathfrak{g}_+}+\|\cdot\|_{\mathfrak{g}_-}$. The following assertions are equivalent.
\begin{itemize}
\item[(1)] $\mathfrak{g}_+$ is a Banach Lie--Poisson space and a Banach Lie bialgebra with respect to $\mathfrak{g}_-$; 
\item[(2)] $(\mathfrak{g}, \mathfrak{g}_+, \mathfrak{g}_-)$ is a Banach Manin triple for the non-degenerate symmetric bilinear map given by
$$
\begin{array}{lcll}
\langle\cdot,\cdot\rangle_{\mathfrak{g}}~: & \mathfrak{g}\times\mathfrak{g}&\rightarrow & \mathbb{K}\\
& (x,\alpha)\times (y, \beta) & \mapsto & \langle x, \beta\rangle_{\mathfrak{g}_+, \mathfrak{g}_-} + \langle y, \alpha\rangle_{\mathfrak{g}_+, \mathfrak{g}_-}.
\end{array}
$$
\end{itemize}
\end{theorem}
\hspace{-.45cm}Finally Section~5 is devoted to the notion of Banach Poisson--Lie groups. Basic examples are given in Section~5.3. In Section~5.4, we prove that the Lie algebra $\mathfrak{g}$ of a Banach Poisson--Lie group $(G, \mathbb{F}, \pi)$ carries a natural structure of Banach Lie bialgebra with respect to $\mathbb{F}_e$, and, with an additional condition on the Poisson tensor, is a Banach Lie--Poisson space with respect to $\mathbb{F}_e$.

The generalized notion of Banach Poisson manifolds introduced in Part~1 is adapted to the particular examples of Poisson--Lie groups we present in Part~2. Examples of Banach Poisson--Lie group in our sense include the restricted unitary group $\operatorname{U}_{\res}(\mathcal{H})$ and the restricted triangular group $\operatorname{B}_{\res}^+(\mathcal{H})$, which are modelled on non-reflexive Banach spaces (see Section~7.3). 
In Section~8, we show that the restricted Grassmannian viewed as homogeneous space under $\operatorname{U}_{\res}(\mathcal{H})$ inherits a Poisson structure in analogy to the finite-dimensional picture developped in \cite{LW90} and called Bruhat-Poisson structure. Moreover, the natural action of the Poisson--Lie group $\operatorname{B}_{\res}^+(\mathcal{H})$ on the restricted Grassmannian  is a Poisson map, and its orbits are the Schubert cells described in \cite{PS88}. 
These results are summarized in the following Theorem (see Theorem \ref{Poisson}, Theorem \ref{actionB}, and Theorem \ref{cells}).
\begin{theorem}
The restricted Grassmannian $$\operatorname{Gr}_{\res}(\mathcal{H}) =  \operatorname{U}_{\res}(\mathcal{H})/\operatorname{U}(\mathcal{H}_+)\times \operatorname{U}(\mathcal{H}_-) = \operatorname{GL}_{\res}(\mathcal{H})/\operatorname{P}_{\res}(\mathcal{H})$$ carries a natural Poisson structure such that~:
\begin{enumerate}
\item the canonical projection $p~: \operatorname{U}_{\res}(\mathcal{H}) \rightarrow \operatorname{Gr}_{\res}(\mathcal{H})$ is a Poisson map,
\item the natural action of $\operatorname{U}_{\res}(\mathcal{H})$ on $\operatorname{Gr}_{\res}(\mathcal{H})$ by left translations is a Poisson map,
\item the following right action of $\operatorname{B}_{\res}^{+}(\mathcal{H})$ on $\operatorname{Gr}_{\res}(\mathcal{H}) = \operatorname{GL}_{\res}(\mathcal{H})/\operatorname{P}_{\res}(\mathcal{H})$ is a  Poisson map~:
$$
\begin{array}{cll}
 \operatorname{Gr}_{\res}(\mathcal{H})\times \operatorname{B}_{\res}^{+}(\mathcal{H})& \rightarrow & \operatorname{Gr}_{\res}(\mathcal{H})\\
(g \operatorname{P}_{\res}(\mathcal{H}), b) & \mapsto & (b^{-1}g) \operatorname{P}_{\res}(\mathcal{H}).
\end{array}
$$
\item the symplectic leaves of $\operatorname{Gr}_{\res}(\mathcal{H})$ are the Schubert cells and are the orbits of $\operatorname{B}_{\res}^{+}(\mathcal{H})$.
\end{enumerate}
\end{theorem}
\hspace{-.45cm}Let us mention that the infinite-dimensional abelian subgroup of $\operatorname{B}_{\res}^+(\mathcal{H})$ generated by the shift induces the KdV hierarchy as explained in \cite{SW85}. 
\newpage

\part{Banach Poisson--Lie groups and related structures}

This Part is devoted to the general theory of Banach Poisson--Lie groups that will be needed in Part~2.  Examples of Banach Poisson--Lie groups are given in Section~5.3. The Banach Lie bialgebra struture of the Lie algebra of a Banach Poisson--Lie group is constructed in Section~5.4.

\section{Manin triples in the infinite-dimensional setting}

We start in this Section with the easiest notion related to Poisson--Lie groups, namely the notion of Manin triples. It will allow us to set up some notation used in the present paper, and recall the notion of duality pairing, which is crucial for the following Sections. The unboundedness of the triangular truncation on the space of trace class operators and on the space of bounded operators (see Section~\ref{Triangular truncations of operators}) will have important consequences in Section~\ref{no_Manin}. Examples of Banach Manin triples coming from Iwasawa decompositions are given in Section~\ref{Iwasawap}. In particular, the Manin triple $\left(L_2(\mathcal{H}), \mathfrak{u}_2(\mathcal{H}), \mathfrak{b}_2^+(\mathcal{H})\right)$ of Hilbert-Schmidt operators will have a key r\^ole in the proofs of most Theorems in Part~2.

\subsection{Duality pairings of Banach spaces}
In this paper, we will consider real or complex Banach spaces, and we will denote by $\mathbb{K} \in \{\mathbb{R}, \mathbb{C}\}$ the scalar field.
The dual $\mathfrak{g}^*$ of a Banach space $\mathfrak{g}$ will mean the continuous dual, i.e. the Banach space of bounded linear forms with values in $\mathbb{K}$. In a lot of applications, the dual of a Banach space $\mathfrak{g}$ is to big to work with, and one uses proper subspaces of $\mathfrak{g}^*$. 
A duality pairing between two Banach spaces allows to identify one Banach space with a subspace of the dual of the other. Additional structures on one of the Banach spaces (like a Lie bracket for instance) give rise to additional structures on the other Banach space via duality.
\subsubsection{Definition of strong and weak duality pairings}
Let us recall the notion of duality pairing in the infinite-dimensional setting (see \cite{AMR88}, supplement 2.4.C). 
\begin{definition}
Let $\mathfrak{g}_{1}$ and $\mathfrak{g}_{2}$ be two normed vector spaces over the same field $\mathbb{K}\in\{\mathbb{R}, \mathbb{C}\}$, and let $$\langle\cdot,\cdot\rangle_{\mathfrak{g}_1, \mathfrak{g}_2}~:\mathfrak{g}_1\times\mathfrak{g}_2\rightarrow \mathbb{K}$$ be a continuous bilinear map. One says that the map $\langle\cdot,\cdot\rangle_{\mathfrak{g}_1, \mathfrak{g}_2}$ is a \textbf{duality pairing} between $\mathfrak{g}_1$ and $\mathfrak{g}_2$ if and only if it is \textbf{non-degenerate}, i.e. if the following two conditions hold~:
$$\left(\langle x, y \rangle_{\mathfrak{g}_1, \mathfrak{g}_2} = 0,~~ \forall x\in\mathfrak{g}_1 \right) \Rightarrow y = 0\quad \textrm{and}\quad
\left(\langle x, y \rangle_{\mathfrak{g}_1, \mathfrak{g}_2} = 0,~~ \forall y\in\mathfrak{g}_2\right) \Rightarrow x = 0.$$
\end{definition}
\begin{definition}
A duality pairing $\langle\cdot,\cdot\rangle_{\mathfrak{g}_1, \mathfrak{g}_2}$ is a  \textbf{strong duality pairing} between $\mathfrak{g}_1$ and $\mathfrak{g}_2$ if and only if the two continuous linear maps 
\begin{equation}\label{duality_pairing_strong}
\begin{array}{lrl} \mathfrak{g}_1& \longrightarrow & \mathfrak{g}_2^*\\ x & \longmapsto &  \langle x,\cdot\rangle_{\mathfrak{g}_1, \mathfrak{g}_2}\end{array}
\quad
\textrm{and}\quad
\begin{array}{lrl} \mathfrak{g}_2& \longrightarrow & \mathfrak{g}_1^*\\ y & \longmapsto &  \langle \cdot,y\rangle_{\mathfrak{g}_1, \mathfrak{g}_2}\end{array}
\end{equation}
are isomorphisms. In all other cases, the duality pairing is called \textbf{weak}.
\end{definition}
The non-degenerate condition of a duality pairing implies that the maps \eqref{duality_pairing_strong} are injective. In other words, the existence of a duality pairing between $\mathfrak{g}_1$ and $\mathfrak{g}_2$ allows to identify $\mathfrak{g}_1$ with a subspace (not necessary closed!) of the continuous dual $\mathfrak{g}_2^*$ of $\mathfrak{g}_2$, and $\mathfrak{g}_2$ with a subspace of $\mathfrak{g}_1^*$, wheras a strong duality pairing gives isomorphisms 
$\mathfrak{g}_1\simeq \mathfrak{g}_2^*$ and $\mathfrak{g}_{2}\simeq \mathfrak{g}_1^*$. Therefore the existence of a strong duality pairing between $\mathfrak g_1$ and $\mathfrak g_2$ implies that $\mathfrak g_1$ and $\mathfrak g_2$ are reflexive Banach spaces. Note that in the finite-dimensional case, a count of the dimensions shows that any duality pairing is a strong duality pairing.

\begin{remark}{\rm
By Hahn-Banach Theorem, the natural pairing between a Banach space $\mathfrak{g}$ and its continuous dual $\mathfrak{g}^*$ is a duality pairing.  It is a strong duality pairing in the reflexive case $\mathfrak{g}^{**} = \mathfrak{g}$. 
}
\end{remark}

\subsubsection{Notation and Examples}
In order to give examples of duality pairings, let us introduce some notation used in the present paper.  
The letter $\mathcal{H}$ will refer to a general complex separable infinite-dimensional Hilbert space. The inner product in $\mathcal{H}$ will be denoted by $\langle \cdot | \cdot \rangle~:\mathcal{H}\times\mathcal{H}\rightarrow \mathcal{H}$ and will be complex-linear in the second variable, and conjugate-linear in the first variable.

\subsubsection*{Banach algebra $\operatorname{L}_{\infty}(\mathcal H)$ of bounded operators over a Hilbert space $\mathcal{H}$} 
 Denote by $\operatorname{L}_{\infty}(\mathcal H)$ the space of bounded linear maps from $\mathcal{H}$ into itself. It is a Banach space for the norm of operators 
$\|A\|_\infty :=  \sup_{\|x\| \leq 1} \|A x\| $ and a Banach Lie algebra for the bracket given by the commutator of operators ~:
$[A, B]  = A\circ B - B\circ A$, for $A$, $B \in L_\infty(\mathcal{H})$. In the following, we will denote the composition $A\circ B$  of the operators $A$ and $B$ simply by $A B$.

\subsubsection*{Hilbert algebra $\operatorname{L}_2(\mathcal{H})$ of Hilbert-Schmidt operators} 
A bounded operator $A$ admits an adjoint $A^*$ which is the bounded linear operator defined by $\langle A^*x| y\rangle = \langle x| Ay\rangle$. A positive operator is a bounded operator such that $\langle \varphi| A \varphi\rangle\geq 0$ for any $\varphi\in\mathcal{H}$. By polarization, if $A$ is positive then $A^* = A$. The trace of a positive operator $A$ is defined as $$\Tr A := \sum_{n= 1}^{+\infty} \langle \varphi_n| A \varphi_n\rangle\in [0, +\infty],$$ where $\varphi_n$ is any orthonormal basis of $\mathcal{H}$ (the right hand side does not depend on the choice of orthonormal basis, see Theorem~2.1 in \cite{Sim79}). The Schatten class $\operatorname{L}_2(\mathcal{H})$ of Hilbert-Schmidt operators is the subspace of $\operatorname{L}_{\infty}(\mathcal H)$ consisting of bounded operators $A$ such that  $\|A\|_2 := \left(\Tr (A^*A)\right)^{\frac{1}{2}}$ is finite. It is a Banach Lie algebra  for $\|\cdot \|_2$ and for the bracket given by the commutator of operators. It is also an ideal of $\operatorname{L}_{\infty}(\mathcal H)$ in the sense that for any $A\in \operatorname{L}_2(\mathcal{H})$ and any $B\in \operatorname{L}_{\infty}(\mathcal{H})$, one has $A B \in \operatorname{L}_2(\mathcal{H})$ and $B A \in \operatorname{L}_2(\mathcal{H})$.

\subsubsection*{Banach algebra $\operatorname{L}_{1}(\mathcal{H})$ of trace-class operators} 
For  a bounded linear operator $A$, the square root of $A^*A$ is well defined, and denoted by $(A^*A)^{\frac{1}{2}}$ (see Theorem VI.9 in \cite{RS1}).
The Schatten class $\operatorname{L}_{1}(\mathcal{H})$ of trace class operators is the subspace of $\operatorname{L}_{\infty}(\mathcal H)$ consisting of bounded operators $A$ such that  $\|A\|_1 := \Tr (A^*A)^{\frac{1}{2}}$ is finite.  It is a Banach Lie algebra  for $\|\cdot \|_1$ and for the bracket given by the commutator of operators.  We recall that  for any $A \in \operatorname{L}_{1}(\mathcal{H})$ (not necessarly positive),  the trace of $A$ is defined as
$$
\Tr A := \sum_{n=1}^{\infty} \langle \varphi_n| A \varphi_n\rangle,
$$
where $\{\varphi_n\}$ is any orthonormal basis of $\mathcal{H}$ (the right hand side does not depend on the orthonormal basis, see Theorem~3.1 in \cite{Sim79}) and that we have
$$
|\Tr A | \leq \|A\|_1.
$$ 
 Moreover  $\operatorname{L}_{1}(\mathcal{H})$ is an ideal of $\operatorname{L}_{\infty}(\mathcal H)$, i.e. for any $A \in \operatorname{L}_{1}(\mathcal{H})$ and any $B \in \operatorname{L}_{\infty}(\mathcal{H})$, $A B \in \operatorname{L}_{1}(\mathcal{H})$ and $B A \in \operatorname{L}_{1}(\mathcal{H})$, and furthermore $\Tr AB = \Tr BA$. Finally for $A$ and $B$ in $\operatorname{L}_{2}(\mathcal{H})$, one has $A B \in \operatorname{L}_{1}(\mathcal{H})$, $B A \in \operatorname{L}_1(\mathcal{H})$, and $\Tr AB = \Tr BA$ (see Corollary 3.8 in \cite{Sim79}).
 
 \subsubsection*{Banach algebras $\operatorname{L}_{p}(\mathcal{H})$} 
For any $1<p<\infty$, the Schatten class $\operatorname{L}_{p}(\mathcal{H})$ is the subspace of $\operatorname{L}_{\infty}(\mathcal H)$ consisting of bounded operators $A$ such that  $$\|A\|_p := \left(\Tr (A^*A)^{\frac{p}{2}}\right)^{\frac{1}{p}}$$ is finite. It is a Banach algebra for the norm $\|\cdot\|_p$ and for the bracket given by the commutator of operators. Moreover $\operatorname{L}_{p}(\mathcal{H})$ is an ideal of $\operatorname{L}_{\infty}(\mathcal H)$~: for any $A \in \operatorname{L}_{p}(\mathcal{H})$ and any $B \in \operatorname{L}_{\infty}(\mathcal{H})$, $A B \in \operatorname{L}_{p}(\mathcal{H})$ and $B A \in \operatorname{L}_{p}(\mathcal{H})$.

\begin{remark}{\rm
For $1<p<2<q<\infty$, one has 
$$
\operatorname{L}_1(\mathcal{H})\hookrightarrow \operatorname{L}_p(\mathcal{H})\hookrightarrow \operatorname{L}_2(\mathcal{H})\hookrightarrow \operatorname{L}_q(\mathcal{H})\hookrightarrow \operatorname{L}_\infty(\mathcal{H}),
$$
where each injection is a continuous map between Banach spaces. In the following, we will repeatedly use these inclusions.
}
\end{remark}

Let us now give some examples of duality pairings.
\begin{example}\label{ex_duality}
{\rm 
The trace of the product of two operators $(A, B)\mapsto \Tr AB$ is a strong duality pairing between $\operatorname{L}_2(\mathcal{H})$ and itself.}
\end{example}
\begin{example}\label{exDualityL1L2}{\rm
Since $\operatorname{L}_1(\mathcal{H})$ is a dense subspace of $\operatorname{L}_2(\mathcal{H})$, one obtains a weak duality pairing between $\operatorname{L}_1(\mathcal{H})$ and  $\operatorname{L}_2(\mathcal{H})$ by considering the bilinear map  $(A, B)\mapsto \Tr AB$ with $A\in \operatorname{L}_1(\mathcal{H})$ and $B\in \operatorname{L}_{2}(\mathcal H)$.
}
\end{example}
\begin{example}{\rm
 Since the dual of $\operatorname{L}_1(\mathcal{H})$ can be identified with $\operatorname{L}_{\infty}(\mathcal H)$ using the trace, one has  a weak duality pairing between $\operatorname{L}_1(\mathcal{H})$ and $\operatorname{L}_{\infty}(\mathcal H)$ by considering the bilinear map $(A, B)\mapsto \Tr AB$ with $A\in \operatorname{L}_1(\mathcal{H})$ and $B\in \operatorname{L}_{\infty}(\mathcal H)$.  Note that the dual space of $\operatorname{L}_{\infty}(\mathcal H)$ stricktly contains $\operatorname{L}_1(\mathcal{H})$ as a closed subspace.
 }
 \end{example}
 \begin{example}\label{ex_duality2}{\rm
 For $1<p<\infty$, define $1<q<\infty$ by the relation $\frac{1}{p}+\frac{1}{q}=1$. For any $A \in \operatorname{L}_{p}(\mathcal{H})$ and any $B \in \operatorname{L}_{q}(\mathcal{H})$, $A B \in \operatorname{L}_{1}(\mathcal{H})$ and $B A \in \operatorname{L}_{1}(\mathcal{H})$ with
$$
\|AB\|_1\leq \|A\|_p\|B\|_q \quad\quad\textrm{  and  }\quad\quad \|BA\|_1\leq \|A\|_p\|B\|_q,
$$
(see Proposition~5, page~41 in \cite{RS2}) and furthermore $\Tr AB = \Tr BA$. Moreover the trace of the product of two operators $(A, B)\mapsto \Tr AB$ is a strong duality pairing between $\operatorname{L}_p(\mathcal{H})$ and $\operatorname{L}_q(\mathcal{H})$ and gives rise to the following identifications (see Proposition~7, page~43 in \cite{RS2} and Theorem~VI.26, page~212 in \cite{RS1})~:
$$
\left(\operatorname{L}_p(\mathcal{H})\right)^* \simeq \operatorname{L}_q(\mathcal{H})\quad\quad\textrm{ and }\quad\quad \left(\operatorname{L}_q(\mathcal{H})\right)^* \simeq \operatorname{L}_p(\mathcal{H})
$$
}
\end{example}

\subsection{Duals and injection of Banach spaces}
Suppose that $\mathfrak{h}$ is a Banach space that injects continuously into another Banach space $\mathfrak{g}$, i.e. one has a continuous injection $\iota~:\mathfrak{h}\hookrightarrow \mathfrak{g}$. Then one can consider two different dual spaces :  the dual space $\mathfrak{h}^*$ which consists of  linear forms on the Banach space $\mathfrak{h}$ which are continuous with respect to the operator norm associated to the Banach norm $\|\cdot\|_{\mathfrak{h}}$ on $\mathfrak{h}$, and the norm dual $\iota(\mathfrak{h})^*$ of the subspace $\iota(\mathfrak{h})\subset\mathfrak{g}$ endowed with the norm $\|\cdot\|_{\mathfrak{g}}$ of $\mathfrak{g}$, consisting of continuous linear forms on the normed vector space $(\iota(\mathfrak{h}), \|\cdot\|_{\mathfrak{g}})$. Note that, since $\mathbb{R}$ is complete, $\iota(\mathfrak{h})^*$ is complete even if $\iota(\mathfrak{h})$ is not closed in $\mathfrak{g}$ (see for instance \cite{Bre10} section~1.1). Let us compare these two duals : $\mathfrak{h}^*$ on one hand and $\iota(\mathfrak{h})^*$ on the other hand.
First note that one gets a well-defined map
$$
\begin{array}{llll}
\iota^*~:& \mathfrak{g}^*& \rightarrow &\mathfrak{h}^*\\
& f & \mapsto & f\circ \iota
\end{array}
$$
since $f\circ \iota $ is continuous for the operator norm induced by the norm of $\mathfrak{h}$ whenever $f$ is continuous for the operator norm induced by the norm on $\mathfrak{g}$.
Note that $\iota^*$ is surjective if and only if any continuous form on $\mathfrak{h}$ can be extended to a continuous form on $\mathfrak{g}$. On the other hand, $\iota^*$ is injective if and only if the only continuous form on $\mathfrak{g}$ that vanishes on $\iota(\mathfrak{h})$ is the zero form. 


 Suppose that the range of $\iota~:\mathfrak{h}\hookrightarrow \mathfrak{g}$ is closed. Then $\iota(\mathfrak{h})$ endowed with the norm of $\mathfrak{g}$ is a Banach space. It follows that $\iota$ is a continuous bijection from the Banach space $\mathfrak{h}$ onto the Banach space $\iota(\mathfrak{h})$, therefore by the open mapping theorem, it is an isomorphism of Banach spaces  (see for instance Corollary~2.7 in \cite{Bre10}). In this case, any continuous form on $\mathfrak{h}$ is continuous for the norm of $\mathfrak{g}$ i.e. one has $\mathfrak{h}^* = \iota(\mathfrak{h})^*$. By Hahn-Banach theorem, any continuous form on $\iota(\mathfrak{h})$ can be extended to a continuous form on $\mathfrak{g}$ with the same norm (see Corollary 1.2 in \cite{Bre10}). Therefore the dual map $\iota^*~:\mathfrak{g}^* \rightarrow \mathfrak{h}^*$ is surjective. Its kernel is the annihilator $\iota(\mathfrak{h})^0$ of $\iota(\mathfrak{h})$  and $\mathfrak{h}^*$ is isomorphic to the quotient space $\mathfrak{g}^*/\iota(\mathfrak{h})^0$. 

\begin{example}{\rm
The injection of the Banach space of compact operators $\mathcal{K}(\mathcal{H})$ on a separable Hilbert space $\mathcal{H}$ into the Banach space of bounded operators $\operatorname{L}_\infty(\mathcal{H})$ is closed. The dual map $\iota^*~: \operatorname{L}_\infty(\mathcal{H})^* \rightarrow \mathcal{K}(\mathcal{H})^* $ is surjective and $\mathcal{K}(\mathcal{H})^* $ can be identified with the space $\operatorname{L}_1(\mathcal{H})$ of trace class operators on $\mathcal{H}$ using the trace. Therefore $\operatorname{L}_1(\mathcal{H})$ is isomorphic to the quotient space $\operatorname{L}_\infty(\mathcal{H})^*/\mathcal{K}(\mathcal{H})^0$. 
}
\end{example}

 Suppose now  that the range of $\iota~:\mathfrak{h}\hookrightarrow \mathfrak{g}$ is dense in $\mathfrak{g}$. In this case, any continuous form on $\iota(\mathfrak{h})$ extends in a unique way to a continuous form on $\mathfrak{g}$ with the same norm i.e. $\iota(\mathfrak{h})^* = \mathfrak{g}^*$. The kernel of $\iota^*$ consists of continuous maps on $\mathfrak{g}$ that vanish on the dense subspace $\iota(\mathfrak{h})$, hence is reduced to $0$. In other words $\iota^*~:\mathfrak{g}^*\rightarrow \mathfrak{h}^*$ is injective (see also Corollary 1.8 in \cite{Bre10}).

\begin{example}{\rm
Consider the inclusion $\iota~:\operatorname{L}_1(\mathcal{H})\hookrightarrow \operatorname{L}_2(\mathcal{H})$ of the space of trace-class operators into the space of Hilbert-Schmidt operators on $\mathcal{H}$. Then the range of $\iota$ is dense. This leads to the injection $\iota^*~: \operatorname{L}_2(\mathcal{H})^*= \operatorname{L}_2(\mathcal{H}) \hookrightarrow \operatorname{L}_1(\mathcal{H})^* = \operatorname{L}_{\infty}(\mathcal{H})$.}
\end{example}

\subsection{Definition of Banach Manin triples}
The notion of Manin triple is a notion of linear algebra that can be adapted in a straightforward way to the Banach context.
\begin{definition}
A Banach Manin triple consists of a triple of Banach Lie algebras $(\mathfrak{g}, \mathfrak{g}_+, \mathfrak{g}_-)$ over a field $\mathbb{K}$ and a \textbf{non-degenerate symmetric bilinear} continuous map $\langle\cdot, \cdot\rangle_{\!\mathfrak{g}}$ on $\mathfrak{g}$ such that 
\begin{enumerate}
\item the bilinear map $\langle\cdot,\cdot\rangle_{\!\mathfrak{g}}$ is invariant with respect to the bracket $[\cdot, \cdot]_{\mathfrak{g}}$ of $\mathfrak{g}$, i.e.
\begin{equation}\label{invariance_bilinear_map}
\langle [x, y]_{\mathfrak{g}}, z\rangle_{\!\mathfrak{g}}  + \langle y, [x, z]_{\mathfrak{g}}\rangle_{\!\mathfrak{g}} = 0,~~\forall x, y, z \in \mathfrak{g};
\end{equation}
\item $\mathfrak{g} = \mathfrak{g}_+\oplus\mathfrak{g}_-$ as Banach spaces;
\item both $\mathfrak{g}_+$ and $\mathfrak{g}_-$ are Banach Lie subalgebras of $\mathfrak{g}$;
\item both $\mathfrak{g}_+$ and $\mathfrak{g}_-$ are isotropic with respect to the bilinear map $\langle\cdot,\cdot\rangle_{\!\mathfrak{g}}$.
\end{enumerate}
\end{definition}
\hspace{-0.45cm}Note that in the Banach context, it is important to ask for the continuity of the bilinear map $\langle\cdot, \cdot\rangle_{\!\mathfrak{g}}$, as well as for a decomposition $\mathfrak{g} = \mathfrak{g}_+\oplus\mathfrak{g}_-$ of $\mathfrak{g}$ into the sum of two closed Banach subspaces. Let us make some remarks which are simple consequences of the definition of a Manin triple.

\begin{remark}
{\rm 
 Given a Manin triple $(\mathfrak{g}, \mathfrak{g}_+, \mathfrak{g}_-)$, condition (2) implies that any continuous linear form $\alpha$ on $\mathfrak{g}$ decomposes in a continuous way as 
$$
\alpha = \alpha\circ p_{\mathfrak{g}_+} + \alpha\circ p_{\mathfrak{g}_-},
$$
where $p_{\mathfrak{g}_+}$ (resp. $p_{\mathfrak{g}_-}$) is the continuous projection onto $\mathfrak{g}_+$ (resp.  $\mathfrak{g}_-$) with respect to the decomposition $\mathfrak{g} = \mathfrak{g}_+\oplus\mathfrak{g}_-$. In other words, one has a decomposition of the continuous dual $\mathfrak{g}^*$ of $\mathfrak{g}$ as
$$
\mathfrak{g}^* = \mathfrak{g}_-^0 \oplus \mathfrak{g}_+^0,
$$
where $\mathfrak{g}_\pm^0$ is the annihilator of $\mathfrak{g}_{\pm}$, i.e.
$$
\mathfrak{g}_\pm^0 := \{\alpha\in\mathfrak{g}^*~: \alpha(x) = 0,~~~\forall x \in \mathfrak{g}_{\pm}\}.
$$
Moreover any continuous linear form $\beta$ on $\mathfrak{g}_+$ can be extended in a unique way to a continuous linear form on $\mathfrak{g}$ belonging to $\mathfrak{g}_-^0$ 
by $\beta\mapsto \beta\circ p_{+}$. It follows that one has an isomorphism 
$$
\mathfrak{g}_+^* \simeq \mathfrak{g}_-^0,$$ and similarly
$$
\mathfrak{g}_-^* \simeq \mathfrak{g}_+^0.
$$
}
\end{remark}
\begin{remark}{\rm
Given a Manin triple $(\mathfrak{g}, \mathfrak{g}_+, \mathfrak{g}_-)$ where $\langle\cdot,\cdot\rangle_{\mathfrak{g}}$ is a \textbf{strong} duality pairing, any continuous linear form on $\mathfrak{g}$ can be written as $\langle x, \cdot\rangle_{\mathfrak{g}}$ for some $x\in\mathfrak{g}$. 
In particular, for any subspace $\mathfrak{h}\subset\mathfrak{g}$, one has
$$
\mathfrak{h}^0 \simeq \mathfrak{h}^{\perp},
$$
where
$$
\mathfrak{h}^{\perp} := \{ x\in\mathfrak{g}~:\langle x, y\rangle_{\mathfrak{g}} = 0, ~~\forall y \in \mathfrak{h}\}.
$$
Moreover, any continuous linear form $\beta$ on $\mathfrak{g}_+$ can be represented as $\beta(x) = \langle x, y\rangle_{\mathfrak{g}}$ for a unique element $y\in \mathfrak{g}_-$. Therefore, in this case,
$$
\mathfrak{g}_- \simeq \mathfrak{g}^*_+
$$
and similarly
$$
\mathfrak{g}_+ \simeq \mathfrak{g}^*_-.
$$
}
\end{remark}

\subsection{Triangular truncations of operators}\label{Triangular truncations of operators}
Endow the separable complex Hilbert space $\mathcal{H}$ with an orthonormal basis $\{|n\rangle\}_{n\in\mathbb{Z}}$ ordered according to decreasing values of $n$.  For $1\leq p \leq \infty$, consider the following Banach Lie subalgebras of $\operatorname{L}_p(\mathcal{H})$
\begin{equation}\label{triangularLp}
\begin{array}{l}
\operatorname{L}_p(\mathcal{H})_{-} := \{x \in \operatorname{L}_p(\mathcal{H})~: x |n\rangle \in \textrm{span}\{|m\rangle, m\leq n\}\}\\   \quad \quad\quad \quad\quad \quad \textrm{(lower triangular operators)}\\ \\ 
\operatorname{L}_p(\mathcal{H})_{++} := \{x \in \operatorname{L}_p(\mathcal{H})~: x |n\rangle \in \textrm{span}\{|m\rangle, m> n\}\}\\  \quad \quad \quad \quad\quad\quad  \textrm{(strictly upper triangular operators)}.
\end{array}
\end{equation}
and
\begin{equation}\label{triangularLp+}
\begin{array}{l}
\operatorname{L}_{p}(\mathcal{H})_{+} := \{\alpha\in \operatorname{L}_{p}(\mathcal{H})~: \alpha |n\rangle \in \textrm{span}\{|m\rangle, m \geq n\}\}\\  \quad \quad \quad \quad\quad\quad\textrm{(upper triangular operators)}\\ \\
\operatorname{L}_{p}(\mathcal{H})_{--} := \{\alpha\in \operatorname{L}_{p}(\mathcal{H})~: \alpha |n\rangle \in \textrm{span}\{|m\rangle, m<n\}\}\\  \quad \quad \quad \quad\quad  \quad\textrm{(strictly lower triangular operators)}.
\end{array}
\end{equation}
The linear transformation $T_-$  consisting in taking the lower triangular part of an operator with respect to the orthonormal basis $\{|n\rangle\}_{n\in\mathbb{Z}}$ of $\mathcal{H}$  is called a triangular truncation or triangular projection  (see \cite{A78}) and is defined as follows~:
\begin{equation}\label{T-}
\langle m | T_-(A) n\rangle := \left\{\begin{array}{cc} \langle m | A n\rangle &  \quad \textrm{if} \quad m\leq n\\ 0 & \quad\textrm{if}\quad m> n\end{array}\right.
\end{equation}
Similarly, the linear transformation $T_{++}$ consisting in taking the stricktly upper triangular part of an operator with respect to $\{|n\rangle\}_{n\in\mathbb{Z}}$ is defined as follows~:
\begin{equation}\label{T++}
\langle m | T_{++}(A) n\rangle := \left\{\begin{array}{cc} \langle m | A n\rangle &  \quad \textrm{if} \quad m> n\\ 0 & \quad\textrm{if}\quad m\leq n\end{array}\right.
\end{equation}
The linear transformation $D$ consisting in taking the diagonal part of a linear operator is defined by
\begin{equation}\label{diagonal}
\langle m | D(A) n\rangle := \left\{\begin{array}{cc} \langle m | A n\rangle &  \quad \textrm{if} \quad n=m\\ 0 & \quad\textrm{if}\quad n\neq m\end{array}\right.
\end{equation}
\begin{remark}\label{prendre_triangle}{\rm
The triangular truncations $T_-$ and $T_{++}$ are unbounded on $ \operatorname{L}_{\infty}(\mathcal{H})$ and on $ \operatorname{L}_{1}(\mathcal{H})$, but are bounded on $ \operatorname{L}_{p}(\mathcal{H})$ for $1<p<\infty$  (see \cite{M61}, \cite{KP70}, \cite{GK70} as well as Proposition~4.2 in \cite{A78} for the proof and more detail on the subject).
%
%
See also \cite{D88} for an example of bounded operator whose triangular truncation is unbounded (Hilbert matrix). 
As far as we know the existence and construction of a trace class operator whose triangular projection is not trace class is an open problem.
We refer the reader to \cite{Bel11} for related functional-analytic issues in the theory of Banach Lie groups.
}
\end{remark}
Denote by $T_+ = T_{++} + D$ (resp. $T_{--} = T_{-} - D$) the linear transformation consisting in taking the upper triangular part (resp. strictly lower triangular part) of an operator. One has for $1<p<\infty$,
\begin{equation}\label{lppm}
\operatorname{L}_{p}(\mathcal{H}) = \operatorname{L}_{p}(\mathcal{H})_{+}\oplus \operatorname{L}_{p}(\mathcal{H})_{--},
\end{equation} 
and
\begin{equation}\label{lpmp}
\operatorname{L}_{p}(\mathcal{H}) = \operatorname{L}_{p}(\mathcal{H})_{-}\oplus \operatorname{L}_{p}(\mathcal{H})_{++}.
\end{equation} 


%
%

%
%
\subsection{Example of Iwasawa Manin triples}\label{Iwasawap}

The Iwasawa decomposition of a finite-dimensional semi-simple Lie group is a generalization of the decomposition of $GL(n, \mathbb{C})$ as the product of $SU(n)\times A\times N$, where $A$ is the abelian group of diagonal matrices with positive real coefficients, and $N$ is the group of triangular matrices whose diagonal entries are all equal to $1$. The product $A\times N$ is often denoted by $B$ for Borel subgroup. At the level of Lie algebras,  the Iwasawa decomposition gives rise to the decomposition $M(n, \mathbb{C}) = \mathfrak{u}(n) \oplus \mathfrak{b}(n)$, where $\mathfrak{b}(n)$ is the Lie algebra of complex triangular matrices with real coefficients on the diagonal. Since the triangular truncation defined in Section~\ref{Triangular truncations of operators} is bounded on $L^p(\mathcal{H})$ for $1<p<\infty$, we can generalize this decomposition to the Banach context (see Lemma~\ref{decub}). As explained in \cite{LW90}, $\left(M(n, \mathbb{C}), \mathfrak{u}(n), \mathfrak{b}(n)\right)$ is an example of Manin triple, where the duality pairing is given by the imaginary part of the trace.  This duality pairing can be defined on $L^p(\mathcal{H})$ for $1<p\leq 2$ because in this case $L^p(\mathcal{H})$ injects into its dual. This gives rise to Banach Manin triples, that we will call Iwasawa Manin triples (see Proposition~\ref{triples} below).

We will use the following notation.
The real Banach Lie algebra $\mathfrak{u}_{p}(\mathcal{H})$ is the Lie algebra of skew-Hermitian operators in $ \operatorname{L}_p(\mathcal{H})$~:
\begin{equation}\label{u2}
\mathfrak{u}_p(\mathcal{H}) := \{A\in  \operatorname{L}_{p}(\mathcal H)~: A^* = -A\}.\quad\quad\quad\quad\quad\quad\quad\quad\quad\quad\quad\quad\quad\quad\quad\quad\quad\quad\quad
\end{equation} 
The real Banach subalgebras $\mathfrak{b}^+_p(\mathcal{H})$ and $\mathfrak{b}_p^-(\mathcal{H})$  of $\operatorname{L}_p(\mathcal{H})$ are the triangular Banach algebras defined as follows~: 
\begin{equation}\label{b2pm}
\begin{array}{l}
\mathfrak{b}^+_p(\mathcal{H}) := \{\alpha\in \operatorname{L}_p(\mathcal{H})~: \alpha |n\rangle \in~ \textrm{span}\{|m\rangle, m\geq n\}~\textrm{and}~\langle n|\alpha|n\rangle\in\mathbb{R}, \textrm{for}~ n\in\mathbb{Z}\},
\\
\mathfrak{b}^-_p(\mathcal{H}) := \{\alpha\in \operatorname{L}_p(\mathcal{H})~: \alpha |n\rangle \in~ \textrm{span}\{|m\rangle, m\leq n\}~\textrm{and}~\langle n|\alpha|n\rangle\in\mathbb{R}, \textrm{for}~ n\in\mathbb{Z}\}.
\end{array}
\end{equation}

\begin{lemma}\label{decub}
For $1<p<\infty$, one has the following direct sum decompositions of $\operatorname{L}_p(\mathcal{H})$ into the sum of closed subalgebras
\begin{equation}\label{b+}
\operatorname{L}_p(\mathcal{H}) = \mathfrak{u}_p(\mathcal{H})\oplus\mathfrak{b}^+_p(\mathcal{H}),
\end{equation}
and
\begin{equation}\label{b-}
\operatorname{L}_p(\mathcal{H}) = \mathfrak{u}_p(\mathcal{H})\oplus\mathfrak{b}^-_p(\mathcal{H}).
\end{equation}
The projection $p_{\mathfrak{u}_p,+}$  onto $\mathfrak{u}_p(\mathcal{H})$ with respect to the decomposition \eqref{b+} reads
\begin{equation}\label{projectionu+}
p_{\mathfrak{u}_{p},+}(A) = T_{--}(A)- T_{--}(A)^* +\frac{1}{2}\left[D(A) - D(A)^*\right], \quad\textrm{where}\quad A\in \operatorname{L}_p(\mathcal{H}).
\end{equation}
%
Similarly, the projection $p_{\mathfrak{u}_p,-}$   onto $\mathfrak{u}_p(\mathcal{H})$ with respect to the decomposition \eqref{b-} reads~:
\begin{equation}\label{projectionu-}
p_{\mathfrak{u}_{p},-}(A) = T_{++}(A)- T_{++}(A)^* +\frac{1}{2}\left[D(A) - D(A)^*\right], \quad\textrm{where}\quad  A\in \operatorname{L}_p(\mathcal{H}).
\end{equation}

\end{lemma}

\begin{proof}
Since the triangular truncations $T_+~: \operatorname{L}_p(\mathcal{H})\rightarrow \operatorname{L}_p(\mathcal{H})$ and $T_{++}~: \operatorname{L}_p(\mathcal{H})\rightarrow \operatorname{L}_p(\mathcal{H})$ are bounded for $1<p<\infty$ (see remark~\ref{prendre_triangle}), the same is true for the operator $D = T_+-T_{++}$.  The Lemma follows as in the finite-dimensional case.
\end{proof}

\begin{proposition}\label{triples}
For $1<p\leq 2$, the triples of Banach Lie algebras $(\operatorname{L}_p(\mathcal{H}), \mathfrak{u}_{p}(\mathcal{H}), \mathfrak{b}^+_p(\mathcal{H}))$ and 
$(\operatorname{L}_p(\mathcal{H}), \mathfrak{u}_{p}(\mathcal{H}), \mathfrak{b}^-_p(\mathcal{H}))$ are real Banach Manin triples with respect to the pairing given by the imaginary part of the trace
\begin{equation}\label{imparttrace}
\begin{array}{lcll}
\langle\cdot,\cdot\rangle_{\mathbb{R}}~: &\operatorname{L}_p(\mathcal{H})\times \operatorname{L}_p(\mathcal{H})& \longrightarrow &\mathbb{R}\\
& (x, y) & \longmapsto & \Im \Tr\left(x y\right).
\end{array}
\end{equation}

\end{proposition}

\begin{proof}
\begin{itemize}
\item Let us show that the bilinear form on  $\operatorname{L}_p(\mathcal{H})$ given by the imaginary part of the trace is invariant with respect to the bracket given by the commutator.  Set $q:= \frac{p}{p-1}$. Then $1<p\leq2\leq q<\infty$. For any $x, y, z \in \operatorname{L}_p(\mathcal{H})$, recall that $\operatorname{L}_p(\mathcal{H})\cdot \operatorname{L}_p(\mathcal{H})\subset \operatorname{L}_p(\mathcal{H})$, $\operatorname{L}_p(\mathcal{H})\subset L_q(\mathcal{H})$, and $\operatorname{L}_p(\mathcal{H})\cdot L_q(\mathcal{H})\subset L_1(\mathcal{H})$.  Therefore one has
$$
\begin{array}{ll}
\Tr ([x, y] z) &= \Tr (x y z - y x z) = \Tr (x y z) - \Tr (y x z) \\&= \Tr (y z x) - \Tr (y x z) = -\Tr y [x, z],\end{array}
$$
where the second equality follows from the fact that both $x y z$ and $y x z$ are in $\operatorname{L}_1(\mathcal{H})$, and the third is justified since $y z$ belongs to $\operatorname{L}_1(\mathcal{H})$ and $x$ is bounded. Taking the imaginary part of the trace preserves this invariance.
Hence $\langle\cdot,\cdot\rangle_{\mathbb{R}}$ is invariant with respect to the Lie bracket of $\operatorname{L}_p(\mathcal{H})$.

\item By Lemma~\ref{decub}, one has the direct sum decompositions $$\operatorname{L}_p(\mathcal{H}) = \mathfrak{u}_p(\mathcal{H})\oplus\mathfrak{b}^\pm_p(\mathcal{H}).$$ 

\item Note that $\langle\cdot,\cdot\rangle_{\mathbb{R}}$ is well-defined because $L_p(\mathcal{H})\subset L_q(\mathcal{H})$ for $1<p\leq 2$. It is clearly symmetric and continuous. Let us show that $\langle\cdot,\cdot\rangle_{\mathbb{R}}$ is a non-degenerate bilinear form on $\operatorname{L}_p(\mathcal{H})$.
Denote by $\mathcal{H}_{\mathbb{R}}$ the real Hilbert space generated by $\{|n\rangle\}_{n\in\mathbb{Z}}$. Any bounded linear operator $A$ on the complex Hilbert space $\mathcal{H} = \mathcal{H}_{\mathbb{R}} + i \mathcal{H}_{\mathbb{R}}$ 
 can be written in blocks as
$$
A = \left(\begin{array}{cc} \Re A & -\Im A\\ \Im A & \Re A\end{array} \right).
$$
where $\Re A~:\mathcal{H}_{\mathbb{R}}\rightarrow \mathcal{H}_{\mathbb{R}}$ and $\Im A~:\mathcal{H}_{\mathbb{R}}\rightarrow i\mathcal{H}_{\mathbb{R}}$. In particular, $A\in L_p(\mathcal{H})$ is the $\mathbb{C}$-linear extension of $\Re A + \textrm{i} \Im A$ (note that this is not the decomposition of $A$ into its symmetric and skew-symmetric parts).
Therefore, for any $A, B \in \operatorname{L}_p(\mathcal{H})$,
$$
 \Im \Tr\left(AB\right) = \Tr\left(\Re A \Im B + \Im A \Re B\right).
$$
Suppose that $ \Im \Tr\left(AB\right) = 0$ for any $B\in \operatorname{L}_p(\mathcal{H})$. Since $\operatorname{L}_p(\mathcal{H})$ is dense in $\operatorname{L}_2(\mathcal{H})$, 
this implies that $\Tr \Re A\cdot C = 0$ for any operator $C\in \operatorname{L}_2(\mathcal{H}_\mathbb{R})$, and $\Tr \Im A \cdot D = 0$ for any $D\in \operatorname{L}_2(\mathcal{H}_{\mathbb{R}})$. It follows that $\Re A = 0$ and $\Im A = 0$ because the trace is a strong duality pairing between $ \operatorname{L}_2(\mathcal{H}_{\mathbb{R}})$ and itself.
%
\item It is easy to show that $\mathfrak{u}_p(\mathcal{H})\subset \left(\mathfrak{u}_{p}(\mathcal{H})\right)^\perp$, $\mathfrak{b}_p^+(\mathcal{H})\subset \left(\mathfrak{b}_p^+(\mathcal{H})\right)^\perp$ and 
$\mathfrak{b}_p^-(\mathcal{H})\subset \left(\mathfrak{b}_p^-(\mathcal{H})\right)^\perp$, in other words $\mathfrak{u}_{p}(\mathcal{H})$, $\mathfrak{b}_p^+(\mathcal{H})$ and $\mathfrak{b}_p^-(\mathcal{H})$ are isotropic subspaces with respect to the pairing $\langle\cdot,\cdot\rangle_{\mathbb{R}}$.
\end{itemize}
\end{proof}

\begin{remark}{\rm
In the previous Proposition, the condition $1<p\leq 2$ is necessary in order to define the trace of the product of two elements in $\operatorname{L}_p(\mathcal{H})$ ($\operatorname{L}_p(\mathcal{H})$ is contained in its dual $L_q(\mathcal{H})$ for $1<p\leq 2$).
}
\end{remark}

%


\section{From Manin triples to $1$-cocycles}
The existence of a Lie bracket on a Banach space $\mathfrak{g}_+$ has consequences on any Banach space $\mathfrak{g}_-$ in duality with $\mathfrak{g}_+$. Under some stability and continuity conditions (see Section~\ref{Coadjoint action on a subspace of the dual}), $\mathfrak{g}_+$ will act on $\mathfrak{g}_-$ by coadjoint action, as well as  on the space of bounded multilinear maps on $\mathfrak{g}_-$ (see Section~\ref{adrep}). When $\mathfrak{g}_+$ and $\mathfrak{g}_-$ form a Banach Manin triple, 
a natural  $1$-cocycle with respect to the action of $\mathfrak{g}_+$ on the space of skew-symmetric bilinear maps on $\mathfrak{g}_-$ can be defined  (see Section~\ref{Manin_cocycle}).

\subsection{Adjoint and coadjoint actions}
Recall that a Banach Lie algebra $\mathfrak{g}_+$ acts on itself, its continuous dual $\mathfrak{g}_+^*$ and bidual $\mathfrak{g}_+^{**}$   by the adjoint and coadjoint actions~:
$$
\begin{array}{llll}
\ad~: &\mathfrak{g}_+\times\mathfrak{g}_+&\longrightarrow &\mathfrak{g}_+\\
& (x, y) & \longmapsto & \ad_x y := [x , y ],
\end{array}
$$
$$
\begin{array}{llll}
-\ad^*~: &\mathfrak{g}_+\times\mathfrak{g}_+^*&\longrightarrow &\mathfrak{g}_+^*\\
& (x, \alpha) & \longmapsto & -\ad^*_x \alpha := -\alpha\circ \ad_x,
\end{array}
$$
and
$$
\begin{array}{llll}
\ad^{**}~: &\mathfrak{g}_+\times\mathfrak{g}_+^{**}&\longrightarrow &\mathfrak{g}_+^{**}\\
& (x, \mathcal{F}) & \longmapsto & \ad^{**}_x \mathcal{F} := \mathcal{F}\circ \ad^*_x.
\end{array}
$$
Here the notation $\ad^*_x~:\mathfrak{g}_+^*\rightarrow\mathfrak{g}_+^*$ means the dual map of $\ad_x~:\mathfrak{g}_+\rightarrow \mathfrak{g}_+$. Remark that the actions $\ad$ and $\ad^{**}$ coincide on the subspace $\mathfrak{g}_+$ of $\mathfrak{g}_+^{**}$. These actions extend in a natural way to spaces of bounded multilinear maps from any Banach product of copies of $\mathfrak{g}_+$  and $\mathfrak{g}_+^*$. 
For Banach spaces  $\mathfrak{g}_1, \dots, \mathfrak{g}_k$ and $\mathfrak{h}$, we will use the notation $\operatorname{L}(\mathfrak{g}_1,\mathfrak{g}_2, \dots\mathfrak{g}_k;\mathfrak{h})$ to denote the \textbf{Banach space of continuous $k$-multilinear maps} from the product Banach space $\mathfrak{g}_1\times\dots\times\mathfrak{g}_k$ to the Banach space $\mathfrak{h}$ (note the semi-colon separating the initial Banach spaces from the final one). 
Let us recall (see Proposition 2.2.9 in \cite{AMR88}) that one has the following  isometric isomorphisms of Banach spaces
\begin{equation}\label{isomorphism_dual}
\operatorname{L}(\mathfrak{g}_+^*; \operatorname{L}(\mathfrak{g}_+,\mathfrak{g}_+;\mathbb{K}))\simeq \operatorname{L}(\mathfrak{g}_+^*,\mathfrak{g}_+,\mathfrak{g}_+; \mathbb{K}) \simeq \operatorname{L}(\mathfrak{g}_+,\mathfrak{g}_+^*; \operatorname{L}(\mathfrak{g}_+;\mathbb{K}))  \simeq \operatorname{L}(\mathfrak{g}_+,\mathfrak{g}_+^*;\mathfrak{g}_+^*).
\end{equation}
In particular, since the map $\ad~:\mathfrak{g}_+\times\mathfrak{g}_+\rightarrow \mathfrak{g}_+$ is bilinear and continuous, its dual map $\ad^*$  is continuous as a map from $\mathfrak{g}^*_+$ to $\operatorname{L}(\mathfrak{g}_+,\mathfrak{g}_+;\mathbb{K})$ 
and, following the sequence of isomorphisms in \eqref{isomorphism_dual}, it follows that $\ad^*~:\mathfrak{g}_+\times\mathfrak{g}^*_+\rightarrow\mathfrak{g}_+^*$ is continuous.
Similarly, using the following isometric isomorphisms of Banach spaces
$$
\operatorname{L}(\mathfrak{g}_+^{**}; \operatorname{L}(\mathfrak{g}_+,\mathfrak{g}_+^*;\mathbb{K})) \simeq \operatorname{L}(\mathfrak{g}_+^{**},\mathfrak{g}_+, \mathfrak{g}_+^*; \mathbb{K}) \simeq \operatorname{L}(\mathfrak{g}_+,\mathfrak{g}_+^{**}; \operatorname{L}(\mathfrak{g}_+^*;\mathbb{K}))  \simeq \operatorname{L}(\mathfrak{g}_+,\mathfrak{g}_+^{**};\mathfrak{g}_+^{**}),
$$
it follows that $\ad^{**}~:\mathfrak{g}_+\times\mathfrak{g}_+^{**}\rightarrow \mathfrak{g}_+^{**}$ is continuous.

\subsection{Coadjoint action on a subspace of the dual}\label{Coadjoint action on a subspace of the dual}
Suppose that we have a continuous injection from a Banach space $\mathfrak{g}_-$ into the dual space $\mathfrak{g}_+^*$ of a Banach Lie algebra $\mathfrak{g}_+$, in such a way that $\mathfrak{g}_-$ is stable by the coadjoint action of $\mathfrak{g}_+$ on its dual, i.e. is such that
\begin{equation}\label{stability_coadjoint}
\ad_{x}^* \alpha \in \mathfrak{g}_-, ~~\forall x \in \mathfrak{g}_+, \forall \alpha\in \mathfrak{g}_-.
\end{equation}
Then the coadjoint action $-\ad^*~:\mathfrak{g}_+\times\mathfrak{g}_+^*\rightarrow\mathfrak{g}_+^*$ restricts to a continuous bilinear map $-\ad^*_{|\mathfrak{g}_-}~:\mathfrak{g}_+\times\mathfrak{g}_-\rightarrow\mathfrak{g}_+^*$, where $\mathfrak{g}_+\times\mathfrak{g}_-$ is endowed with the Banach structure of the product of Banach spaces $\mathfrak{g}_+$ and $\mathfrak{g}_-$. In other words
$$-\ad^*_{|\mathfrak{g}_-}\in L(\mathfrak{g}_+,\mathfrak{g}_-;\mathfrak{g}_+^*) \simeq L(\mathfrak{g}_+;L(\mathfrak{g}_-;\mathfrak{g}_+^*)).
$$
Moreover, condition \eqref{stability_coadjoint} implies that $-\ad^*$ takes values in $\mathfrak{g}_-$, i.e. that one gets a well-defined action
$$
\begin{array}{llll}
-\ad^*_{|\mathfrak{g}_-}~: &\mathfrak{g}_+\times\mathfrak{g}_-&\longrightarrow &\mathfrak{g}_-\\
& (x, \alpha) & \longmapsto & -\ad^*_x \alpha := -\alpha\circ \ad_x.
\end{array}
$$
However, this action will in general not be continuous if one endows the target space with its Banach space topology. Nevertheless it is continuous if the target space is equipped with the topology induced from $\mathfrak{g}_+^*$. Under the additional assumption that $-\ad^{*}_{|\mathfrak{g}_-}~:\mathfrak{g}_+\times\mathfrak{g}_-\rightarrow \mathfrak{g}_-$ is continuous with respect to the Banach space topologies of $\mathfrak{g}_+$ and $\mathfrak{g}_-$ (for instance in the case where $\mathfrak{g}_-$ is a closed subspace of the dual $\mathfrak{g}_+^*$), $\mathfrak{g}_+$ acts also continuously on $\mathfrak{g}_-^*$ by 
$$
\begin{array}{llll}
(\ad^{*}_{|\mathfrak{g}_-})^*~: &\mathfrak{g}_+\times\mathfrak{g}_-^{*}&\longrightarrow &\mathfrak{g}_-^{*}\\
& (x, \mathcal{F}) & \longmapsto & \mathcal{F}\circ \ad^*_x.
\end{array}
$$

\subsection{Adjoint action on the space of continuous multilinear  maps}\label{adrep}
Suppose that we have a continuous injection from a Banach space $\mathfrak{g}_-$ into the dual space $\mathfrak{g}_+^*$ of a Banach Lie algebra $\mathfrak{g}_+$ and that $\mathfrak{g}_+$ acts continously on $\mathfrak{g}_-$ by coadjoint action, i.e. suppose that $-\ad^*_{|\mathfrak{g}_-}$ takes values in $\mathfrak{g}_-$ and that $-\ad^*_{|\mathfrak{g}_-}~:\mathfrak{g}_+\times \mathfrak{g}_-\rightarrow\mathfrak{g}_-$ is continuous. In order to simplify notation, we will write just $\ad^*$ for $\ad^*_{|\mathfrak{g}_-}$ and $\ad^{**}$ for $(\ad^{*}_{|\mathfrak{g}_-})^*$. In order to compactify notations, let us 
denote by $L^{r,s}(\mathfrak{g}_-, \mathfrak{g}_+; \mathbb{K})$ the \textbf{Banach space of continuous multilinear maps} from $\mathfrak{g}_-\times\dots\times\mathfrak{g}_-\times\mathfrak{g}_+\times\dots \times \mathfrak{g}_+$ to $\mathbb{K}$, where $\mathfrak{g}_-$ is repeated \textbf{$r$-times} and $\mathfrak{g}_+$ is repeated \textbf{$s$-times}. Since $\mathfrak{g}_+$ acts continuously by adjoint action on itself and by coadjoint action on $\mathfrak{g}_-$, one can define a continuous linear action of $\mathfrak{g}_+$ on  $L^{r,s}(\mathfrak{g}_-, \mathfrak{g}_+; \mathbb{K})$, also called adjoint action,  by
\begin{align*}
 \ad^{(r, s)}_x \textbf{t} (\alpha_1, \dots, \alpha_r, x_1, \dots, x_s)= \sum_{i =1}^{r} \textbf{t}(\alpha_1, \dots, \ad^*_x \alpha_i, \dots, \alpha_r, x_1, \dots, x_s)\\ - \sum_{i=1}^{s} \textbf{t}(\alpha_1, \dots,  \alpha_r, x_1, \dots, \ad_x x_i, \dots x_s),
\end{align*}
where $\textbf{t}\in L^{r,s}(\mathfrak{g}_-, \mathfrak{g}_+; \mathbb{K})$, for $i\in\{1, \dots, r\}$, $\alpha_i\in \mathfrak{g}_-$, and for $i \in\{1, \dots, s\}$, $x_i\in \mathfrak{g}_+$. 
In particular, the adjoint action of $\mathfrak{g}_+$ on $L^{2,0}(\mathfrak{g}_-, \mathfrak{g}_+; \mathbb{K}) := L(\mathfrak{g}_-, \mathfrak{g}_-;\mathbb{K})$ reads~:
\begin{equation}\label{ad2}
\ad^{(2,0)}_x \textbf{t}(\alpha_1, \alpha_2) = \textbf{t}(\ad^*_x\alpha_1, \alpha_2) + \textbf{t}(\alpha_1, \ad^*_x \alpha_2).
\end{equation}

\subsection{Subspaces of skew-symmetric bilinear maps}\label{lambda_g_+}
Note that the adjoint action $\ad^{(2,0)}$ defined in \eqref{ad2} preserves the subspace of skew-symmetric continuous bilinear maps on $\mathfrak{g}_-$, denoted by $\Lambda^2\mathfrak{g}_-^{*}$~:
$$\Lambda^2\mathfrak{g}_-^{*} := \left\{\textbf{t}\in L(\mathfrak{g}_-, \mathfrak{g}_-;\mathbb{K})~:~\forall e_1, e_2\in \mathfrak{g}_-, \textbf{t}(e_1, e_2) = -\textbf{t}(e_2, e_1)\right\}.$$
For any subspace $\mathfrak{g}_+\subset \mathfrak{g}_-^*$, the subspace  $\Lambda^2\mathfrak{g}_+\subset\Lambda^2\mathfrak{g}_-^{*}$ refers to the subspace consisting of elements $\textbf{t}\in \Lambda^2\mathfrak{g}_-^{*}$ such that, for $\alpha \in \mathfrak{g}_-$, the maps $\alpha\mapsto \textbf{t}(e_1, \alpha)$ belong to  $\mathfrak{g}_+\subset\mathfrak{g}_-^{*}$ for any $e_1\in\mathfrak{g}_-$. 
$$
\Lambda^2\mathfrak{g}_+ := \left\{\textbf{t}\in \Lambda^2\mathfrak{g}_-^{*}~:~\forall e_1\in \mathfrak{g}_-,  \textbf{t}(e_1, \cdot) \in \mathfrak{g}_+\right\}.
$$

\subsection{Definition of 1-Cocycles} \label{cocycle_section}
Let us recall the notion of $1$-cocycle. Let $G$ be a Banach Lie group, and consider an affine action of $G$ on a Banach space $V$, i.e. a group morphism $\Phi$ of $G$ into the Affine group $\textrm{Aff}(V)$ of transformations of $V$. Using the isomorphism $\textrm{Aff}(V) = \textrm{GL}(V) \rtimes V$, $\Phi$ decomposes into $(\varphi, \Theta)$ where $\varphi~: G\rightarrow \textrm{GL}(V)$ and $\Theta~: G \rightarrow V$. The condition that $\Phi$ is a group morphism implies that $\varphi$ is a group morphism and that $\Theta$ satisfies~:
\begin{equation}\label{group_cocycle}
\Theta(gh) = \Theta(g) + \varphi(g)(\Theta(h)),
\end{equation}
where $g,h\in G$.
One says that $\Theta$ is a \textbf{$1$-cocycle on $G$ relative to $\varphi$}. 
The derivative $d\Phi$ of $\Phi$ at the unit element of $G$ is a Lie algebra morphism of the Lie algebra $\mathfrak{g}$ of $G$ into the Lie algebra $\mathfrak{aff}(V)$ of $\textrm{Aff}(V)$.  By the isomorphism $\mathfrak{aff}(V) = \mathfrak{gl}(V)\rtimes V$, $d\Phi$ decomposes into $(d\varphi, d\Theta)$ where $d\varphi~:\mathfrak{g}\rightarrow \mathfrak{gl}(V)$ is the Lie algebra morphism induced by $\varphi$ and $d\Theta~:\mathfrak{g}\rightarrow V$ satisfies~:
\begin{equation}\label{ptitcocycle}
d\Theta\left([x, y]\right) = d\varphi(x)\left(d\Theta(y)\right) - d\varphi(y)\left(d\Theta(x)\right),
\end{equation}
for $x, y \in \mathfrak{g}$. 
One says that $d\Theta$ is a \textbf{$1$-cocycle on $\mathfrak{g}$ relative to $d\varphi$}.

\begin{example}{\rm
Let us consider in particular the Banach space $V = L(\mathfrak{g}_-,\mathfrak{g}_-;\mathbb{K})$ of bilinear maps on $\mathfrak{g}_-$, where $\mathfrak{g}_-$ is a Banach space that injects continuously in the dual space $\mathfrak{g}_+^*$ of a Banach Lie algebra $\mathfrak{g}_+$, is stable under the coadjoint action of $\mathfrak{g}_+$, and such that the coadjoint action of $\mathfrak{g}_+$ on $\mathfrak{g}_-$ is continuous.  A $1$-cocycle $\theta$ on $\mathfrak{g}_+$ relative to the natural action $\ad^{(2,0)}$ of $\mathfrak{g}_+$ on $L(\mathfrak{g}_-, \mathfrak{g}_-;\mathbb{K})$ given by \eqref{ad2} is a map $\theta~:\mathfrak{g}_+\rightarrow L(\mathfrak{g}_-, \mathfrak{g}_-;\mathbb{K})$ which satisfies~:
\begin{equation*}\begin{array}{ll}
\theta\left([x, y]\right) & = \ad^{(2,0)}_x\left(\theta(y)\right) - \ad^{(2,0)}_y\left(\theta(x)\right)\\
\end{array}
\end{equation*}
where $x, y\in \mathfrak{g}_+$. For $\alpha$ and $\beta$ in $\mathfrak{g}_-$, previous condition reads
\begin{equation}\label{cocycle2}
\theta\left([x, y]\right)(\alpha, \beta)  = \theta(y)(\ad^*_x\alpha, \beta) + \theta(y)(\alpha, \ad^*_x\beta) -\theta(x)(\ad^*_y\alpha, \beta) - \theta(x)(\alpha, \ad^*_y\beta).
\end{equation}
}
\end{example}

\begin{remark}{\rm
A continuous map $\theta~:\mathfrak{g}_+\rightarrow L(\mathfrak{g}_-, \mathfrak{g}_-;\mathbb{K})$ from a Banach Lie algebra $\mathfrak{g}_+$ to the Banach space of bilinear maps on $\mathfrak{g}_-$ satisfying equation~\eqref{cocycle2} defines an affine action of  $\mathfrak{g}_+$  on $L(\mathfrak{g}_-,\mathfrak{g}_-;\mathbb{K})$ whose linear part is the adjoint action $\ad^{(2,0)}$ given by equation \eqref{ad2}.
}
\end{remark}

%

\subsection{Manin triples and associated $1$-cocycles}\label{Manin_cocycle}

The following proposition enable to define $1$-cocycles naturally associated to a Manin triple.

\begin{theorem}\label{Manin_to_bialgebra}
Let $(\mathfrak{g}, \mathfrak{g}_+, \mathfrak{g}_-)$ be a Manin triple for a non-degenerate symmetric bilinear continuous map $\langle\cdot, \cdot\rangle_{\mathfrak{g}}~:\mathfrak{g}\times\mathfrak{g}\rightarrow \mathbb{K}$. Then
\begin{enumerate}
\item The map $\langle\cdot, \cdot\rangle_{\mathfrak{g}}$ restricts to a duality pairing $\langle\cdot, \cdot\rangle_{\mathfrak{g}_+, \mathfrak{g}_-}~:\mathfrak{g}_+\times\mathfrak{g}_-\rightarrow \mathbb{K}$.

\item The subspace $\mathfrak{g}_+\hookrightarrow \mathfrak{g}_-^*$ is stable under the coadjoint action of $\mathfrak{g}_-$ on $\mathfrak{g}_-^*$ and
$$\ad^*_\alpha(x) = -p_{\mathfrak{g}_+}\left([\alpha, x]_{\mathfrak{g}}\right)$$  for any $x\in\mathfrak{g}_+$ and $\alpha\in\mathfrak{g}_-$. In particular, the map 
\begin{equation*}
\begin{array}{llll}
\ad^*_{\mathfrak{g}_-}~:&\mathfrak{g}_-\times\mathfrak{g}_+&\rightarrow& \mathfrak{g}_+\\
& (\alpha, x)&\mapsto & -p_{\mathfrak{g}_+}\left([\alpha, x]_{\mathfrak{g}}\right)
\end{array}
\end{equation*}
is continuous.

\item The subspace $\mathfrak{g}_-\hookrightarrow \mathfrak{g}_+^*$ is stable under the coadjoint action of $\mathfrak{g}_+$ on $\mathfrak{g}_+^*$ 
and  $$\ad^*_x(\alpha) = -p_{\mathfrak{g}_-}\left([x, \alpha]_{\mathfrak{g}}\right)$$ for any $x\in\mathfrak{g}_+$ and $\alpha\in\mathfrak{g}_-$. In particular, the map
\begin{equation*}
\begin{array}{llll}
\ad^*_{\mathfrak{g}_+}~:&\mathfrak{g}_+\times\mathfrak{g}_-&\rightarrow& \mathfrak{g}_-\\
& (x, \alpha)&\mapsto & -p_{\mathfrak{g}_-}\left([x, \alpha]_{\mathfrak{g}}\right)
\end{array}
\end{equation*}
is continuous.

\item The dual map to the bracket $[\cdot, \cdot]_{\mathfrak{g}_-}$ restricts to a $1$-cocycle $\theta_+~: \mathfrak{g}_+\rightarrow \Lambda^2\mathfrak{g}_+$ with respect to the adjoint action $\ad^{(2,0)}$ of $\mathfrak{g}_+$ on  $\Lambda^2\mathfrak{g}_+\subset \Lambda^2\mathfrak{g}_-^*$.

\item The dual map to the bracket $[\cdot, \cdot]_{\mathfrak{g}_+}$ restricts to a $1$-cocycle $\theta_-~: \mathfrak{g}_-\rightarrow \Lambda^2\mathfrak{g}_-$ with respect to the adjoint action $\ad^{(2,0)}$ of $\mathfrak{g}_-$ on  $\Lambda^2\mathfrak{g}_-\subset \Lambda^2\mathfrak{g}_+^*$.
\end{enumerate}

\end{theorem}

\begin{proof}
\begin{enumerate}
\item
Let us show that the restriction of the non-degenerate bilinear form $\langle\cdot,\cdot\rangle_{\mathfrak{g}}~:\mathfrak{g}\times\mathfrak{g}\rightarrow \mathbb{K}$ to $\mathfrak{g}_+\times\mathfrak{g}_-$ denoted by
$$\langle \cdot, \cdot \rangle_{\mathfrak{g}_+, \mathfrak{g}_-}~:\mathfrak{g}_+\times\mathfrak{g}_-\rightarrow \mathbb{K}$$
is a non-degenerate duality pairing between $\mathfrak{g}_+$ and $\mathfrak{g}_-$. Suppose that there exists $x\in\mathfrak{g}_+$ such that $\langle x, \alpha\rangle_{\mathfrak{g}_+, \mathfrak{g}_-} = 0$ for all $\alpha\in\mathfrak{g}_-$. Then, since $\mathfrak{g}_+$ is isotropic for  $\langle\cdot,\cdot\rangle_{\mathfrak{g}}$, one has  $\langle x, y\rangle_{\mathfrak{g}} = 0$ for all $y\in\mathfrak{g}$, and the non-degeneracy of 
$\langle\cdot,\cdot\rangle_{\mathfrak{g}}$ implies that $x=0$. The same argument apply interchanging $\mathfrak{g}_+$ and $\mathfrak{g}_-$, thus $\langle \cdot, \cdot \rangle_{\mathfrak{g}_+, \mathfrak{g}_-}$ is non-degenerate. As a consequence, one obtains two continuous injections 
$$
\begin{array}{lll}
\mathfrak{g}_-&\hookrightarrow & \mathfrak{g}^*_+\\
\alpha & \mapsto & \langle \cdot, \alpha \rangle_{\mathfrak{g}_+, \mathfrak{g}_-},
\end{array}
\quad
\textrm{and}\quad
\begin{array}{lll}
\mathfrak{g}_+&\hookrightarrow & \mathfrak{g}^*_-\\
x & \mapsto & \langle x, \cdot \rangle_{\mathfrak{g}_+, \mathfrak{g}_-}.
\end{array}
$$
\item[(2)-(3)]
Let us show that  both $$\mathfrak{g}_+\subset  \mathfrak{g}^*_-$$ and $$\mathfrak{g}_-\subset  \mathfrak{g}^*_+$$ are stable under the coadjoint action of $\mathfrak{g}_-$ on $\mathfrak{g}^*_-$ and $\mathfrak{g}_+$ on $\mathfrak{g}^*_+$ respectively. Indeed, the invariance of the bilinear form $\langle\cdot,\cdot\rangle_{\mathfrak{g}}$ with respect to the bracket $[\cdot,\cdot]_{\mathfrak{g}}$ implies
that for any $x\in\mathfrak{g}_+$ and $\alpha\in\mathfrak{g}_-$,
$$
\langle x, [\alpha, \cdot]_{\mathfrak{g}}\rangle_{\mathfrak{g}} = - \langle [\alpha, x]_{\mathfrak{g}},  \cdot\rangle_{\mathfrak{g}}.
$$
Hence, since $\mathfrak{g}_-$ is isotropic, 
$$
\langle x, [\alpha, \cdot]_{\mathfrak{g}}\rangle_{\mathfrak{g}_+, \mathfrak{g}_-} = - \langle p_{\mathfrak{g}_+}\left([\alpha, x]_{\mathfrak{g}}\right),  \cdot\rangle_{\mathfrak{g}_+, \mathfrak{g}_-}, 
$$
for any $x\in \mathfrak{g}_+$ and any $\alpha\in\mathfrak{g}_-$. It follows that
 $$\ad^*_\alpha(x) = -p_{\mathfrak{g}_+}\left([\alpha, x]_{\mathfrak{g}}\right)$$ and  similarly $$\ad^*_x(\alpha) = -p_{\mathfrak{g}_-}\left([x, \alpha]_{\mathfrak{g}}\right)$$ for any $x\in\mathfrak{g}_+$ and $\alpha\in\mathfrak{g}_-$.
The continuity of the corresponding adjoint maps follows from the continuity of the bracket $[\cdot, \cdot]_{\mathfrak{g}}$ and of the projections $p_{\mathfrak{g}_+}$ and $p_{\mathfrak{g}_-}$.

 \item[(4)-(5)] Let us prove that the dual map of the Lie bracket on $\mathfrak{g}_-$ restricts to a $1$-cocycle with respect to the adjoint action of $\mathfrak{g}_+$ on $\Lambda^2\mathfrak{g}_+$.
The dual map $$[\cdot, \cdot]_{\mathfrak{g}_-}^*~:  \mathfrak{g}_-^{*} \rightarrow  L(\mathfrak{g}_-, \mathfrak{g}_- ; \mathbb{K})  $$
to the bilinear map $[\cdot, \cdot]_{\mathfrak{g}_-}$ 
assigns to $\mathcal{F}(\cdot)\in \mathfrak{g}_-^{*}$ the bilinear form $\mathcal{F}\left([\cdot, \cdot]_{\mathfrak{g}_-}\right)$ and takes values in $\Lambda^2\mathfrak{g}_-^{*}$. 
 Since by (2), $\mathfrak{g}_-\subset  \mathfrak{g}^*_+$ is stable under the coadjoint action of $\mathfrak{g}_+$ and since the coadjoint action $\ad^*~:\mathfrak{g}_+\times\mathfrak{g}_-\rightarrow \mathfrak{g}_-$ is continuous, one can consider the adjoint action of $\mathfrak{g}_+$ on $\Lambda^2\mathfrak{g}_-^{*}$ defined by \eqref{ad2}.
Since the duality pairing $\langle\cdot, \cdot\rangle_{\mathfrak{g}_+, \mathfrak{g}_-}$ induces a continuous  injection $\mathfrak{g}_+\hookrightarrow \mathfrak{g}_-^*$, one can consider the subspace $\Lambda^2\mathfrak{g}_+$ of $\Lambda^2\mathfrak{g}_-^*$ defined in Section \ref{lambda_g_+}.
Denote by $\theta_+~:\mathfrak{g}_+ \rightarrow  L(\mathfrak{g}_-, \mathfrak{g}_- ; \mathbb{K})$ the restriction of $[\cdot, \cdot]_{\mathfrak{g}_-}^*$ to the subspace $\mathfrak{g}_+\subset \mathfrak{g}_-^{*}$~:
$$
\theta_+(x) =  \langle x, [\cdot, \cdot]_{\mathfrak{g}_-}\rangle_{\mathfrak{g}_+, \mathfrak{g}_-}.
$$
Using the identification $L(\mathfrak{g}_-, \mathfrak{g}_- ; \mathbb{K})  \simeq L(\mathfrak{g}_-; \mathfrak{g}_-^{*})$, one has 
$$
\theta_+(x)(\alpha) = \langle x, [\alpha, \cdot]_{\mathfrak{g}_-}\rangle_{\mathfrak{g}_+, \mathfrak{g}_-} = \ad^*_{\alpha}x(\cdot).
$$
One sees immediately that the map  $\theta_+$ takes values in $\Lambda^2\mathfrak{g}_+$ if and only if $\ad^*_{\alpha}x\in\mathfrak{g}_+$ for any $\alpha\in\mathfrak{g}_-$ and for any $x\in\mathfrak{g}_+$, which is verified by (2).
Using the fact that the duality pairing $\langle\cdot, \cdot \rangle_{\mathfrak{g}_+, \mathfrak{g}_-}$ is the restriction of $\langle\cdot, \cdot \rangle_{\mathfrak{g}}$ and that $\langle\cdot, \cdot \rangle_{\mathfrak{g}}$ is invariant with respect to the bracket $[\cdot, \cdot]_{\mathfrak{g}}$, one has  $$\langle [x, y], [\alpha, \beta]\rangle_{\mathfrak{g}_-, \mathfrak{g}_+} = -\langle [\alpha, [x,y]], \beta\rangle_{\mathfrak{g}},$$ and the Jacobi identity verified by $ [\cdot, \cdot]_{\mathfrak{g}}$ implies
$$
\langle [x, y], [\alpha, \beta]\rangle_{\mathfrak{g}_-, \mathfrak{g}_+} = -\langle [[\alpha, x], y], \beta\rangle_{\mathfrak{g}} - \langle [x, [\alpha, y]], \beta\rangle_{\mathfrak{g}}.
$$
Using the decomposition $$-[\alpha, x] = -p_{\mathfrak{g}_-}[\alpha, x] - p_{\mathfrak{g}_+}[\alpha, x] = -\ad^*_x \alpha + \ad^*_\alpha x,$$ and similarly
$$-[\alpha, y] = -p_{\mathfrak{g}_-}[\alpha, y] - p_{\mathfrak{g}_+}[\alpha, y]= -\ad^*_y \alpha + \ad^*_\alpha y,$$ one gets
\begin{equation*}
\quad\quad\langle [x, y], [\alpha, \beta]\rangle_{\mathfrak{g}_+, \mathfrak{g}_-} = \langle [\ad^*_\alpha x-\ad^*_x\alpha, y], \beta\rangle_{\mathfrak{g}} + \langle [x,  \ad^*_\alpha y -\ad^*_y\alpha], \beta\rangle_{\mathfrak{g}},
\end{equation*}
hence
\begin{equation}\label{weinstein}
\begin{array}{ll}
\langle [x, y], [\alpha, \beta]\rangle_{\mathfrak{g}_+, \mathfrak{g}_-}   &=  \langle [ \ad^*_\alpha x, y], \beta\rangle_{\mathfrak{g}}+ \langle [x, \ad^*_\alpha y], \beta\rangle_{\mathfrak{g}}  \\
&\quad+ \langle y, [\ad^*_x\alpha, \beta]\rangle_{\mathfrak{g}}  - \langle x, [\ad^*_y \alpha, \beta]\rangle_{\mathfrak{g}}.
\end{array}
\end{equation}
It follows that
\begin{equation*}
\ad^*_{\alpha}[x, y] = [\ad^*_\alpha x, y] + [x, \ad^*_\alpha y] + \ad^*_{\ad^*_x\alpha}y  - \ad^*_{\ad^*_y\alpha}x.
\end{equation*}
On the other hand, the condition \eqref{cocycle2} that $\theta_+$ is a $1$-cocycle reads~:
\begin{equation}\label{2.7}
\begin{array}{ll}
\langle [x, y], [\alpha, \beta]\rangle_{\mathfrak{g}_+, \mathfrak{g}_-}& = +\langle y, [\ad^*_x\alpha, \beta]\rangle_{\mathfrak{g}_+, \mathfrak{g}_-} +\langle y, [\alpha, \ad^*_x\beta]\rangle_{\mathfrak{g}_+, \mathfrak{g}_-} \\ & \quad-\langle x, [\ad^*_y\alpha, \beta]\rangle_{\mathfrak{g}_+, \mathfrak{g}_-} - \langle x, [\alpha, \ad^*_y\beta]\rangle_{\mathfrak{g}_+, \mathfrak{g}_-}.
\end{array}
\end{equation}
The first and third terms in the RHS of \eqref{2.7} equal the last two terms in the RHS of \eqref{weinstein}. Using the invariance ~\eqref{invariance_bilinear_map} of the bilinear form $\langle\cdot,\cdot\rangle_{\mathfrak{g}}$ with respect to the bracket $[\cdot, \cdot]_{\mathfrak{g}}$,  the last term in the RHS of \eqref{2.7} reads
\begin{equation*}
\begin{array}{ll}
- \langle x, [\alpha, \ad^*_y\beta]\rangle_{\mathfrak{g}_+, \mathfrak{g}_-}  & =\langle [\alpha, x], \ad^*_y\beta\rangle_{\mathfrak{g}} =  \langle p_{\mathfrak{g}_+}([\alpha, x]), \ad^*_y\beta\rangle_{\mathfrak{g}_+, \mathfrak{g}_-}\\ & =  - \langle \ad^*_\alpha x, \ad^*_y\beta\rangle_{\mathfrak{g}_+, \mathfrak{g}_-} = -\langle [y, \ad^*_\alpha x], \beta\rangle_{\mathfrak{g}_+, \mathfrak{g}_-},\end{array}
\end{equation*}
and similarly the second term in the RHS of \eqref{2.7} reads
\begin{equation*}
\langle y , [\alpha,  \ad^*_x\beta]\rangle_{\mathfrak{g}_+, \mathfrak{g}_-} = \langle [x, \ad^*_\alpha y], \beta\rangle_{\mathfrak{g}_+, \mathfrak{g}_-}.
\end{equation*}
Hence the equivalence between  \eqref{2.7} and  \eqref{weinstein} follows. By interchanging the roles of $\mathfrak{g_+}$ and $\mathfrak{g}_-$, one proves (5) in a similar way.
 \end{enumerate}
 \end{proof}

In the proof of Theorem~\ref{Manin_to_bialgebra}, we have showed the following~:
\begin{proposition}
Let $\mathfrak{g} = \mathfrak{g}_+\oplus\mathfrak{g}_-$  be a decomposition of a Banach Lie algebra $\mathfrak{g}$ into the direct sum of two Banach Lie subalgebras, and suppose that $\mathfrak{g}$ is endowed with a non-degenerate symmetric bilinear map $\langle\cdot,\cdot\rangle_{\mathfrak{g}}$, invariant with respect to the Lie bracket in $\mathfrak{g}$. Then the cocycle condition~\eqref{cocycle2}
for the restriction $\theta_+$ of $[\cdot, \cdot]_{\mathfrak{g}_-}^*~:\mathfrak{g}_-^*\rightarrow \Lambda^2\mathfrak{g}_-^*$ to the subspace $\mathfrak{g}_+\subset \mathfrak{g}_-^{*}$ reads
\begin{equation}\label{cocycle_weinstein}
\ad^*_{\alpha}[x, y] = [\ad^*_\alpha x, y] + [x, \ad^*_\alpha y] + \ad^*_{\ad^*_x\alpha}y  - \ad^*_{\ad^*_y\alpha}x,
\end{equation}
where $x, y \in\mathfrak{g}_+$ and $\alpha \in\mathfrak{g}_-$. 
\end{proposition}

\begin{remark}{\rm
Equation~\eqref{cocycle_weinstein} is exactly the formula given in \cite{LW90} page 507, but with the opposite sign convention for the coadjoint map $\ad^*$.
}
\end{remark}

\section{Generalized Banach Poisson manifolds and related notions}

In this Section, we generalize the definition of  Poisson manifolds to the Banach context (Section~\ref{def_Poisson}). Example of generalized Banach Poisson manifolds are Banach symplectic manifolds (Section~\ref{symp}) and Banach Lie--Poisson spaces (Section~\ref{symp}). 

\subsection{Definition of generalized Banach Poisson manifolds}\label{def_Poisson}

The notions of Banach Poisson manifolds and Banach Lie--Poisson spaces were introduced in \cite{OR03}. The notion of sub  Poisson structures in the Banach context was introduced in \cite{CP12}.  
 In the case of locally convex spaces, an analoguous definition of weak Poisson manifold structure was defined in \cite{NST14}.
In the symplectic case, related notions were introduced in \cite{DGR15} enabling the study of the orbital stability of some Hamiltonian PDE's. In the present paper, we restrict ourselves to the Banach setting but generalize slightly these notions to the case where an arbitrary duality pairing is considered, and where the existence of Hamiltonian vector fields is not assumed (this last point is assumed in \cite{NST14} and \cite{CP12}). Moreover, instead of working with subalgebras of the space of smooth functions $\mathcal{C}^{\infty}(M)$ on a Banach manifold $M$, we will work with subbundles of the cotangent bundle (see Remark~\ref{remark_Poisson} below). 



\begin{definition}
Consider a unital subalgebra $\mathcal{A}\subset\mathcal{C}^{\infty}(M)$ of smooth functions on a Banach manifold $M$, i.e. $\mathcal{A}$ is a vector subspace of $\mathcal{C}^{\infty}(M)$ containing the constants and  stable under pointwise multiplication.
An $\mathbb{R}$-bilinear operation $\{\cdot,\cdot\}~: \mathcal{A}\times\mathcal{A}
\rightarrow\mathcal{A}$ is called a \textbf{Poisson bracket} on $M$ if it satisfies~:
\begin{enumerate}
\item[(i)] anti-symmetry ~: $\{f, g\} = -\{g, f\}~;$
\item[(ii)] Jacobi identity~: $\{\{f, g\}, h\} + \{\{g, h\}, f\} + \{\{h, f\}, g\} =  0$~;
\item[(iii)] Leibniz formula~: $\{f, gh\} = \{f, g\}h + g\{f, h\}$~;
\end{enumerate}
\end{definition}
\begin{remark}\label{remark_Poisson}{\rm
\begin{enumerate}
\item
Note that the Leibniz rule implies that for any $f\in\mathcal{A}$, $\{f, \cdot\}$ acts by derivations on the subalgebra $\mathcal{A}\subset \mathcal{C}^{\infty}(M)$. 
When $M$ is finite-dimensional and $\mathcal{A} = \mathcal{C}^{\infty}(M)$, this condition implies that $\{f, \cdot\}$ is a smooth vector field $X_f$ on $M$, called the Hamiltonian vector field associated to $f$, uniquely defined by its action on $\mathcal{C}^{\infty}(M)$~:
$$
X_f(h) = dh(X_f) = \{f, h\}.
$$
It is worth noting that on an  infinite-dimensional Hilbert space, 
there exists derivations of order greater than 1, i.e. that do not depend only on the differentials of functions
(see Lemma 28.4 in \cite{KM97}, chapter VI).
%
%
 It follows that, contrary to the finite-dimensional case, one may not be able to associate a Poisson tensor (see Definition~\ref{Poisson_tensor} below) to a given Poisson bracket. Examples of Poisson brackets not given by Poisson tensors were constructed in \cite{BGT18}. 
\item Given a covector $\xi\in T^*_pM$, it is always possible to extend it to a locally defined $1$-form $\alpha$ with $\alpha_p = \xi$ (for instance by setting $\alpha$ equal to a constant in a chart around $p\in M$). However, it may not be possible to extend it to a smooth $1$-form on $M$. It may therefore not be possible to find a smooth real function on $M$ whose differential equals $\xi$ at $p\in M$. The difficulty resides in defining smooth bump functions, which are, in the finite dimensional Euclidean case, usually constructed using the differentiability of the norm. In \cite{R64}, it was shown that a Banach space admits a $\mathcal{C}^1$-norm away from the origin if and only if its dual is separable.
Remark that $L_\infty(\mathcal{H})$ is not separable (since it contains the nonseparable Banach space $l_{\infty}$ as the space of diagonal operators). It follows that the dual of $L_\infty(\mathcal{H})$ is nonseparable (since by Theorem~III.7 in \cite{RS1}, if the dual of a Banach space is separable, so is the Banach space itself). Therefore working with unital subalgebras of smooth functions on a Banach manifold modelled on $L_\infty(\mathcal{H})$ (or on $L_{\res}(\mathcal{H})$ and $\mathfrak{u}_{\res}(\mathcal{H})$ defined below) may lead to unexpected difficulties. 
For this reason, we will adapt the definition of Banach Poisson manifold and work with local sections of subbundles of the cotangent bundle.
The link between unital subalgebras of $\mathcal{C}^{\infty}(M)$ and subbundles of the cotangent bundle is given by next definition.
\end{enumerate}
}
\end{remark}

\begin{definition}
Let $M$ be a Banach manifold and $\mathcal{A}$ be a unital subalgebra of $\mathcal{C}^{\infty}(M)$. The first jet of $\mathcal{A}$, denoted by $\operatorname{J}^1(\mathcal{A})$ is the subbundle of the cotangent bundle $T^*M$ whose fiber over $p\in M$ is the space of differentials of functions in $\mathcal{A}$,
$$
\operatorname{J}^1(\mathcal{A})_p =\{ df_p~: f\in\mathcal{A} \}.
$$
\end{definition}


\begin{definition}
We will say that $\mathbb{F}$ is a subbundle  of $T^*M$ \textbf{in duality} with the tangent bundle to $M$ if, for every $p\in M$, 
\begin{enumerate}
\item $\mathbb{F}_p$ is an injected Banach space of $T_p^*M$, i.e. $\mathbb{F}_p$ admits a Banach space structure such that the injection $\mathbb{F}_p\hookrightarrow T_p^*M$ is continuous,
\item the natural duality pairing between $T_p^*M$ and $T_pM$ restricts to a duality pairing between $\mathbb{F}_p$ and $T_pM$, i.e.  $\mathbb{F}_p$ separates points in $T_pM$.
\end{enumerate}
\end{definition}

We will denote by $\Lambda^2\mathbb{F}^{*}$ the vector bundle over $M$ whose fiber over $p$ is the Banach space of continuous skew-symmetric bilinear maps on the subspace $\mathbb{F}_p$ of $T_p^*M$.

\begin{definition}\label{Poisson_tensor}
Let $M$ be a Banach manifold and $\mathbb{F}$ a subbundle of $T^*M$ in duality with $TM$. A smooth section $\pi$ of $\Lambda^2\mathbb{F}^*$  is called a \textbf{Poisson tensor} on $M$ with respect to $\mathbb{F}$ if~:
\begin{enumerate}
\item for any closed local sections $\alpha$, $\beta$ of  $\mathbb{F}$, the differential $d\left(\pi(\alpha, \beta)\right)$ is a local section of $\mathbb{F}$;
\item (Jacobi) for any closed local sections $\alpha$, $\beta$, $\gamma$ of $\mathbb{F}$,
\begin{equation}\label{Jacobi_Poisson}
\pi\left(\alpha, d\left(\pi(\beta, \gamma)\right)\right) + \pi\left(\beta, d\left(\pi(\gamma, \alpha)\right)\right) + \pi\left(\gamma, d\left(\pi(\alpha, \beta)\right) \right)= 0.
\end{equation}
\end{enumerate}
\end{definition}

\begin{remark}{\rm
\begin{enumerate}
\item The first condition in Definition \ref{Poisson_tensor} is necessary in order to make sence of equation~\eqref{Jacobi_Poisson} since the Poisson tensor is defined only on local sections of $\mathbb{F}$.
\item  Consider  a unital subalgebra $\mathcal{A}$ of $\mathcal{C}^{\infty}(M)$ and set $\mathbb{F} = \operatorname{J}^1(\mathcal{A})$ the first jet of functions in $\mathcal{A}$. 
Then equation \eqref{Jacobi_Poisson} for a Poisson tensor $\pi$ on $M$ with respect to $\mathbb{F}$ is equivalent to the Jacobi identity for the Poisson bracket  defined for $f,g\in\mathcal{A}$ by $\{f, g\} = \pi(df, dg)$.
\end{enumerate}
}
\end{remark}

\begin{definition}\label{Poisson-Manifold}
A \textbf{generalized Banach Poisson manifold} is a triple $(M, \mathbb{F}, \pi)$ consisting of a smooth Banach manifold $M$, a subbundle $\mathbb{F}$ of the cotangent bundle $T^*M$ in duality with $TM$, and a Poisson tensor $\pi$ on $M$ with respect to $\mathbb{F}$.
\end{definition}

\begin{remark}
{\rm
Let us make the link between our definition of generalized Banach Poisson manifold and related notions in the literature. Consider  a unital subalgebra $\mathcal{A}$ of $\mathcal{C}^{\infty}(M)$, set $\mathbb{F} = \operatorname{J}^1(\mathcal{A})$ the first jet of functions in $\mathcal{A}$, and consider a Poisson bracket on $\mathcal{A}$ given by a Poisson tensor~:
$\{f, g\} = \pi(df, dg)$. Our definition of generalized Banach Poisson manifold differs from the one given in \cite{NST14} and the definition of sub  Poisson manifold given in \cite{CP12} by the fact that we do not assume the existence of Hamiltonian vector fields associated to functions $f\in \mathcal{A}$ (condition P3 in Definition~2.1 in \cite{NST14} and condition $P~:T^\flat M\rightarrow TM$ in \cite{CP12}). In other words, for $f\in\mathcal{A}$, $\{f, \cdot\}$ is a derivation on $\mathcal{A}\subset\mathcal{C}^{\infty}(M)$ that may not --with our definition of Poisson manifold-- be given by a smooth vector field on $M$. However, since the Poisson bracket is given by a smooth Poisson tensor, $\{f, \cdot\}$ is a smooth section of the bundle $\operatorname{J}^1(\mathcal{A})^*$ whose fiber over $p\in M$ is the dual Banach space to $\operatorname{J}^1(\mathcal{A})_p$. Moreover, in order to stay in the Banach context, we suppose that $\mathbb{F}_p$ has a structure of Banach space.
}\end{remark}

\subsection{Banach Symplectic manifolds}\label{symp}
An important class of finite-dimensional Poisson manifolds is provided by symplectic manifolds. As we will see below, this is also the case in the Banach setting, i.e. general Banach symplectic manifolds (not necessarily strong symplectic) are particular examples of generalized Banach Poisson manifolds.
Let us recall the following definitions. The exterior derivative $d$ associates to a  $n$-form on a Banach manifold $M$ a  $(n+1)$-form on $M$. In particular, for any $2$-form $\omega$ on a Banach manifold $M$, the exterior derivative of $\omega$ is the $3$-form $d\omega$ defined by~:
$$
\begin{array}{ll}
d\omega_p(X, Y, Z) = &\!\!\! - \omega_p([\tilde{X}, \tilde{Y}], \tilde{Z}) + \omega_p([\tilde{X}, \tilde{Z}], \tilde{Y}) - \omega_p([\tilde{Y}, \tilde{Z}], \tilde{X})
 +\left\langle d_p\left(\omega(\tilde{Y}, \tilde{Z})\right), \tilde{X}\right\rangle_{T_p^*M, T_pM} \\&\!\!\! -  \left\langle d_p\left(\omega(\tilde{X}, \tilde{Z})\right), \tilde{Y}\right\rangle_{T_p^*M, T_pM} \!\!\! +  \left\langle d_p\left(\omega(\tilde{X}, \tilde{Y})\right), \tilde{Z}\right\rangle_{T_p^*M, T_pM},
\end{array}
$$
where $\tilde{X}, \tilde{Y}, \tilde{Z}$ are any smooth extensions of $X$, $Y$ and $Z\in T_pM$ around $p\in M$. An expression of this formula in a chart shows that it does not depend on the extensions $\tilde{X}, \tilde{Y}, \tilde{Z}$, but only on the values of these vector fields at $p\in M$, i.e. it defines an tensor (see Proposition~3.2, chapter~V in \cite{La01}). The contraction or interior product $i_X\omega$ of a $n$-form $\omega$ with a vector field $X$ is the $(n-1)$-form defined by $$i_X\omega(Y_1, \cdots, Y_{n-1}) := \omega(X, Y_1, \cdots, Y_{n-1}).$$ The Lie derivative $\mathcal{L}_X$ with respect to a vector field $X$ can be defined using the Cartan formula
\begin{equation}\label{Cartan_formula}
\mathcal{L}_X = i_X d + d\, i_X.
\end{equation}
The Lie derivative, the bracket $[X, Y]$ of two vector fields $X$ and $Y$, and the interior product satisfy the following relation  (see Proposition~5.3, chapter~V in \cite{La01})~:
\begin{equation}\label{relation_Lie_bracket_i}
i_{[X, Y]} = \mathcal{L}_{X}i_Y - i_Y\mathcal{L}_X.
\end{equation}
Let us recall the definition of a Banach (weak) symplectic manifold.
\begin{definition}
A \textbf{Banach symplectic manifold} is a Banach manifold $M$ endowed with a $2$-form $\omega\in \Gamma\left(\Lambda^2T^*M\right)$ such that
\begin{enumerate}
\item $\omega$ is non-degenerate~: $\omega_p^\sharp~: T_pM \rightarrow  T_p^*M$, $ X \mapsto  i_X\omega := \omega(X, \cdot)$ 
 is injective $\forall p\in M$~;
\item $\omega$ is closed~: $d\omega = 0$.
\end{enumerate}
\end{definition}

\begin{lemma}\label{lemma_symplectic_vector_field}
Let $(M, \omega)$ be a Banach symplectic manifold. 
Consider $\alpha$ and $\beta$  two closed local sections of $\omega^\sharp(T M)$, i.e. $d\alpha = d\beta = 0$, $\alpha = \omega(X_\alpha, \cdot)$ and $\beta = \omega(X_\beta, \cdot)$ for some local vector fields $X_{\alpha}$ and $X_\beta$. Then 
\begin{enumerate}
\item  $X_\alpha$ and $X_{\beta}$ are symplectic vector fields~: $\mathcal{L}_{X_{\alpha}}\omega = 0 = \mathcal{L}_{X_{\beta}}\omega$
\item $i_{[X_{\alpha}, X_\beta]}\omega = -d(\omega(X_\alpha, X_\beta)).$
\end{enumerate}
\end{lemma}

\begin{proof}
\begin{enumerate}
\item
Using the Cartan formula \eqref{Cartan_formula}, one has $\mathcal{L}_{X_{\alpha}}\omega = i_{X_\alpha}d\omega + d\,i_{X_\alpha}\omega = d\,i_{X_\alpha}\omega$, since $\omega$ is closed. But by definition $i_{X_\alpha}\omega = \alpha$ is closed. Using $d\circ d = 0$ (see Supplement~6.4A in \cite{AMR88} for a proof of this identity in the Banach context), it follows that $\mathcal{L}_{X_{\alpha}}\omega=0$. Similarly $\mathcal{L}_{X_{\beta}}\omega = 0$.
\item
By relation \eqref{relation_Lie_bracket_i}, one has 
$$
i_{[X_\alpha, X_\beta]}\omega = \mathcal{L}_{X_\alpha}i_{X_\beta}\omega - i_{X_\beta}\mathcal{L}_{X_\alpha}\omega,
$$
where the second term in the RHS vanishes by (1). Using Cartan formula, one gets
$$
i_{[X_\alpha, X_\beta]}\omega = d\,i_{X_\alpha}i_{X_\beta}\omega + i_{X_\alpha}d\,(i_{X_\beta}\omega) = d\,i_{X_\alpha}i_{X_\beta}\omega = d\left(\omega(X_{\beta}, X_{\alpha})\right) =- d\left(\omega(X_{\alpha}, X_{\beta})\right),
$$
where we have used that $i_{X_\beta}\omega = \beta $ is closed.
\end{enumerate}
\end{proof}

\begin{proposition}
Any  Banach symplectic manifold $(M, \omega)$ is naturally a  generalized Banach Poisson manifold $(M, \mathbb{F}, \pi)$ with
\begin{enumerate}
\item $\mathbb{F} =  \omega^\sharp(T M)$;
\item  $\pi~:\omega^\sharp(TM)\times \omega^\sharp(TM) \rightarrow \mathbb{R}$ defined by  $(\alpha, \beta)  \mapsto  \omega(X_{\alpha}, X_{\beta})$
%
where $X_\alpha$ and $X_{\beta}$ are uniquely defined by $\alpha = \omega(X_\alpha, \cdot)$ and $\beta = \omega(X_{\beta}, \cdot)$.
\end{enumerate}
\end{proposition}

\begin{proof}
\begin{enumerate}
\item By Lemma~\ref{lemma_symplectic_vector_field}, for any closed local sections $\alpha$ and $\beta$ of $\mathbb{F}$, with $\alpha = \omega(X_\alpha, \cdot)$ and $\beta = \omega(X_\beta, \cdot)$, one has 
$$
d\left(\pi(\alpha, \beta)\right) := d\left(\omega(X_{\alpha}, X_\beta)\right) = -i_{[X_\alpha, X_\beta]}\omega,
$$
hence is a local section of $\mathbb{F} = \omega^\sharp(T M)$.
\item
Let us show that $\pi$ satisfies the Jacobi identity~\eqref{Jacobi_Poisson}. Consider closed local sections $\alpha, \beta$ and $\gamma$ of $\mathbb{F}$  and define the local vector fields $X_\alpha$, $X_\beta$ and $X_\gamma$ by 
$\alpha = i_{X_\alpha}\omega$, $\beta = i_{X_\beta}\omega$ and $\gamma = i_{X_\gamma}\omega$. 
Using Lemma~\ref{lemma_symplectic_vector_field}, the differential of $\omega$ satisfies
$$
\begin{array}{lll}
d\omega(X_\alpha, X_\beta, X_\gamma)& = &2\left( - \omega([{X}_\alpha, X_\beta], X_\gamma) + \omega([X_\alpha, X_\gamma], X_\beta) - \omega([X_\beta, X_\gamma], X_\alpha)\right)
\\& = &2\left(\pi\left(d\left(\pi(\alpha, \beta), \gamma\right) \right)\right)+ \pi\left(d\left(\pi( \gamma, \alpha)\right), \beta\right) +  \pi\left( d\left(\pi(\beta, \gamma)\right), \alpha\right).
\end{array}
$$
Since $\omega$ is closed, the Jacobi identity~\eqref{Jacobi_Poisson} is satisfied.
\end{enumerate}
\end{proof}

\subsection{Banach Lie--Poisson spaces}\label{LP}
Banach Lie--Poisson spaces were introduced in \cite{OR03}. Here we extend this notion to an arbitrary duality pairing.
\begin{definition}
Consider a duality pairing $\langle\cdot, \cdot\rangle_{\mathfrak{g}_+, \mathfrak{g}_-}~:\mathfrak{g}_+\times\mathfrak{g}_-\rightarrow\mathbb{K}$ between two Banach spaces.
We will say that $\mathfrak{g}_+$ is a  \textbf{Banach Lie--Poisson space with respect to} $\mathfrak{g}_-$ if $\mathfrak{g}_-$ is a Banach Lie algebra $(\mathfrak{g}_-, [\cdot,\cdot]_{\mathfrak{g}_-})$ which acts continuously on $\mathfrak{g}_+\hookrightarrow \mathfrak{g}_-^*$ by coadjoint action, i.e.
$$
\ad^*_\alpha x \in\mathfrak{g}_+, 
$$
for all $x \in \mathfrak{g}_+$ and $\alpha\in\mathfrak{g}_-$, and $\ad^*~: \mathfrak{g}_-\times\mathfrak{g}_+\rightarrow \mathfrak{g}_+$ is continuous.
\end{definition}

\begin{remark}
{\rm
A Banach Lie--Poisson space $\mathfrak g_+$ with respect to its continuous dual space $\mathfrak{g}_+^*$ is a Banach Lie--Poisson space in the sense of Definition~4.1 in \cite{OR03}.
}
\end{remark}

%
%

The following Theorem is a generalization of Theorem~4.2 in \cite{OR03} to the case of an arbitrary duality pairing between two Banach spaces $\mathfrak{g}_+$ and $\mathfrak{g}_-$. See also Corollary~2.11 in \cite{NST14} for an analogous statement. 
We will include the proof for sake of completeness.

\begin{theorem}\label{4.2}
Consider a duality pairing $\langle\cdot, \cdot\rangle_{\mathfrak{g}_+, \mathfrak{g}_-}~:\mathfrak{g}_+\times\mathfrak{g}_-\rightarrow\mathbb{K}$ between two Banach spaces, and suppose that  \textbf{$\mathfrak{g}_+$ is a Banach Lie--Poisson space  with respect to $\mathfrak{g}_-$. }

Denote by $\mathbb{F}$ the subbundle of $T^*\mathfrak{g}_+\simeq \mathfrak{g}_+\times\mathfrak{g}_+^*$ whose fiber at   $x\in\mathfrak{g}_+$ is given by
$$
\mathbb{F}_x = \{x\}\times\mathfrak{g}_-\subset \{x\}\times\mathfrak{g}_+^*\simeq T_x^*\mathfrak{g}_+.
$$ 
For $\alpha$ and $\beta$ any two local sections of $\mathbb{F}$, define a tensor $\pi\in\Lambda^2\mathbb{F}^*$ by~:
$$
\pi_x(\alpha, \beta) := \left\langle x, [\alpha(x), \beta(x)]_{\mathfrak{g}_-}\right\rangle_{\mathfrak{g}_+, \mathfrak{g}_-}.
$$
Then $(\mathfrak g_+, \mathbb{F}, \pi)$ is a generalized Banach Poisson manifold, and $\pi$ takes values in $\Lambda^2\mathfrak{g}_+\subset\Lambda^2\mathbb{F}^*$.

Let $\mathcal{A}$ be the unital subalgebra of $\mathcal{C}^{\infty}(\mathfrak{g}_+)$ consisting of all  functions with differentials in $\mathfrak{g}_-$~:
$$
\mathcal{A} :=\{ f \in \mathcal{C}^\infty(\mathfrak{g}_+)~: d_xf \in \mathfrak{g}_-\subset \mathfrak{g}_+^* \textrm{ for any }x\in\mathfrak{g}_+\}.
$$
Define the bracket of two functions $f, h$  in  $\mathcal{A}$ by
\begin{equation}\label{pipoisson}
\{f, h\}(x) := \pi_x(df_x, dh_x) = \left\langle x, [df_x, dh_x]_{\mathfrak{g}_-}\right\rangle_{\mathfrak{g}_+, \mathfrak{g}_-}, 
\end{equation}
where $x\in\mathfrak{g}_+$, and $df$ and $dh$ denote the Fr\'echet derivatives of $f$ and $h$ respectively. Then $\{\cdot,\cdot\}~:\mathcal{A}\times\mathcal{A}\rightarrow\mathcal{A}$ is a Poisson bracket on $\mathfrak{g}_+$. If $h$ is a smooth function on $\mathfrak{g}_+$ belonging to $\mathcal{A}$, the associated Hamiltonian vector field is given by
$$
X_h(x) = -\ad^*_{dh_x}x \in\mathfrak{g}_+.
$$
\end{theorem}

\begin{proof}
Let $\alpha$ and $\beta$ be any closed local sections of $\mathbb{F}$.  Then $\alpha$ and $\beta$ are functions from $\mathfrak{g}_+$ to $\mathfrak{g}_-$, and we will denote by $T_x\alpha~:T_x\mathfrak{g}_+\simeq\mathfrak{g}_+\rightarrow \mathfrak{g}_-\simeq T_{\alpha(x)}\mathfrak{g}_-$ and similarly $T_x\beta~: \mathfrak{g}_+\rightarrow \mathfrak{g}_-$ their derivatives at $x\in\mathfrak{g}_+$.
For any tangent vector $X\in T_x\mathfrak{g}_+ \simeq \mathfrak{g}_+ $, one has
$$
\begin{array}{ll}
d_x\pi\left(\alpha, \beta\right)(X) 
&= \left\langle X, [\alpha(x), \beta(x)]_{\mathfrak{g}_-}\right\rangle_{\mathfrak{g}_+, \mathfrak{g}_-} +  \left\langle x, [T_x\alpha(X), \beta]_{\mathfrak{g}_-}\right\rangle_{\mathfrak{g}_+, \mathfrak{g}_-} + \left\langle x, [\alpha, T_x\beta(X)]_{\mathfrak{g}_-}\right\rangle_{\mathfrak{g}_+, \mathfrak{g}_-}\\
& = \left\langle X, [\alpha(x), \beta(x)]_{\mathfrak{g}_-}\right\rangle_{\mathfrak{g}_+, \mathfrak{g}_-} -  \left\langle \ad^*_{\beta}x, T_x\alpha(X) \right\rangle_{\mathfrak{g}_+, \mathfrak{g}_-} + \left\langle \ad^*_{\alpha}x, T_x\beta(X)\right\rangle_{\mathfrak{g}_+, \mathfrak{g}_-}
\end{array}
$$
Since $\alpha$ and $\beta$ are closed local sections of $\mathbb{F}\subset T^*\mathfrak{g}_+$, by Poincar\'e Lemma (see Theorem~4.1 in \cite{La02}), there exist locally real valued smooth functions $f$ and $g$ on $\mathfrak{g}_+$ such that $\alpha = df$ and $\beta = dg$. It follows that $T_x\alpha \in L\left(\mathfrak{g}_+; L(\mathfrak{g}_+, \mathbb{R})\right)\simeq L^2(\mathfrak{g}_+;\mathbb{R})$ is the second derivative $d_x^2f$ of $f$ at $x\in\mathfrak{g}_+$ and is symmetric (see Proposition~3.3 in \cite{La02}). Similarly $T_x\beta = d_x^2g$ is a symmetric bilinear map on $\mathfrak{g}_+$. Consequently
$$
-  \left\langle \ad^*_{\beta}x, T_x\alpha(X) \right\rangle_{\mathfrak{g}_+, \mathfrak{g}_-}  = - \left\langle X, T_x\alpha(\ad^*_{\beta}x) \right\rangle_{\mathfrak{g}_+, \mathfrak{g}_-} 
$$
and 
$$
\left\langle \ad^*_{\alpha}x, T_x\beta(X)\right\rangle_{\mathfrak{g}_+, \mathfrak{g}_-} = \left\langle X, T_x\beta(\ad^*_{\alpha}x)\right\rangle_{\mathfrak{g}_+, \mathfrak{g}_-}.
$$
Therefore, for any closed local section $\alpha$ and $\beta$ of $\mathbb{F}$, and any $x\in\mathfrak{g}_+$,
\begin{equation}\label{dpi}
d_x\pi\left(\alpha, \beta\right)  = [\alpha(x), \beta(x)]_{\mathfrak{g}_-} - T_x\alpha(\ad^*_{\beta}x) + T_x\beta(\ad^*_{\alpha}x)
\end{equation}
belongs to $\mathfrak{g}_-$. It follows that  $d\pi\left(\alpha, \beta\right)$ a local section of $\mathbb{F}$.
Let us show that $\pi$ satisfies the Jacobi identity~\eqref{Jacobi_Poisson}. One has 
$$
\begin{array}{ll}
\pi_x\left(\alpha, d\left(\pi(\beta, \gamma)\right)\right) 
 = & \left\langle x, [\alpha(x),  [\beta(x), \gamma(x)]_{\mathfrak{g}_-} ]_{\mathfrak{g}_-}\right\rangle_{\mathfrak{g}_+, \mathfrak{g}_-} 
- \left\langle \ad^*_{\alpha}x,  T_x\beta(\ad^*_{\gamma}x)\right\rangle_{\mathfrak{g}_+, \mathfrak{g}_-} 
\\&+ \left\langle \ad^*_\alpha x, T_x\gamma(\ad^*_{\beta}x) ]_{\mathfrak{g}_-}\right\rangle_{\mathfrak{g}_+, \mathfrak{g}_-} 
\end{array}
$$
By the Jacobi identity for the Lie bracket $[\cdot, \cdot]_{\mathfrak{g}_-}$ and by the symmetry of $T_x\alpha$, $T_x\beta$ and $T_x\gamma$, the Jacobi identity for $\pi$ is satisfied.
Moreover, for any local section $\alpha$ of $\mathbb{F}$, $\pi_x(\alpha, \cdot) = \ad^*_{\alpha} x$ belongs to $\mathfrak{g}_+$ since $\mathfrak{g}_+$ is a Banach Lie--Poisson space  with respect to $\mathfrak{g}_-$. Therefore $\pi\in\Lambda^2\mathfrak{g}_+\subset\Lambda^2\mathbb{F}^*$.

The bracket \eqref{pipoisson} of two functions $f,g\in\mathcal{A}$ takes values in $\mathcal{A}$ because, by equation~\eqref{dpi}, $d_x\{f, g\}$ belongs to $\mathfrak{g}_-$.
By definition $\{\cdot,\cdot\}$ is skew-symmetric and satisfies the Leibniz rule. The Jacobi identity for $\{\cdot,\cdot\}$ follows from the Jacobi identity for $\pi$. The expression of the hamiltonian vector field associated to $h\in\mathcal{A}$ is straightforward.
\end{proof}

We give below some examples 
 of Banach Lie--Poisson spaces (see \cite{OR03}, \cite{OR07}, 
 and \cite{BRT07} for more information on these spaces). 

\begin{example}\label{exLpLq}{\rm
\textit{Dual Banach Lie algebras of operators.} 
Let  $p$ and $q$ be such that $1< p\leq q< \infty$ and $\frac{1}{p}+\frac{1}{q} = 1$. Then $\operatorname{L}_p(\mathcal{H})^* \simeq L_q(\mathcal{H})$ and $L_q(\mathcal{H})^* \simeq \operatorname{L}_p(\mathcal{H})$ where the duality pairing is given by the trace (see example~\ref{ex_duality2}).
Moreover
\begin{align*}
\ad^*_\alpha x(\beta) = \Tr \left(x [\alpha, \beta]_{L_{q}(\mathcal{H})}\right) = \Tr\left(x\alpha\beta- x\beta\alpha\right)=  \Tr\left(x\alpha\beta- \alpha x\beta\right) = \Tr\left([x, \alpha]\beta\right),
\end{align*}
where the first bracket is the Lie bracket of the dual space $L_{q}(\mathcal{H})$, and the second is the commutator of the bounded linear operators $x\in L_{p}(\mathcal{H})$ and $\alpha\in L_{q}(\mathcal{H})$. Since $\operatorname{L}_p(\mathcal{H})$ is an ideal of $L_{\infty}(\mathcal{H})$,  $[x, \alpha]\in \operatorname{L}_p(\mathcal{H})$, and the pairing given by the trace being non-degenerate, one has  $$\ad^*_\alpha x = [x, \alpha]\in \operatorname{L}_p(\mathcal{H})$$ for any $x\in \operatorname{L}_p(\mathcal{H})$ and any $\alpha\in L_{q}(\mathcal{H})$. Therefore $\operatorname{L}_p(\mathcal{H})$ is a Banach Lie--Poisson space with respect to $L_{q}(\mathcal{H})$. 
In the same manner, one has for any $x\in \operatorname{L}_p(\mathcal{H})$ and any $\alpha\in L_{q}(\mathcal{H})$ 
 $$\ad^*_x\alpha = [\alpha, x]\in L_{q}(\mathcal{H}),$$ 
hence $L_{q}(\mathcal{H})$ is a Banach Lie--Poisson space with respect to $L_{p}(\mathcal{H})$.
}\end{example}

\begin{example}\label{exL1Linfty}{\rm \textit{Trace class operators and bounded operators.} 
For the same reasons as in the previous example, the Banach Lie algebra $L_1(\mathcal{H})$ is a Banach Lie--Poisson space with respect to  $L_{\infty}(\mathcal{H})$ and $L_{\infty}(\mathcal{H})$ is a Banach Lie--Poisson space with respect to $L_1(\mathcal{H})$, the (weak) duality pairing being given by the trace. 
}\end{example}

\begin{example}\label{exL1L2}{\rm \textit{Trace class operators and Hilbert-Schmidt operators.}
Since the trace is a weak duality pairing between $L_1(\mathcal{H})$ and $L_2(\mathcal{H}) \subset L_{\infty}(\mathcal{H})$ (see Example~\ref{exDualityL1L2}), one can consider the coadjoint action of $L_1(\mathcal{H})$ on $L_2(\mathcal{H})$ and vice-versa. For any $x\in \operatorname{L}_1(\mathcal{H})$ and any $\alpha\in L_{2}(\mathcal{H})$, one has
$$\ad^*_x\alpha  = - \ad^*_\alpha x = [\alpha, x] \in  \operatorname{L}_1(\mathcal{H})\cap L_{2}(\mathcal{H}),$$ therefore $L_1(\mathcal{H})$ is Banach Lie--Poisson space with respect to $L_2(\mathcal{H})$, and $L_2(\mathcal{H})$ is a Banach Lie--Poisson space with respect to $L_1(\mathcal{H})$. Using~\eqref{pipoisson}, one obtains a Poisson bracket on $L_1(\mathcal{H})$ defined on the algebra of functions on $L_1(\mathcal{H})$ with differentials in $L_{2}(\mathcal{H})\subset L_{\infty}(\mathcal{H})$, as well as a  Poisson bracket on $L_{2}(\mathcal{H})$ defined on those functions on $L_{2}(\mathcal{H})$ which have their differential in $L_1(\mathcal{H})\subset L_{2}(\mathcal{H})$.
}\end{example}

\begin{example}{\rm
 \textit{Banach Lie algebras of upper and lower triangular operators.}\label{exUpLow}
For $1<p<\infty$, consider the Banach algebra $\operatorname{L}_p(\mathcal{H})_-$  of lower triangular operators in $\operatorname{L}_p(\mathcal{H})$ defined by \eqref{triangularLp+} and its complement $\operatorname{L}_p(\mathcal{H})_{++}$ consisting in stricktly upper triangular operators in $\operatorname{L}_p(\mathcal{H})$.  
One can identify $\operatorname{L}_p(\mathcal{H})_{-} ^*$ with $\operatorname{L}_p(\mathcal{H})^*/ \left(\operatorname{L}_p(\mathcal{H})_{-}\right)^0$ where 
$$
\left(\operatorname{L}_p(\mathcal{H})_{-}\right)^0 := \{\alpha \in L_{q}(\mathcal{H}), \Tr\left(\alpha x\right) = 0,~~\forall x\in \operatorname{L}_p(\mathcal{H})_{-}\}
$$
Recall that $\operatorname{L}_p(\mathcal{H})^* \simeq L_q(\mathcal{H})$ where  $\frac{1}{p}+\frac{1}{q}=1$, the duality pairing being given by the trace. 
It is easy to see that $\left(\operatorname{L}_p(\mathcal{H})_{-}\right)^0 $ is isomorphic to the Banach space $L_q(\mathcal{H})_{--}$  of stricktly lower triangular operators in $L_q(\mathcal{H})$.
Therefore, by the direct sum decomposition  \eqref{lppm}, one has
$$\operatorname{L}_p(\mathcal{H})_{-} ^*\simeq L_{q}(\mathcal{H})_{+}.$$
The coadjoint action of an element $\alpha\in L_{q}(\mathcal{H})_{+}$ on $x\in \operatorname{L}_p(\mathcal{H})_{-}\subset \left(\operatorname{L}_p(\mathcal{H})_{-}\right)^{**}$ reads
\begin{align*}
\ad_{\alpha}^*x(\beta) = \Tr \left(x [\alpha, \beta]_{L_{q}(\mathcal{H})_{+}}\right) = \Tr\left([x, \alpha]\beta\right),
\end{align*}
where $\beta$ is an arbitrary element in $L_{q}(\mathcal{H})_{+}$.
Since $\operatorname{L}_p(\mathcal{H})$ and $L_q(\mathcal{H})$ are ideals in $L_\infty(\mathcal{H})$, one has $$[x, \alpha] \in [\operatorname{L}_p(\mathcal{H}), L_q(\mathcal{H})]\subset \operatorname{L}_p(\mathcal{H})\cap L_q(\mathcal{H}).$$ The relation $\operatorname{L}_p(\mathcal{H})_{++} \subset \left(L_{q}(\mathcal{H})_{+}\right)^0$ then implies
 \begin{align*}
\ad_{\alpha}^*x(\beta) = \Tr\left(p_{\operatorname{L}_p(\mathcal{H})_{-}}\left([x, \alpha]\right)\beta\right),~~\forall \beta \in L_{q}(\mathcal{H})_{+},
\end{align*}
where $p_{\operatorname{L}_p(\mathcal{H})_{-}}$ is the projection onto $\operatorname{L}_p(\mathcal{H})_{-}$ with respect to the direct sum decomposition \eqref{lpmp}.
From $\operatorname{L}_p(\mathcal{H})_-\subset (L_q(\mathcal{H})_{--})^0$ and from the direct sum decomposition~\eqref{lppm}, it follows that 
$$\ad_{\alpha}^*x = p_{\operatorname{L}_p(\mathcal{H})_{-}}\left([x, \alpha]\right).$$ 
In particular, $\ad_{\alpha}^*x\in \operatorname{L}_p(\mathcal{H})_{-}$ for any $x\in \operatorname{L}_p(\mathcal{H})_{-}$ and any $\alpha \in L_{q}(\mathcal{H})_{+}$. Therefore $\operatorname{L}_p(\mathcal{H})_{-}$ is a Banach Lie--Poisson space with respect to $L_{q}(\mathcal{H})_{+}$. Similarly one has $$\ad_{x}^*\alpha = p_{L_q(\mathcal{H})_{+}}\left([\alpha, x]\right),$$ for any $x\in \operatorname{L}_p(\mathcal{H})_{-}$ and any $\alpha \in L_q(\mathcal{H})_{+}$. Therefore $L_q(\mathcal{H})_{+}$ is a Banach Lie--Poisson space with respect to $\operatorname{L}_p(\mathcal{H})_{-}$. Note that the existence of the projections $p_{\operatorname{L}_p(\mathcal{H})_{-}}$ and $p_{L_q(\mathcal{H})_{+}}$ is crucial in this example. This is the reason why we have excluded the case $p=1$ and $q=\infty$.
}\end{example}

\begin{example}\label{exIwasawa}{\rm
 \textit{Iwasawa Banach Lie algebras.}
For $1<p<\infty$, consider the unitary algebra $\mathfrak{u}_p(\mathcal{H})$ defined by \eqref{u2}, and its complement $\mathfrak{b}^+_p(\mathcal{H})$ defined by \eqref{b2pm}.  For $q := \frac{p}{p-1}$, let us denote by  $\langle\cdot,\cdot\rangle_{\mathbb{R}}$ the continuous bilinear map given by the imaginary part of the trace~:
$$
\begin{array}{lcll}
\langle\cdot,\cdot\rangle_{\mathbb{R}}~: & \operatorname{L}_p(\mathcal{H})\times L_q(\mathcal{H})& \longrightarrow &\mathbb{R}\\
& (x, \alpha) & \longmapsto & \Im \Tr\left(x\alpha\right).
\end{array}
$$
It is a strong duality pairing between $ \operatorname{L}_p(\mathcal{H})$ and $L_q(\mathcal{H})$ viewed as real Banach spaces. By Lemma~\ref{decub}, one has the direct sum decomposition $$\operatorname{L}_p(\mathcal{H}) = \mathfrak{u}_p(\mathcal{H})\oplus\mathfrak{b}^+_p(\mathcal{H}).$$
%
Since $\left(\mathfrak{u}_{p}(\mathcal{H})\right)^0 \simeq \mathfrak{u}_q(\mathcal{H})$ and $\left(\mathfrak{b}_p^+(\mathcal{H})\right)^0 \simeq \mathfrak{b}_q^+(\mathcal{H})$, one has 
$$\mathfrak{u}_{p}(\mathcal{H})^* \simeq L_q(\mathcal{H})/\left(\mathfrak{u}_{p}(\mathcal{H})\right)^0 \simeq L_q(\mathcal{H})/\mathfrak{u}_q(\mathcal{H}) \simeq \mathfrak{b}^+_q(\mathcal{H})$$ and similarly $$\mathfrak{b}_{p}^+(\mathcal{H})^* = \mathfrak{u}_q(\mathcal{H}).$$ 
Consider the coadjoint action of an element $\alpha\in\mathfrak{b}^+_q(\mathcal{H})$ on an element $x\in\mathfrak{u}_p(\mathcal{H})\subset\mathfrak{u}_p(\mathcal{H})^{**}$
\begin{align*}
\ad_{\alpha}^*x(\beta)  = \langle x, [\alpha, \beta]\rangle_{L^p, L^q} = \Im\Tr \left(x [\alpha, \beta]_{\mathfrak{b}_q}\right) = \Im\Tr\left([x, \alpha]\beta\right),
\end{align*}
where $\beta$ is an arbitrary element in $\mathfrak{b}^+_q(\mathcal{H})$.
Since $\mathfrak{b}^+_p(\mathcal{H}) \subset \left(\mathfrak{b}^+_q(\mathcal{H})\right)^0$ and $[\operatorname{L}_p(\mathcal{H}), L_q(\mathcal{H})]\in \operatorname{L}_p(\mathcal{H})\cap L_q(\mathcal{H})$, one has
\begin{align*}
\ad_{\alpha}^*x(\beta)  
= \Im\Tr\left(p_{\mathfrak{u}_p,+}\left([x, \alpha]\right)\beta\right) = \langle p_{\mathfrak{u}_p,+}\left([x, \alpha]\right), \beta\rangle_{L^p, L^q},~~\forall  \beta \in \mathfrak{b}_q^+(\mathcal{H}), 
\end{align*}
where $p_{\mathfrak{u}_p,+}$ is the projection onto $\mathfrak{u}_p(\mathcal{H})$ defined by \eqref{projectionu+}.  Therefore
$$\ad^*_{\alpha} x = p_{\mathfrak{u}_p,+}\left([x, \alpha]\right).$$ Analogously one has
$$
\ad^*_{x}\alpha = p_{\mathfrak{b}_q^+}\left([\alpha, x]\right),
$$
for any $x\in\mathfrak{u}_p(\mathcal{H})$ and any $\alpha\in \mathfrak{b}_q^+(\mathcal{H})$. Consequently $\mathfrak{u}_p(\mathcal{H})$ and $\mathfrak{b}_q^+(\mathcal{H})$ are dual Banach Lie--Poisson spaces. Similarly $\mathfrak{u}_p(\mathcal{H})$ and $\mathfrak{b}_q^-(\mathcal{H})$ are dual Banach Lie--Poisson spaces. 
}
\end{example}

\section{Banach Lie bialgebras}

In the finite dimensional case,  a couple $(\mathfrak{g}, \mathfrak{g}^*)$ of Lie algebras is a Lie bialgebra if and only if  the triple of Lie algebras $(\mathfrak{g}\oplus \mathfrak{g}^*, \mathfrak{g}, \mathfrak{g}^*)$ form a Manin triple. In that case, $(\mathfrak{g}^*, \mathfrak{g})$ is also a Lie bialgebra.
The symmetry of the situation comes from the fact that $\mathfrak{g}^{**} = \mathfrak{g}$ for finite dimensional spaces. 
In Section~\ref{Definition of Banach Lie bialgebras}, we introduce the notion of Banach Lie bialgebra with respect to an arbitrary duality pairing. 
In Section~\ref{Banach Lie bialgebras versus Manin triples}, we show that a Banach Lie bialgebra $\mathfrak{g}_+$ with respect to a Banach Lie algebra $\mathfrak{g}_-$ gives rise to a Manin triple $(\mathfrak{g}_+\oplus\mathfrak{g}_-, \mathfrak{g}_+, \mathfrak{g}_-)$ if and only if $\mathfrak{g}_+$  is also a Banach Lie--Poisson space with respect to $\mathfrak{g}_-$ (see Theorem~\ref{bialgebra_to_manin}).
 
\subsection{Definition of Banach Lie bialgebras}\label{Definition of Banach Lie bialgebras}

Let us introduce the notion of Banach Lie bialgebras. We refer the reader to \cite{LW90} for the corresponding notion in the finite-dimensional case.

\begin{definition}\label{Bialgebra_def}
Let $\mathfrak{g}_+$ be a Banach Lie algebra over the field $\mathbb{K}\in\{\mathbb{R}, \mathbb{C}\}$, and consider a duality pairing $\langle\cdot,\cdot\rangle_{\mathfrak{g}_+, \mathfrak{g}_-}$ 
between $\mathfrak{g}_+$ and a Banach space $\mathfrak{g}_-$.  One says that $\mathfrak{g}_+$ is a \textbf{Banach Lie bialgebra with respect to} $\mathfrak{g}_-$ 
  if 
 \begin{enumerate}
 \item $\mathfrak{g}_+$ acts continuously by coadjoint action on $\mathfrak{g}_-\subset \mathfrak{g}_+^*$~; 
 \item there is given a Banach Lie algebra structure  on $\mathfrak{g}_-$ such that the dual map of the Lie bracket $[\cdot, \cdot]_{\mathfrak{g}_-}~:\mathfrak{g}_-\times\mathfrak{g}_-\rightarrow\mathfrak{g}_-$ restricts to
  a $1$-cocycle $\theta~: \mathfrak{g}_+\rightarrow \Lambda^2\mathfrak{g}_-^*$ with respect to the adjoint action $\ad^{(2,0)}$ of $\mathfrak{g}_+$ on  $\Lambda^2\mathfrak{g}_-^*$ (recall that $\mathfrak{g}_+$ can be viewed as a subspace of $\mathfrak{g}_-^*$).
   \end{enumerate} 
\end{definition}

\begin{remark}{\rm
A finite-dimensional Lie bialgebra $(\mathfrak{g}, \mathfrak{g}^*)$ (see Definition~1.7 in \cite{LW90})  is a Banach Lie bialgebra $\mathfrak{g}$ with respect to its dual space $\mathfrak{g}^*$,  where the duality pairing is the natural pairing between $\mathfrak{g}$ and $\mathfrak{g}^*$.
}
\end{remark}
\begin{remark}{\rm
\begin{enumerate}
\item The first condition in Definition~\ref{Bialgebra_def} means that
$\mathfrak{g}_-$ is preserved by the coadjoint action of $\mathfrak{g}_+,$ i.e $$\ad^*_x\mathfrak{g}_-\subset \mathfrak{g}_-\subset\mathfrak{g}_+^*$$ for any $x\in\mathfrak{g}_+,$ and that  the action map
 $$
 \begin{array}{cll}
 \mathfrak{g}_+\times\mathfrak{g}_-& \rightarrow & \mathfrak{g}_-\\ (x, \alpha) & \mapsto & \ad^*_x\alpha
 \end{array}
 $$
 is continuous. This condition is necessary in order to define the action of $\mathfrak{g}_+$ on the space $\Lambda^2\mathfrak{g}_-^*$ of continuous skew-symmetric maps on $\mathfrak{g}_-$ by \eqref{ad2}.
 \item 
The map $\theta$ is a $1$-cocycle on $\mathfrak{g}_+$ if it satisfies~:
\begin{equation*}\begin{array}{ll}
\theta\left([x, y]\right) & = \ad^{(2,0)}_x\left(\theta(y)\right) - \ad^{(2,0)}_y\left(\theta(x)\right)\\
\end{array}
\end{equation*}
where $x, y\in \mathfrak{g}_+$. 
The second condition in Definition~\ref{Bialgebra_def} means therefore that (see section~\ref{cocycle_section})
  \begin{equation}\label{cocycle}
\theta\left([x, y]\right)(\alpha, \beta)  = \theta(y)(\ad^*_x\alpha, \beta) + \theta(y)(\alpha, \ad^*_x\beta) -\theta(x)(\ad^*_y\alpha, \beta) - \theta(x)(\alpha, \ad^*_y\beta),
\end{equation}
 for any $x, y$ in $\mathfrak{g}_+$ and any $\alpha, \beta$ in $\mathfrak{g}_-$.
In a more explicite form, the cocycle condition reads
\begin{equation}\label{cocycle_mitte}
\begin{array}{ll}
\langle[x, y]_{\mathfrak{g}_+},[\alpha, \beta]_{\mathfrak{g}_-}\rangle_{\mathfrak{g}_+, \mathfrak{g}_-}
= &\langle y, [\ad^*_x\alpha, \beta]_{\mathfrak{g}_-}\rangle_{\mathfrak{g}_+, \mathfrak{g}_-} 
+ \langle y, [\alpha, \ad^*_x\beta]_{\mathfrak{g}_-}\rangle_{\mathfrak{g}_+, \mathfrak{g}_-} \\
&- \langle x, [\ad^*_y\alpha, \beta]_{\mathfrak{g}_-}\rangle_{\mathfrak{g}_+, \mathfrak{g}_-} 
- \langle x, [\alpha, \ad^*_y\beta]_{\mathfrak{g}_-} \rangle_{\mathfrak{g}_+, \mathfrak{g}_-} ,
\end{array}
\end{equation}
  for any $x, y$ in $\mathfrak{g}_+$ and any $\alpha, \beta$ in $\mathfrak{g}_-$.
  \item Let us remark that we do not assume that the cocycle $\theta$ takes values in the subspace $\Lambda^2\mathfrak{g}_+$ of $\Lambda^2\mathfrak{g}_-^{*}$. This is related to the generalized notion of Poisson manifolds given in Definition \ref{Poisson-Manifold}.  
  \end{enumerate}
 }
\end{remark}

Let us first give examples of Banach Lie algebras which are Banach Lie--Poisson spaces (see Section~\ref{LP})  but \textbf{not}  Banach Lie bialgebras.
\begin{example}\label{not_Banach_Lie_bialgebra}{\rm
For $1\leq p\leq \infty$, 
consider $\mathfrak{g}_+ :=\operatorname{L}_p(\mathcal{H})$ and its dual space $\mathfrak{g}_-:=\operatorname{L}_q(\mathcal{H})$, the duality pairing $\langle\cdot, \cdot\rangle_{\mathfrak{g}_+, \mathfrak{g}_-} $ being given by the trace.
By example \ref{exLpLq}, $\operatorname{L}_p(\mathcal{H})$ is a Banach Lie--Poisson space with respect to $\operatorname{L}_q(\mathcal{H})$. For $x\in \operatorname{L}_p(\mathcal{H})$ and $\alpha\in L_q(\mathcal{H})$, one has
$\ad^*_\alpha x = [x, \alpha]\in \operatorname{L}_p(\mathcal{H})$ and $\ad^*_x\alpha = [\alpha, x]\in \operatorname{L}_q(\mathcal{H})$. Therefore, for any $\alpha, \beta\in \operatorname{L}_q(\mathcal{H})$ and $x, y\in \operatorname{L}_p(\mathcal{H})$, 
one has
\begin{equation}\label{notbi}
\begin{array}{l}
\langle y, [\ad^*_x\alpha, \beta]_{\mathfrak{g}_-}\rangle_{\mathfrak{g}_+, \mathfrak{g}_-} 
+ \langle y, [\alpha, \ad^*_x\beta]_{\mathfrak{g}_-}\rangle_{\mathfrak{g}_+, \mathfrak{g}_-} 
- \langle x, [\ad^*_y\alpha, \beta]_{\mathfrak{g}_-}\rangle_{\mathfrak{g}_+, \mathfrak{g}_-} 
- \langle x, [\alpha, \ad^*_y\beta]_{\mathfrak{g}_-} \rangle_{\mathfrak{g}_+, \mathfrak{g}_-} \\
= 2~\langle\ad^*_{\alpha}[x, y], \beta \rangle_{\mathfrak{g}_+, \mathfrak{g}_-}.
\end{array}
\end{equation}
%
This implies that $\operatorname{L}_p(\mathcal{H})$ is not a  Banach Lie bialgebra with respect to $\operatorname{L}_q(\mathcal{H})$ (compare with the cocycle condition~\eqref{cocycle_mitte}).
}
\end{example}
\begin{example}{\rm
By example \ref{exL1Linfty}, $\operatorname{L}_1(\mathcal{H})$ is a Banach Lie--Poisson space with respect to $\operatorname{L}_\infty(\mathcal{H})$. A computation analoguous as in previous example shows that $\operatorname{L}_1(\mathcal{H})$ is not a Banach Lie bialgebra with respect to $\operatorname{L}_\infty(\mathcal{H})$.}
\end{example}

\begin{example}{\rm
By example \ref{exL1L2}, $\operatorname{L}_1(\mathcal{H})$ is a Banach Lie--Poisson space with respect to $\operatorname{L}_2(\mathcal{H})$. It is easy to see that equation~\eqref{notbi} is satisfied for any $\alpha, \beta\in \operatorname{L}_2(\mathcal{H})$ and $x, y\in \operatorname{L}_1(\mathcal{H})$, 
hence $\operatorname{L}_1(\mathcal{H})$ is not a Banach Lie bialgebra with respect to $\operatorname{L}_2(\mathcal{H})$.}
\end{example}
Let us now give examples of Banach Lie--Poisson spaces which are also Banach Lie bialgebras. In Example~\ref{lp-Bialgebra} and Example~\ref{Iwasawa_Banach_Lie_bialgebra}, the cocycle condition can be checked by hand using the expression of the coadjoint actions.
\begin{example}\label{lp-Bialgebra}{\rm \textit{Banach Lie bialgebra of upper and lower triangular operators.}
For $1<p<\infty$, consider the Banach algebra $\operatorname{L}_p(\mathcal{H})_-$  of lower triangular operators in $\operatorname{L}_p(\mathcal{H})$ defined by \eqref{triangularLp+} and its dual space $\operatorname{L}_{q}(\mathcal{H})_{+}$, where $\frac{1}{p}+\frac{1}{q} = 1$ and where  the duality pairing is given by the trace. Then $\operatorname{L}_p(\mathcal{H})_-$ is a Banach Lie bialgebra with respect to $\operatorname{L}_{q}(\mathcal{H})_{+}$.
}
\end{example}

\begin{example}\label{Iwasawa_Banach_Lie_bialgebra}{\rm 
\textit{Iwasawa Banach Lie bialgebras.}
Let $p$ and $q$ be such that $1<p<\infty$, $1<q<\infty$ and $\frac{1}{p}+\frac{1}{q} = 1$.
Consider the Banach Lie algebra $\mathfrak{u}_p(\mathcal{H})$ and its dual Banach space $\mathfrak{b}_q^+(\mathcal{H})$, endowed with its natural Banach Lie algebra structure, which makes $\mathfrak{u}_p(\mathcal{H})$ into a Banach Lie--Poisson space (see example~\ref{exIwasawa}). In this case the duality pairing is given by the imaginary part of the trace.  Then $\mathfrak{u}_p(\mathcal{H})$ is a Banach Lie bialgebra with respect to  $\mathfrak{b}^+_q(\mathcal{H})$.}
\end{example}

\subsection{Banach Lie bialgebras versus Manin triples}\label{Banach Lie bialgebras versus Manin triples}
In the finite-dimensional case, the notion of Lie bialgebra is equivalent to the notion of Manin triple (see \cite{Dr83} or section~1.6 in \cite{Ko04}). In the infinite-dimensional case the notion of Banach Lie--Poisson space comes into play.

\begin{theorem}\label{bialgebra_to_manin}
Consider two Banach Lie algebras 
 $\left(\mathfrak{g}_+, [\cdot,\cdot]_{\mathfrak{g}_+}\right)$ and  $\left(\mathfrak{g}_-, [\cdot,\cdot]_{\mathfrak{g}_-}\right)$ in duality.  Denote by $\mathfrak{g}$ the Banach space $\mathfrak{g} = \mathfrak{g}_+\oplus \mathfrak{g}_-$ with norm $\|\cdot\|_{\mathfrak{g}} = \|\cdot\|_{\mathfrak{g}_+}+\|\cdot\|_{\mathfrak{g}_-}$. The following assertions are equivalent.
\begin{itemize}
\item[(1)] $\mathfrak{g}_+$ is a Banach Lie--Poisson space and a Banach Lie bialgebra with respect to $\mathfrak{g}_-$ with cocycle $\theta_+ := [\cdot, \cdot]_{\mathfrak{g}_-}^*~:\mathfrak{g}_+ \rightarrow \Lambda^2\mathfrak{g}_-^*$;
\item[(2)] $(\mathfrak{g}, \mathfrak{g}_+, \mathfrak{g}_-)$ is a Manin triple for the natural non-degenerate symmetric bilinear map
$$
\begin{array}{lcll}
\langle\cdot,\cdot\rangle_{\mathfrak{g}}~: & \mathfrak{g}\times\mathfrak{g}&\rightarrow & \mathbb{K}\\
& (x,\alpha)\times (y, \beta) & \mapsto & \langle x, \beta\rangle_{\mathfrak{g}_+, \mathfrak{g}_-} + \langle y, \alpha\rangle_{\mathfrak{g}_+, \mathfrak{g}_-}
\end{array}
$$
with bracket given by
\begin{equation}\label{brak2}
\begin{array}{lcll}
[\cdot, \cdot]_{\mathfrak{g}}~:&\mathfrak{g}\times\mathfrak{g}&\rightarrow&\mathfrak{g}=\mathfrak{g}_+\oplus\mathfrak{g}_-\\
& (x, \alpha)\times(y, \beta) & \mapsto & \left([x, y]_{\mathfrak{g}_+}+\ad^*_{\beta}x -\ad^*_{\alpha}y,\quad [\alpha, \beta]_{\mathfrak{g}_-} + \ad^*_y\alpha - \ad^*_x\beta\right).
\end{array}
\end{equation}
\item[(3)] $\mathfrak{g}_-$ is a Banach Lie--Poisson space and a Banach Lie bialgebra with respect to $\mathfrak{g}_+$ with cocycle $\theta_- := [\cdot, \cdot]_{\mathfrak{g}_+}^*~:\mathfrak{g}_-\rightarrow \Lambda^2\mathfrak{g}^*_+$; 
\end{itemize}
\end{theorem}

\begin{proof}
$(2)\Rightarrow(1)$ follows from Theorem \ref{Manin_to_bialgebra}. Let us prove $(1)\Rightarrow(2)$. 
\begin{itemize}
\item
Since $\mathfrak{g}_+$ is a Banach Lie--Poisson space with respect to $\mathfrak{g}_-$,  $\mathfrak{g}_-$ is a Banach Lie algebra $(\mathfrak{g}_-, [\cdot,\cdot]_{\mathfrak{g}_-})$ such that the coadjoint action of $\mathfrak{g}_-$ on $\mathfrak{g}_-^*$ preserves the subspace $\mathfrak{g}_+\subset \mathfrak{g}_-^*$ and the map 
$$
\begin{array}{lcll}
\ad^*_{\mathfrak{g}_-}~: & \mathfrak{g}_-\times\mathfrak{g}_+&\rightarrow & \mathfrak{g}_+\\
& (\alpha, x)& \mapsto & \ad^*_\alpha x,
\end{array}
$$
is continuous. Since $\mathfrak{g}_+$ is a Banach Lie bialgebra, the coadjoint action of $\mathfrak{g}_+$ on $\mathfrak{g}_+^*$ preserves the subspace $\mathfrak{g}_-\subset \mathfrak{g}_+^*$ and the map 
$$
\begin{array}{lcll}
\ad^*_{\mathfrak{g}_+}~: & \mathfrak{g}_+\times\mathfrak{g}_-&\rightarrow & \mathfrak{g}_-\\
& (x, \alpha)& \mapsto & \ad^*_x \alpha,
\end{array}
$$
is continuous. Therefore  bracket \eqref{brak2} is continuous on $\mathfrak{g} = \mathfrak{g}_+\oplus\mathfrak{g}_-$.

\item Let us show that the symmetric non-degenerate pairing $\langle\cdot,\cdot\rangle_{\mathfrak{g}}$ is invariant with respect to the bracket $[\cdot, \cdot]_{\mathfrak{g}}$. For this, we will use the fact that $\mathfrak{g}_+$ and $\mathfrak{g}_-$ are isotropic subspaces for $\langle \cdot, \cdot\rangle_{\mathfrak{g}}$. For $x\in\mathfrak{g}_+$ and $\alpha\in\mathfrak{g}_-$, one has 
\begin{equation}\label{decbracket}
[x, \alpha]_{\mathfrak{g}} = (\ad^*_{\alpha}x, -\ad^*_x\alpha).
\end{equation}
Therefore, for any $x\in\mathfrak{g}_+$ and any $\alpha, \beta\in\mathfrak{g}_-$, one has
$$
\begin{array}{ll}
\langle [x, \alpha]_{\mathfrak{g}}, \beta\rangle_{\mathfrak{g}} &= \langle \ad^*_{\alpha}x, \beta\rangle_{\mathfrak{g}} = \langle x , \ad_{\alpha}\beta\rangle_{\mathfrak{g}} = \langle x, [\alpha, \beta]_{\mathfrak{g}}\rangle_{\mathfrak{g}} \\&= - \langle x, [\beta, \alpha]_{\mathfrak{g}}\rangle_{\mathfrak{g}} = -\langle \ad^*_{\beta}x, \alpha\rangle_{\mathfrak{g}} = \langle [\beta, x]_{\mathfrak{g}}, \alpha\rangle_{\mathfrak{g}}.
\end{array}
$$
Similarly, for any $x, y\in\mathfrak{g}_+$ and any $\beta\in\mathfrak{g}_-$, one has
$$
\begin{array}{ll}
\langle [x, y]_{\mathfrak{g}}, \beta\rangle_{\mathfrak{g}} &= \langle y,  \ad^*_{x}\beta\rangle_{\mathfrak{g}} = \langle y, [\beta, x]_{\mathfrak{g}}\rangle_{\mathfrak{g}} = 
-\langle \ad^*_y\beta, x\rangle_{\mathfrak{g}} = \langle [y, \beta]_{\mathfrak{g}}, x\rangle_{\mathfrak{g}}.
\end{array}
$$
By linearity, it follows that $\langle\cdot, \cdot\rangle_{\mathfrak{g}}$ is invariant with respect to $[\cdot, \cdot]_{\mathfrak{g}}$.

\item It remains to verify that $[\cdot, \cdot]_{\mathfrak{g}}$ satisfies the Jacobi identity. Let us first show that for any $x,y\in\mathfrak{g}_+$ and any $\alpha\in\mathfrak{g}_-$,
$$
[\alpha, [x, y]] = [[\alpha, x], y] + [x, [\alpha, y]].
$$
 The dual map $[\cdot, \cdot]_{\mathfrak{g}_-}^*~: \mathfrak{g}_-^*\rightarrow \Lambda^2\mathfrak{g}_-^*$  of the bilinear map $[\cdot, \cdot]_{\mathfrak{g}_-}~: \Lambda^2\mathfrak{g}_-\rightarrow \mathfrak{g}_-$ is 
 $$
 [\cdot, \cdot]^*_{\mathfrak{g}_-}(\mathcal{F}) = \mathcal{F}([\cdot, \cdot]_{\mathfrak{g}_-}).
 $$
 In particular, its restriction $\theta_+~: \mathfrak{g}_+\rightarrow \Lambda^2\mathfrak{g}_-^*$  to $\mathfrak{g}_+\subset \mathfrak{g}_-^*$ reads
 $$
 \theta(x)(\alpha, \beta) = \langle x, [\alpha, \beta]_{\mathfrak{g}_-}\rangle = \langle [x, \alpha]_\mathfrak{g}, \beta\rangle_\mathfrak{g} = \langle \ad^*_{\alpha}x, \beta\rangle_{\mathfrak{g}}.
 $$
Since $\mathfrak{g}_+$ is a Banach Lie--Poisson space, the cocycle $\theta_+ = [\cdot, \cdot]_{\mathfrak{g}_-}^*$ restricted to $\mathfrak{g}_+\subset \mathfrak{g}_-^*$ takes values in $\Lambda^2\mathfrak{g}_+ $. The cocycle condition \eqref{cocycle} reads
 \begin{equation}\label{cocycle_explicite}
\begin{array}{ll}
\langle [x, y], [\alpha, \beta]\rangle_{\mathfrak{g}_+, \mathfrak{g}_-}& = +\langle y, [\ad^*_x\alpha, \beta]\rangle_{\mathfrak{g}_+, \mathfrak{g}_-} +\langle y, [\alpha, \ad^*_x\beta]\rangle_{\mathfrak{g}_+, \mathfrak{g}_-} \\ & \quad-\langle x, [\ad^*_y\alpha, \beta]\rangle_{\mathfrak{g}_+, \mathfrak{g}_-} - \langle x, [\alpha, \ad^*_y\beta]\rangle_{\mathfrak{g}_+, \mathfrak{g}_-},
\end{array}
\end{equation}
where $x,y\in\mathfrak{g}_+$ and $\alpha, \beta\in\mathfrak{g}_-$. Using the definition of the bracket $\langle\cdot, \cdot\rangle_{\mathfrak{g}}$ and its invariance with respect to $[\cdot, \cdot]_{\mathfrak{g}}$, this is equivalent to
$$
\begin{array}{ll}
-\langle [\alpha, [x, y]], \beta\rangle_{\mathfrak{g}}& = -\langle [\ad^*_x\alpha, y], \beta\rangle_{\mathfrak{g}} -\langle [\alpha, y], \ad^*_x\beta\rangle_{\mathfrak{g}} \\ & \quad+\langle [\ad^*_y\alpha, x], \beta\rangle_{\mathfrak{g}} + \langle [\alpha, x],  \ad^*_y\beta\rangle_{\mathfrak{g}}.
\end{array}
$$
Using the fact that $\mathfrak{g}_+$ and $\mathfrak{g}_-$ are isotropic subspaces for $\langle\cdot, \cdot\rangle_{\mathfrak{g}}$ and relation \eqref{decbracket}, one gets
$$
\begin{array}{ll}
-\langle [\alpha, [x, y]], \beta\rangle_{\mathfrak{g}}& = -\langle [\ad^*_x\alpha, y], \beta\rangle_{\mathfrak{g}} +\langle \ad^*_\alpha y, \ad^*_x\beta\rangle_{\mathfrak{g}_+, \mathfrak{g}_-} \\ & \quad+\langle [\ad^*_y\alpha, x], \beta\rangle_{\mathfrak{g}}  -\langle \ad^*_\alpha x,  \ad^*_y\beta\rangle_{\mathfrak{g}_+, \mathfrak{g}_-} .
\end{array}
$$
Using the definition of the coadjoint actions, one obtains
$$
\begin{array}{ll}
-\langle [\alpha, [x, y]], \beta\rangle_{\mathfrak{g}}& = -\langle [\ad^*_x\alpha, y], \beta\rangle_{\mathfrak{g}} +\langle [x, \ad^*_\alpha y], \beta\rangle_{\mathfrak{g}_+, \mathfrak{g}_-} \\ & \quad+\langle [\ad^*_y\alpha, x], \beta\rangle_{\mathfrak{g}}  -\langle [y, \ad^*_\alpha x],  \beta\rangle_{\mathfrak{g}_+, \mathfrak{g}_-},
\end{array}
$$
or, in a more compact manner,
$$
\begin{array}{ll}
-\langle [\alpha, [x, y]], \beta\rangle_{\mathfrak{g}}& = \langle [\ad^*_{\alpha} x-\ad^*_x\alpha, y], \beta\rangle_{\mathfrak{g}} +\langle [x, \ad^*_\alpha y-\ad^*_y\alpha], \beta\rangle_{\mathfrak{g}}.
\end{array}
$$
Using  $[x, \alpha]_{\mathfrak{g}} = \ad^*_{\alpha}x -\ad^*_x\alpha$, and $[y, \alpha]_{\mathfrak{g}} = \ad^*_{\alpha}y -\ad^*_y\alpha$, one eventually gets
\begin{equation}\label{jacobi}
-\langle [\alpha, [x, y]], \beta\rangle_{\mathfrak{g}} = -\langle [[\alpha, x], y], \beta\rangle_{\mathfrak{g}} - \langle [x, [\alpha, y]], \beta\rangle_{\mathfrak{g}},
\end{equation}
for any $x,y\in\mathfrak{g}_+$ and any $\alpha, \beta\in\mathfrak{g}_-$.
Since $\langle\cdot, \cdot\rangle_{\mathfrak{g}}$ restricts to the duality pairing between $\mathfrak{g}_+$ and $\mathfrak{g}_-$, it follows that 
\begin{equation}\label{jacobi1+}
p_{\mathfrak{g}_+}[\alpha, [x, y]] = p_{\mathfrak{g}_+}[[\alpha, x], y] + p_{\mathfrak{g}_+}[x, [\alpha, y]],
\end{equation}
for any $x,y\in\mathfrak{g}_+$ and any $\alpha\in\mathfrak{g}_-$. On the other hand, considering the projection on $\mathfrak{g}_-$ one has
$$
p_{\mathfrak{g}_-}[\alpha, [x, y]] = \ad^*_{[x,y]}\alpha,
$$ as well as
$$
p_{\mathfrak{g}_-}[[\alpha, x], y]  = \ad^*_y\ad^*_x\alpha,
$$ 
and
$$
p_{\mathfrak{g}_-}[x, [\alpha, y]] = -\ad^*_x\ad^*_y\alpha.
$$
Since the bracket in $\mathfrak{g}_+$ satisfied Jacobi identity, it follows that
$$
\langle \alpha, [[x, y], z]\rangle_{\mathfrak{g}_+, \mathfrak{g}_-} =  \langle\alpha, [x, [y, z]]\rangle_{\mathfrak{g}_+, \mathfrak{g}_-} -\langle \alpha, [y, [x, z]]\rangle_{\mathfrak{g}_+, \mathfrak{g}_-},
$$
therefore
\begin{equation}\label{jacobi1-}
p_{\mathfrak{g}_-}[\alpha, [x, y]] = p_{\mathfrak{g}_-}[[\alpha, x], y] + p_{\mathfrak{g}_-}[x, [\alpha, y]],
\end{equation}
for any $x,y\in\mathfrak{g}_+$ and any $\alpha\in\mathfrak{g}_-$.
Combining \eqref{jacobi1+} and \eqref{jacobi1-}, it follows that
$$
[\alpha, [x, y]] = [[\alpha, x], y] + [x, [\alpha, y]],
$$
for any $x,y\in\mathfrak{g}_+$ and any $\alpha\in\mathfrak{g}_-$.

\item It remains to show that for any $x\in\mathfrak{g}_+$ and any $\alpha, \beta\in\mathfrak{g}_-$,
$$
[x, [\alpha, \beta]] = [[x, \alpha], \beta] + [\alpha, [x, \beta]].
$$
Since the bracket in $\mathfrak{g}_-$ satisfies Jacobi identity, similarly to \eqref{jacobi1-} remplacing $\mathfrak{g}_-$ by $\mathfrak{g}_+$, one has
\begin{equation}\label{jacobi2+}
p_{\mathfrak{g}_+}[x, [\alpha, \beta]] = p_{\mathfrak{g}_+}[[x, \alpha], \beta] + p_{\mathfrak{g}_+}[\alpha, [x, \beta]].
\end{equation}
Let us show that 
$$
p_{\mathfrak{g}_-}[x, [\alpha, \beta]] = p_{\mathfrak{g}_-}[[x, \alpha], \beta] + p_{\mathfrak{g}_-}[\alpha, [x, \beta]],
$$
for any $x\in\mathfrak{g}_+$ and any $\alpha, \beta\in\mathfrak{g}_-$. For any $x, y\in\mathfrak{g}_+$ and any $\alpha, \beta\in\mathfrak{g}_-$, one has
$$
\langle y, p_{\mathfrak{g}_-}[x, [\alpha, \beta]] \rangle_{\mathfrak{g}_+, \mathfrak{g}_-} = -\langle y, \ad^*_x[\alpha, \beta]\rangle_{\mathfrak{g}_+, \mathfrak{g}_-} = -\langle [x, y], [\alpha, \beta]\rangle_{\mathfrak{g}_+, \mathfrak{g}_-} = \langle [\alpha, [x, y]], \beta\rangle_{\mathfrak{g}}.
$$
On the other hand, for any $x, y\in\mathfrak{g}_+$ and any $\alpha, \beta\in\mathfrak{g}_-$, one has
$$
\langle y, p_{\mathfrak{g}_-}[[x, \alpha], \beta] \rangle_{\mathfrak{g}_+, \mathfrak{g}_-}  = \langle  y, [[x, \alpha], \beta] \rangle_{\mathfrak{g}} = \langle   [[\alpha, x], y],  \beta\rangle_{\mathfrak{g}},
$$
and
$$
\langle y, p_{\mathfrak{g}_-}[\alpha, [x, \beta]]  \rangle_{\mathfrak{g}_+, \mathfrak{g}_-}  = \langle y, [\alpha, [x, \beta]] \rangle_{\mathfrak{g}} = \langle [x, [\alpha, y]], \beta\rangle_{\mathfrak{g}}.
$$
By \eqref{jacobi}, it follows that 
\begin{equation}\label{jacobi2-}
p_{\mathfrak{g}_-}[x, [\alpha, \beta]] = p_{\mathfrak{g}_-}[[x, \alpha], \beta] + p_{\mathfrak{g}_-}[\alpha, [x, \beta]].
\end{equation}
Combining \eqref{jacobi2+} and \eqref{jacobi2-}, it follows that
$$
[x, [\alpha, \beta]] = [[x, \alpha], \beta] + [\alpha, [x, \beta]],
$$
for any $x\in\mathfrak{g}_+$ and any $\alpha, \beta\in\mathfrak{g}_-$. 
This ends the proof of $(1)\Rightarrow(2)$. The equivalence with $(3)$ follows by symmetry of $(2)$ with respect to exchange of $\mathfrak{g}_+$ into $\mathfrak{g}_-$.
\end{itemize}
\end{proof}

\begin{remark}
{\rm
It is noteworthy that the cocycle condition needs only to be verified for one of the Banach Lie algebra $\mathfrak{g}_+$ or $\mathfrak{g}_-$. The following Corollary is therefore a direct consequence of the proof of Theorem~\ref{bialgebra_to_manin}.
}
\end{remark}

\begin{corollary}\label{vice-versa}
Consider two Banach Lie algebras 
 $\left(\mathfrak{g}_+, [\cdot,\cdot]_{\mathfrak{g}_+}\right)$ and  $\left(\mathfrak{g}_-, [\cdot,\cdot]_{\mathfrak{g}_-}\right)$ in duality. If $\mathfrak{g}_+$ is a Banach Lie--Poisson space and a Banach Lie bialgebra with respect to $\mathfrak{g}_-$,  then $\mathfrak{g}_-$ is a Banach Lie--Poisson space and a Banach Lie bialgebra with respect to $\mathfrak{g}_+$.

\end{corollary}

\begin{example}{\rm
By Proposition~\ref{triples}, the triple $\left(\operatorname{L}_p(\mathcal{H}), \mathfrak{u}_p(\mathcal{H}), \mathfrak{b}^+_p(\mathcal{H})\right)$ is a Banach Manin triple for $1<p\leq 2$. Under this condition on $p$, it follows from Theorem~\ref{bialgebra_to_manin} that $\mathfrak{u}_p(\mathcal{H})$ is a Banach Lie--Poisson space and a Banach Lie bialgebra with respect to $\mathfrak{b}^+_p(\mathcal{H})$, and $\mathfrak{b}^+_p(\mathcal{H})$ is a Banach Lie--Poisson space and a Banach Lie bialgebra with respect to $\mathfrak{u}_p(\mathcal{H})$.}
\end{example}

\begin{example}\label{exUpBpManin}{\rm
For $1<p<\infty$, by Example~\ref{exIwasawa},  $\mathfrak{u}_p(\mathcal{H})$  is a Banach Lie--Poisson space with respect to  $\mathfrak{b}_{q}^+(\mathcal{H})$, where $\frac{1}{p}+\frac{1}{q}=1$.
By Example~\ref{Iwasawa_Banach_Lie_bialgebra}, $\mathfrak{u}_p(\mathcal{H})$ is a Banach Lie bialgebra with respect to $\mathfrak{b}^+_q(\mathcal{H})$. We deduce from Theorem~\ref{bialgebra_to_manin} that $\left(\mathfrak{u}_p(\mathcal{H})\oplus \mathfrak{b}_{q}^+(\mathcal{H}), \mathfrak{u}_p(\mathcal{H}), \mathfrak{b}_{q}^+(\mathcal{H})\right) $ form a Banach Manin triple, and that $\mathfrak{b}^+_q(\mathcal{H})$ is a Banach Lie bialgebra with respect to $\mathfrak{u}_p(\mathcal{H})$.
}
\end{example}

\begin{example}{\rm
From Example~\ref{exUpLow}, we know that $\operatorname{L}_p(\mathcal{H})_-$ is a Banach Lie--Poisson space with respect to $L_q(\mathcal{H})_+$.
By Example~\ref{lp-Bialgebra}, $\operatorname{L}_p(\mathcal{H})_-$ is a Banach Lie bialgebra with respect to $L_q(\mathcal{H})_+$.
By Theorem~\ref{bialgebra_to_manin}, the triple of Banach Lie algebras $\left(\operatorname{L}_p(\mathcal{H})_-\oplus L_q(\mathcal{H})_+, \operatorname{L}_p(\mathcal{H})_-, L_q(\mathcal{H})_+\right)$ is a Banach Manin triple.
By corollary~\ref{vice-versa}, $L_q(\mathcal{H})_+$ is a Banach Lie bialgebra with respect to $\operatorname{L}_p(\mathcal{H})_-$.
}\end{example}

\section{Banach Poisson--Lie groups}

This Section is devoted to the notion of Banach Poisson--Lie groups in the general framework of generalized Banach Poisson manifolds (see Section~\ref{def_Poisson}). We start in Section~\ref{section_Poisson_lie_group} with the definition, and show that the compatibility condition between the Poisson tensor and the multiplication on the group gives rise to a $1$-cocycle on the group. In Section~\ref{Jacobi_local}, we use the triviality of the tangent and cotangent bundles  in order to write the Jacobi identity for a Poisson tensor on a group at the level of the Lie algebra (Theorem~\ref{lemmaJacobi}). This allows us to give examples of Banach Poisson--Lie groups in Section~\ref{example_Up_Bp}. Finally, in Section~\ref{tangent_bialgebra}, we prove that the tangent space at the unit element $e$ of a Banach Poisson--Lie group $(G, \mathbb{F}, \pi)$ admits a natural Banach Lie bialgebra structure with respect to $\mathbb{F}_e$, and, in the case when the Poisson tensor $\pi$ is a section of $\Lambda^2TG$, is also a Banach Lie--Poisson space. The integrability problem of a Banach Lie bialgebra into a Banach Poisson--Lie group remains open.

\subsection{Definition of Banach Poisson--Lie groups}\label{section_Poisson_lie_group}
In order to be able to define the notion of Banach Poisson--Lie groups, we need to recall the construction of a Poisson structure on the product of two Poisson manifolds. The following Proposition is straightforward.

\begin{proposition}
Let $(M_1, \mathbb{F}_1, \pi_1)$ and $(M_2, \mathbb{F}_2, \pi_2)$ be two generalized Banach Poisson manifolds. Then the  product $M_1\times M_2$ carries a natural generalized Banach Poisson manifold structure $\left(M_1\times M_2, \mathbb{F}, \pi \right)$ where
\begin{enumerate}
\item $M_1\times M_2$ carries the product Banach manifold structure, in particular the tangent bundle of $M_1\times M_2$ is isomorphic to the direct sum $TM_1\oplus TM_2$ of the vector bundles $TM_1$ and $TM_2$ and the cotangent bundle  of $M_1\times M_2$ is isomorphic to $T^*M_1\oplus T^*M_2$,
\item $\mathbb{F}$ is the subbundle of $T^*M_1\oplus T^*M_2$ defined as 
$$
\mathbb{F}_{(p, q)} = (\mathbb{F}_1)_p\oplus(\mathbb{F}_2)_q, 
$$
\item $\pi$ is defined on $\mathbb{F}$ by
$$
\pi(\alpha_1+\alpha_2, \beta_1 + \beta_2) = \pi_1(\alpha_1, \beta_1) + \pi_2(\alpha_2, \beta_2), \quad \alpha_1, \beta_1 \in \mathbb{F}_1, \alpha_2, \beta_2 \in \mathbb{F}_2.
$$
\end{enumerate}
\end{proposition}

\begin{definition}\label{poisson_map_def}
Let $(M_1, \mathbb{F}_1, \pi_1)$ and $(M_2, \mathbb{F}_2, \pi_2)$ be two generalized Banach Poisson manifolds and $F~:M_1\rightarrow M_2$ a smooth map. One says that $F$ is a \textbf{Poisson map} at $p\in M_1$ if
\begin{enumerate}
\item the tangent map $T_pF~:T_pM_1 \rightarrow T_{F(p)}M_2$ satisfies $T_pF^*(\mathbb{F}_2)_{F(p)}\subset (\mathbb{F}_1)_p$, i.e. for any covector $\alpha \in (\mathbb{F}_2)_{F(p)}$, the covector $\alpha\circ T_{p}F$ belongs to $(\mathbb{F}_1)_p$~;
\item   $(\pi_1)_p\left(\alpha\circ T_pF, \beta\circ T_pF\right) = (\pi_2)_{F(p)}\left(\alpha, \beta\right)$ for any $\alpha, \beta \in (\mathbb{F}_2)_{F(p)}$.
\end{enumerate}
One says that $F$ is a Poisson map if it is a Poisson map at any $p\in M_1$.
\end{definition}
\begin{definition}
A \textbf{Banach Poisson--Lie group} $G$ is a Banach Lie group  equipped with a generalized Banach Poisson manifold structure such that the group multiplication $m: G \times G \rightarrow G$ is a Poisson map, where $G \times G$ is endowed with the product Poisson structure.
\end{definition}
The compatibility condition between the multiplication in the group and the Poisson tensor can be checked at the level of the Lie algebra. To see this, let us introduce some notation.
Denote by $L_g~:G\rightarrow G$ and $R_g~: G\rightarrow G$ the left and right translations  by $g\in G$. By abuse of notation, we will also denote by $L_g$ and $R_g$ the induced actions of $g\in G$ on the tangent bundle $TG$. The induced actions on the cotangent bundle $T^*G$ will be denoted by  $L_g^*$ and $R_g^*$, and on the dual $T^{**}G$ of the cotangent bundle  by $L_g^{**}$ and $R_g^{**}$. In particular, for $g\in G$ and $\alpha\in T_u^*G$, $L^*_g\alpha\in T^*_{g^{-1}u}G$ is defined by $L^*_g\alpha(X) := \alpha(L_g X)$. The smooth adjoint action of  $G$ on its Lie algebra $\mathfrak{g}$ will be denoted by $\Ad_g = L_g\circ R_g^{-1}$,  the induced smooth coadjoint action of $G$ on the dual space $\mathfrak{g}^*$ by $\Ad^*_g = L_g^*\circ R_{g^{-1}}^*$, and the induced smooth action of $G$ on the bidual space $\mathfrak{g}^{**}$ by $\Ad^{**}_g = L_g^{**}\circ R_{g^{-1}}^{**}$. For any subspace $\mathfrak{g}_-\subset \mathfrak{g}^*$ invariant under the coadjoint action of $G$, the restriction $$\Ad^*~: G \times \mathfrak{g}_- \rightarrow  \mathfrak{g}_-$$  which maps  $(g, \beta)\in G \times \mathfrak{g}_-$ to the element $\Ad^*(g)\beta\in \mathfrak{g}_-$
is continuous when $\mathfrak{g}_-$ is endowed with the norm of $\mathfrak{g}^*$. In that case,
one can 
define the coadjoint action $\Ad^{**}(g)$ of $g\in G$ on $\Lambda^2\mathfrak{g}_-^*$  by
$$
\Ad^{**}(g)\textrm{\textbf{t}}:= \textrm{\textbf{t}}(\Ad(g)^*\cdot, \Ad(g)^*\cdot), \quad\quad\textrm{ for } \textrm{\textbf{t}}\in \Lambda^2\mathfrak{g}_-^*.
$$
By abuse of notation, we will also denote by 
 $L_g^{**}$  the action of $g\in G$ on a section $\pi$ of $\Lambda^2T^{**}G$~:
 $$L_g^{**}\pi_u(\alpha, \beta) = \pi_u(L_g^*\alpha, L_g^*\beta), \quad\textrm{with}\quad \alpha, \beta \in T^*_{gu}G.$$ 
 Similarly, one defines $$R_u^{**}\pi_g(\alpha, \beta) = \pi_g(R_u^*\alpha, R_u^*\beta), \quad\textrm{with}\quad \alpha, \beta \in T^*_{gu}G.$$


\begin{proposition}\label{GPoisson}
A Banach Lie group $G$ equipped with a generalized Banach Poisson structure $(G, \mathbb{F}, \pi)$  is a Banach Poisson--Lie group if and only if 
\begin{enumerate}
\item $\mathbb{F}$ is invariant under left and right translations by elements in $G$
\item the Poisson tensor $\pi$ is a section of $\Lambda^2\mathbb{F}^*$ satisfying
\begin{equation}\label{Poisson_equation}
\pi_{gu} = L_g^{**}\pi_u + R_u^{**}\pi_g,\quad \forall g,u \in G.
\end{equation}
\end{enumerate}
\end{proposition}

\begin{proof}
 The tangent map $T_{(g,u)}m~:T_gG \oplus T_uG \rightarrow T_{gu}G$ to the multiplicatin $m$ in $G$ maps $(X_g, X_u)$ to $T_gR_u(X_g) + T_uL_g(X_u)$.
The invariance of $\mathbb{F}$ by left and right translations means that for any  $\alpha\in\mathbb{F}_{gu}$, the covector $\alpha\circ T_uL_g$ belongs to $\mathbb{F}_u\subset T_u^*G$ and the covector $\alpha\circ T_gR_u$ belongs to $\mathbb{F}_g \subset T_g^*G$. This is equivalent to the first condition in definition~\ref{poisson_map_def}. The second condition in definition~\ref{poisson_map_def} reads 
$$
 \pi_{G\times G}\left(\alpha\circ T_{(g,u)}m, \beta \circ T_{(g,u)}m\right)= \pi_{gu}(\alpha, \beta),
$$
for any $\alpha$ and $\beta$ in $\mathbb{F}_{gu}$.
By definition of the Poisson structure on the product manifold $G\times G$, one has~:
$$
 \pi_{G\times G}\left(\alpha\circ T_{(g,u)}m, \beta \circ T_{(g,u)}m\right)= \pi_u\left(\alpha\circ T_uL_g, \beta\circ T_uL_g\right) + \pi_g\left(\alpha\circ T_gR_u, \beta\circ T_gR_u\right),
 $$
  hence $m$ is a Poisson map if and only if \eqref{Poisson_equation} is satisfied.
\end{proof}

\begin{corollary}\label{vanishing_pi}
The Poisson tensor $\pi$ of a Banach Poisson--Lie group vanishes at the unit element.
\end{corollary}

\begin{proof}
By equation~\eqref{Poisson_equation}, one has $\pi_e = \pi_e + \pi_e$, hence $\pi_e = 0$.
\end{proof}

\begin{proposition}\label{cela2}
Let  $(G, \mathbb{F}, \pi)$ be a Banach Poisson--Lie group. Then the fiber $\mathbb{F}_e\subset T_e^*G$ over the unit element $e\in G$ is stable by the coadjoint action of $G$.
\end{proposition}

\begin{proof}
Suppose that  $(G, \mathbb{F}, \pi)$ is a Banach Poisson--Lie group. 
The invariance of $\mathbb{F}$ by left  translations implies that for any  $\alpha\in\mathbb{F}_{e}$ and any $g\in G$, the covector $L_{g}^*\alpha:= \alpha\circ T_{g^{-1}}L_g$ belongs to $\mathbb{F}_{g^{-1}}\subset T_{g^{-1}}^*G$. The invariance of $\mathbb{F}$ by right  translations then  implies that the covector $\Ad^*(g)\alpha := R_{g^{-1}}^*\circ L_g^*\alpha = \alpha\circ T_{g^{-1}}L_g\circ T_eR_{g^{-1}}$ belongs to $\mathbb{F}_e \subset T_e^*G$. Hence $\mathbb{F}_e$ is stable by the coadjoint action of $G$.
%
\end{proof}
In next Proposition, we introduce a $1$-cocycle naturally associated to a generalized Banach Poisson--Lie group (see Theorem~1.2 in \cite{LW90} for the finite-dimensional case).

\begin{proposition}\label{cela}
A Banach Lie group $G$ equipped with a generalized Banach Poisson structure $(G, \mathbb{F}, \pi)$  is a Banach Poisson--Lie group if and only if 
\begin{enumerate}
\item $\mathbb{F}$ is invariant under left and right translations by elements in $G$
\item  the map $\Pi_r~: G  \rightarrow  \Lambda^2\mathbb{F}_e^*$ defined by $ g \mapsto  \Pi_r(g) := R_{g^{-1}}^{**}\pi_g$ 
is a $1$-cocycle on $G$ with respect to the coadjoint action $\Ad^{**}$ of $G$ on $\Lambda^2\mathbb{F}_e^*$, i.e.  for any $g, u\in G$,
\begin{equation}\label{pirPoisson}\Pi_r(gu)  = \Ad(g)^{**}\Pi_r(u) + \Pi_r(g).\end{equation}
\end{enumerate}
\end{proposition}

\begin{proof} 
Using the relation $R^{**}_{(gu)^{-1}} = R^{**}_{g^{-1}}\circ R^{**}_{u^{-1}}$, the condition $\pi_{gu} = L_g^{**}\pi_u + R_u^{**}\pi_g$ for all $g,u \in G$ is equivalent to
$$
R_{(gu)^{-1}}^{**}\pi_{gu} = R^{**}_{g^{-1}}\circ R^{**}_{u^{-1}}\circ L_g^{**}\pi_u + R^{**}_{g^{-1}}\circ R^{**}_{u^{-1}}\circ R_u^{**}\pi_g.
$$
Since $R^{**}_{u^{-1}}$ and $L_g^{**}$ commutes, the previous equality simplifies to give
$$
\Pi_r(gu) = R^{**}_{g^{-1}}\circ L_g^{**} \Pi_r(u) + \Pi_r(g) = \Ad(g)^{**}\Pi_r(u) + \Pi_r(g),
$$
which is the cocycle condition (see Section~\ref{cocycle_section}).
%
%
\end{proof}

\subsection{Jacobi tensor and local sections}\label{Jacobi_local}

The following Lemma will be used in Section~\ref{example_Up_Bp} and Section \ref{Bres_section} in order to check the Jacobi identity for Poisson--Lie groups in the Banach setting.
\begin{lemma}\label{lemmaJacobi}
Let $\operatorname{G}$ be a Banach Lie group with  Lie algebra $\mathfrak{g}$, $\mathbb{F}$ a subbundle of $T^*G$ in duality with $TG$, invariant by left and right translations by elements in $\operatorname{G}$, and $\pi$ a smooth section of $\Lambda^2\mathbb{F}^*$. Then
\begin{enumerate}
\item Any closed local section $\alpha$ of $\mathbb{F}$ in a neighborhood $\mathcal{V}_g$ of $g\in\operatorname{G}$ is of the form $\alpha(u) = R_{u^{-1}}^*\alpha_0(u)$, where $\alpha_0~:\mathcal{V}_g \rightarrow \mathbb{F}_{e}\subset \mathfrak{g}^*$ satisfies~:
\begin{equation}\label{closednessalpha}
\langle \alpha_0(g), [X_0, Y_0]\rangle = \langle T_g\alpha_0(R_gY_0), X_0\rangle - \langle T_g\alpha_0(R_gX_0), Y_0\rangle,
\end{equation}
with $T_g\alpha_0~:T_g\operatorname{G}\rightarrow \mathfrak{g}^*$ the tangent map of $\alpha_0$ at $g\in\mathcal{V}_g$, and $X_0$, $Y_0$ any elements in $\mathfrak{g}$.
\item Let  $\Pi_r~:\operatorname{G}\rightarrow \Lambda^2\mathbb{F}_e^*$ be defined by 
$\Pi_r(g)~:= R_{g^{-1}}^{**}\pi(g)$. Then for any closed local sections $\alpha$, $\beta$ of  $\mathbb{F}$ around $g\in\operatorname{G}$, the differential $d\left(\pi(\alpha, \beta)\right)$ at $g$ reads
\begin{equation}\label{ca}
d\left(\pi(\alpha, \beta)\right)(X_g) = T_g\Pi_r(X_g)(\alpha_0(g), \beta_0(g)) + \Pi_r(g)(T_g\alpha_0(X_g), \beta_0(g)) + \Pi_r(g)(\alpha_0(g), T_g\beta_0(X_g)),
\end{equation}
where $X_g\in T_g\operatorname{G}$, $T_g\Pi_r~: T_g\operatorname{G}\rightarrow \Lambda^2\mathbb{F}_e^*$ is the tangent map of $\Pi_r$ at $g$, $\alpha = R_{g^{-1}}^*\alpha_0$ and $\beta = R_{g^{-1}}^*\beta_0$.
\item  Suppose that $i_{\alpha_0}\Pi_r(g)\in\mathfrak{g}\subset \mathbb{F}^*$ for any $\alpha= R_g^*\alpha_0\in\mathbb{F}$. Then for any closed local sections $\alpha$, $\beta$, $\gamma$ of $\mathbb{F}$,
\begin{equation}\label{Jacobi_tensor}
\begin{array}{l}
\!\!\!\!\!\!\!\!\!\pi\left(\alpha, d\left(\pi(\beta, \gamma)\right)\right) + \pi\left(\beta, d\left(\pi(\gamma, \alpha)\right)\right) + \pi\left(\gamma, d\left(\pi(\alpha, \beta)\right) \right)= \\
\quad\!\!
 T_g\Pi_r(R_g i_{\alpha_0}\Pi_r(g))(\beta_0(g), \gamma_0(g)) +\langle \alpha_0(g), [i_{\gamma_0(g)}\Pi_r(g), i_{\beta_0(g)}\Pi_r(g)]\rangle\\ 
 + T_g\Pi_r(R_g i_{\beta_0}\Pi_r(g))(\gamma_0(g), \alpha_0(g)) + \langle \beta_0(g), [i_{\alpha_0(g)}\Pi_r(g), i_{\gamma_0(g)}\Pi_r(g)]\rangle \\
 + T_g\Pi_r(R_g i_{\gamma_0}\Pi_r(g))(\alpha_0(g), \beta_0(g))+\langle \gamma_0(g), [i_{\beta_0(g)}\Pi_r(g), i_{\alpha_0(g)}\Pi_r(g)]\rangle  
\end{array}
\end{equation}
where $\alpha = R_{g^{-1}}^*\alpha_0$, $\beta = R_{g^{1-}}^*\beta_0$, and $\gamma = R_{g^{-1}}^*\gamma_0$. In particular the left hand side of equation \eqref{Jacobi_tensor} defines a tensor.
\end{enumerate}
\end{lemma}

\begin{proof}
\begin{enumerate}
\item Since $\alpha$ is closed, one has~:
$$
d\alpha(X, Y) = \mathcal{L}_X\alpha(Y) - \mathcal{L}_Y \alpha(X) - \alpha([X, Y]) = 0
$$
for any local vector fields $X$ and $Y$ around $g\in \mathcal{V}_g$. But since $d\alpha$ is a tensor (see Proposition~3.2, chapter~V in \cite{La01}), the previous identity depends only on the values of $X$ and $Y$ at $g$. In other words, $\alpha$ is closed if and only if the previous identity is satisfied for any right invariant vector fields $X$ and $Y$. Set $X_g = R_gX_0$ and $Y_g = R_gY_0$ for $X_0, Y_0\in\mathfrak{g}$. One has
$$
\begin{array}{ll}
d\alpha(X, Y)  & =  \mathcal{L}_X\alpha_0(g)(R_{g^{-1}}Y_g) -  \mathcal{L}_Y\alpha_0(g)(R_{g^{-1}}X_g) - \alpha_0(g)(R_{g^{-1}}[X, Y]_g)\\
 & =   \mathcal{L}_X\alpha_0(g)(Y_0) -  \mathcal{L}_Y\alpha_0(g)(X_0) + \alpha_0(g)([X_0, Y_0]_{\mathfrak{g}})
\end{array}
$$
Denote by $f~:\mathcal{V}_g\rightarrow \mathbb{R}$ the function defined by $f(g) = \alpha_0(g)(Y_0) = \langle \alpha_0(g), Y_0\rangle$, where the bracket denotes the natural pairing between $\mathfrak{g}^*$ and $\mathfrak{g}$. Then 
$$
df_g(X_g) = \langle T_g\alpha_0(R_gX_0), Y_0\rangle.
$$
It follows that 
$$\begin{array}{l}
d\alpha(X,Y)  = \langle T_g\alpha_0(R_gX_0), Y_0\rangle 
- \langle T_g\alpha_0(R_gY_0), X_0\rangle + \langle \alpha_0(g), [X_0, Y_0]_{\mathfrak{g}}\rangle.
\end{array}
$$
Therefore $d\alpha(X, Y) = 0$ for any $X$ and $Y$ if and only if 
$$
\langle \alpha_0(g), [X_0, Y_0]_{\mathfrak{g}}\rangle = \langle T_g\alpha_0(R_gY_0), X_0\rangle - \langle T_g\alpha_0(R_gX_0), Y_0\rangle,
$$
for any $X_0$ and $Y_0$ in $\mathfrak{g}$.
\item This is a straighforward application of the chain rule.
\item 
In the case where  $i_{\alpha_0}\Pi_r(g)$ belongs to $\mathfrak{g}$, one has the following expression of the differential of $\pi$~:
$$
d\left(\pi(\beta, \gamma)\right)(X_g) = T_g\Pi_r(X_g)(\beta_0(g), \gamma_0(g)) -\langle T_g\beta_0(X_g), i_{\gamma_0(g)}\Pi_r(g)\rangle + \langle T_g\gamma_0(X_g), i_{\beta_0(g)}\Pi_r(g)\rangle,
$$
where $\langle\cdot, \cdot\rangle$ denotes the duality pairing between $\mathfrak{g}^*$ and $\mathfrak{g}$. Therefore
$$
\begin{array}{lll}
\pi(\alpha, d\left(\pi(\beta, \gamma)\right) &= &\Pi_r(g)\left(\alpha_0(g), R_g^*d\left(\pi(\beta, \gamma)\right)\right)  = d\left(\pi(\beta, \gamma)\right) (R_g i_{\alpha_0(g)}\Pi_r(g))\\ &= &
T_g\Pi_r(R_g i_{\alpha_0(g)}\Pi_r(g))(\beta_0(g), \gamma_0(g)) \\&&-\langle T_g\beta_0(R_g i_{\alpha_0(g)}\Pi_r(g)), i_{\gamma_0(g)}\Pi_r(g)\rangle \\&&+ \langle T_g\gamma_0(R_g i_{\alpha_0(g)}\Pi_r(g)), i_{\beta_0(g)}\Pi_r(g)\rangle.
\end{array}
$$
It follows that
\begin{equation*}
\begin{array}{l}
\!\!\!\!\!\pi\left(\alpha, d\left(\pi(\beta, \gamma)\right)\right) + \pi\left(\beta, d\left(\pi(\gamma, \alpha)\right)\right) + \pi\left(\gamma, d\left(\pi(\alpha, \beta)\right) \right) \\
=-\langle T_g\beta_0(R_g i_{\alpha_0(g)}\Pi_r(g)), i_{\gamma_0(g)}\Pi_r(g)\rangle + \langle T_g\gamma_0(R_g i_{\alpha_0(g)}\Pi_r(g)), i_{\beta_0(g)}\Pi_r(g)\rangle\\
\quad-\langle T_g\gamma_0(R_g i_{\beta_0(g)}\Pi_r(g)), i_{\alpha_0(g)}\Pi_r(g)\rangle + \langle T_g\alpha_0(R_g i_{\beta_0(g)}\Pi_r(g)), i_{\gamma_0(g)}\Pi_r(g)\rangle\\
\quad-\langle T_g\alpha_0(R_g i_{\gamma_0(g)}\Pi_r(g)), i_{\beta_0(g)}\Pi_r(g)\rangle + \langle T_g\beta_0(R_g i_{\gamma_0(g)}\Pi_r(g)), i_{\alpha_0(g)}\Pi_r(g)\rangle\\
\quad+ T_g\Pi_r(R_g i_{\alpha_0}\Pi_r(g))\left(\beta_0(g), \gamma_0(g)\right) + T_g\Pi_r(R_g i_{\beta_0}\Pi_r(g))\left(\gamma_0(g), \alpha_0(g)\right) \\ \quad+ T_g\Pi_r(R_g i_{\gamma_0}\Pi_r(g))\left(\alpha_0(g), \beta_0(g)\right) \\

\end{array}
\end{equation*}
Using \eqref{closednessalpha}, the previous equation simplifies to \eqref{Jacobi_tensor}.
\end{enumerate}
\end{proof}

\subsection{Example of Banach Poisson--Lie groups $\operatorname{U}_{p}(\mathcal{H})$ and $\operatorname{B}_{p}^\pm(\mathcal{H})$ for $1<p\leq 2$}\label{example_Up_Bp}

Let us give some examples of Banach Poisson--Lie groups. We will need to introduce classical Banach Lie groups of operators.

\subsubsection*{General linear group $\operatorname{GL}(\mathcal{H})$.}
The general linear group of $\mathcal{H}$, denoted by $\operatorname{GL}(\mathcal{H})$  is the group consisting of bounded operators $A$ on $\mathcal{H}$ which admit a bounded inverse, i.e. for which there exists a bounded operator $A^{-1}$ satisfying $A A^{-1} = A^{-1} A = \textrm{Id}$, where $\textrm{Id}~:\mathcal{H}\rightarrow \mathcal{H}$ denotes the identity operator $x\mapsto x$.

\subsubsection*{General linear group $\operatorname{GL}_p(\mathcal{H})$.}
The Banach Lie algebra $\operatorname{L}_{p}(\mathcal{H})$ is the Banach Lie algebra of the following Banach Lie group~:
\begin{equation}\label{GL2}
\operatorname{GL}_{p}(\mathcal{H}) := \operatorname{GL}(\mathcal{H})\cap \{\textrm{Id} + A~: A \in \operatorname{L}_{p}(\mathcal{H})\}.
\end{equation}

\subsubsection*{Unitary group $\operatorname{U}(\mathcal{H})$.}The unitary group of $\mathcal{H}$ is defined as the subgroup of $\operatorname{GL}(\mathcal{H})$ consisting of operators $A$ such that $A^{-1} = A^*$ and is denoted by $\operatorname{U}(\mathcal{H})$.

\subsubsection*{Unitary groups $\operatorname{U}_p(\mathcal{H})$.}
The Banach Lie algebra $\mathfrak{u}_{p}(\mathcal{H})$ defined by  \eqref{u2} is the Banach Lie algebra of the following Banach Lie group
\begin{equation}\label{U2}
\operatorname{U}_{p}(\mathcal{H}) := \operatorname{U}(\mathcal{H})\cap \{\textrm{Id} + A~: A \in \operatorname{L}_{p}(\mathcal{H})\}.
\end{equation}

\subsubsection*{Triangular groups $\operatorname{B}^\pm_p(\mathcal{H})$.}
To the Banach Lie algebra $\mathfrak{b}^\pm_{p}(\mathcal{H})$ defined by \eqref{b2pm} is associated the following Banach Lie group~:
$$
\operatorname{B}_{ p}^{\pm}(\mathcal{H}) := \{\alpha\in \operatorname{GL}(\mathcal{H})\cap \left(\textrm{Id}+\mathfrak{b}^\pm_{p}(\mathcal{H})\right)~: \alpha^{-1}\in \textrm{Id} +\mathfrak{b}^\pm_{p}(\mathcal{H}) ~\textrm{and}~\langle n|\alpha|n\rangle\in\mathbb{R}^{+*}, \textrm{for}~ n\in\mathbb{Z}\},
$$
where $\mathbb{R}^{+*}$ is the group of strictly positive real numbers.

Let us now give some examples of Banach Poisson--Lie groups.
Similar results will be proved in the more involved restricted case in Section~\ref{Bres_section}. 
Recall that the orthogonal projections $p_{\mathfrak{u}_p,+}$ and $p_{\mathfrak{u}_p,-}$ are defined by \eqref{projectionu+} and \eqref{projectionu-} respectively. 
\begin{proposition}\label{Bp}
For $1<p\leq 2$, consider the Banach Lie group $\operatorname{B}_{p}^\pm(\mathcal{H})$ with Banach Lie algebra
$\mathfrak{b}_p^\pm(\mathcal{H})$, and the duality pairing $\langle\cdot,\cdot\rangle_{\mathbb{R}}~: \mathfrak{b}_p^\pm(\mathcal{H})\times \mathfrak{u}_p(\mathcal{H})\rightarrow \mathbb{R}$ given by the imaginary part of the trace \eqref{imparttrace}.
Consider
\begin{enumerate}
\item $\mathbb{B}_b:= R_{b^{-1}}^*\mathfrak{u}_p(\mathcal{H}) \subset T^*_b\operatorname{B}_{p}^\pm(\mathcal{H})$, $b\in \operatorname{B}_{p}^\pm(\mathcal{H})$.
\item $\Pi^{\operatorname{B}_{p}^\pm}_r~:\operatorname{B}_{p}^\pm(\mathcal{H})\rightarrow \Lambda^2\mathfrak{u}_p(\mathcal{H})^*$ defined by 
\begin{equation}\label{PiBp}
\Pi^{\operatorname{B}_{p}^\pm}_r(b)(x_1, x_2) := \Im\Tr p_{\mathfrak{b}_p^\pm}(b^{-1}x_1b)\left[ p_{\mathfrak{u}_p,\pm}(b^{-1}x_2 b) \right], 
\end{equation}
where $b\in \operatorname{B}_{p}^\pm(\mathcal{H})$ and $x_1, x_2\in \mathfrak{u}_p(\mathcal{H})$.
\item $\pi^{\operatorname{B}_{p}^\pm}~: \operatorname{B}_{p}^{\pm} \rightarrow \Lambda^2T\operatorname{B}_{p}^\pm(\mathcal{H})$ given by  $\pi^{\operatorname{B}_{p}^\pm}(b) := R_{b}^{**}\Pi^{\operatorname{B}_{p}^\pm}_r(b)$.
\end{enumerate}
Then $(\operatorname{B}_{p}^\pm(\mathcal{H}), \mathbb{B}, \pi^{\operatorname{B}_{p}^\pm})$ is a Banach Poisson--Lie group.
\end{proposition}

\begin{proof}
The expression of the Poisson tensor makes sense because $\operatorname{L}_p(\mathcal{H}) \subset \operatorname{L}_q(\mathcal{H})$ for $1<p\leq 2$ with $\frac{1}{p}+\frac{1}{q} = 1$. The Jacobi identity is a consequence of equation~\eqref{Jacobi_tensor}. 
The compatibility of the Poisson tensor and the multiplication of the group can be checked using equation~\eqref{pirPoisson}.
\end{proof}

Similarly one has~:
\begin{proposition}\label{Up}
For $1<p\leq 2$, consider the Banach Lie group  $ \operatorname{U}_{p}(\mathcal{H})$ with Banach Lie algebra $\mathfrak{u}_p(\mathcal{H})$ and the duality pairing $\langle\cdot,\cdot\rangle_{\mathbb{R}}~: \mathfrak{b}_p^\pm(\mathcal{H})\times \mathfrak{u}_p(\mathcal{H})\rightarrow \mathbb{R}$ given by the imaginary part of the trace \eqref{imparttrace}. Consider
\begin{enumerate}
\item
$\mathbb{U}_u:= R_{u^{-1}}^*\mathfrak{b}_p^\pm(\mathcal{H})\subset T_u^*\operatorname{U}_{p}(\mathcal{H})$, $u\in \operatorname{U}_{p}(\mathcal{H})$,
\item $\Pi^{ \operatorname{U}_{p}^{\pm}}_r~: \operatorname{U}_{p}(\mathcal{H})\rightarrow \Lambda^2\mathfrak{b}_p^\pm(\mathcal{H})^*$ defined by
\begin{equation}\label{PiUp}
\Pi^{ \operatorname{U}_{p}^{\pm}}_r(u)(b_1, b_2) := \Im\Tr p_{\mathfrak{u}_p,\pm}(u^{-1}b_1u)\left[ p_{\mathfrak{b}_p^\pm}(u^{-1}b_2 u)\right],
\end{equation}
where $u \in \operatorname{U}_{p}(\mathcal{H})$ and $b_1, b_2\in \mathfrak{b}_p^\pm(\mathcal{H})$.
\item $\pi^{\operatorname{U}_{p}^\pm}~: \operatorname{U}_{p}(\mathcal{H})\rightarrow \Lambda^2 T\operatorname{U}_{p}(\mathcal{H})$ given by $\pi^{\operatorname{U}_{p}^\pm}(g) := R_{g}^{**}\Pi^{ \operatorname{U}_{p}^{\pm}}_r(g)$.
\end{enumerate}
Then $(\operatorname{U}_{p}(\mathcal{H}), \mathbb{U}, \pi^{\operatorname{U}_{p}^\pm})$ is a Banach Poisson--Lie group. 
\end{proposition}

%


\subsection{The tangent Banach Lie bialgebra of a Banach Poisson--Lie group}\label{tangent_bialgebra}

In this Section, we show that the Banach Lie algebra $\mathfrak{g}$ of any Banach Poisson--Lie group $(G, \mathbb{F}, \pi)$ carries an natural Banach Lie bialgebra structure with respect to $\mathbb{F}_e$ (see Theorem~\ref{PL_bi} below). Moreover, when the Poisson tensor is a section of $\Lambda^2TG\subset \Lambda^2T^{**}G$, then $\mathfrak{g}$ is a Banach Lie--Poisson space with respect to $\mathbb{F}_e$ (see Theorem~\ref{PL_LP}).

\begin{theorem}\label{PL_bi}
Let $(G_+, \mathbb{F}, \pi)$ be  a Banach Poisson--Lie group and suppose that $\mathfrak{g}_- : = \mathbb{F}_e$ is a Banach subspace of $\mathfrak{g}_+^*$. Then $\mathfrak{g}_+$ is a Banach Lie bialgebra with respect to $\mathfrak{g}_-$. The Lie bracket in $\mathfrak{g}_-$ is given by 
\begin{equation}\label{def_bracket}
[\alpha_1, \beta_1]_{\mathfrak{g}_-} := T_e\Pi_r(\cdot)(\alpha_1, \beta_1)\in \mathfrak{g}_-\subset \mathfrak{g}_+^*,\quad  \alpha_1, \beta_1 \in \mathfrak{g}_-\subset \mathfrak{g}_+^*, 
\end{equation}
where $\Pi_r~:= R_{g^{-1}}^{**}\pi~: G_+ \rightarrow \Lambda^2\mathfrak{g}_-^*$, and  $T_e\Pi_r :\mathfrak{g}_+\rightarrow \Lambda^2\mathfrak{g}_-^*$ denotes the differential of $\Pi_r$ at the unit element $e\in G_+$.
%
%
%
%
%
\end{theorem}

\begin{proof}
\begin{itemize}
\item Let us show that  the dual map $T_e\Pi_r^*~: \left(\Lambda^2\mathfrak{g}_-^*\right)^* \rightarrow \mathfrak{g}_+^*$ defines a skew-symmetric bilinear map $[\cdot, \cdot]_{\mathfrak{g}_-}$ on $\mathfrak{g}_-$ with values in $\mathfrak{g}_-\subset \mathfrak{g}_+^*$.
 Let $\alpha$ and $\beta$ be any local sections of $\mathbb{F}$ in a neighboorhood $\mathcal{V}_e$ of the unit element $e\in G_+$. Define $\alpha_0~:\mathcal{V}_e\rightarrow \mathfrak{g}_-$ and $\beta_0~:\mathcal{V}_e\rightarrow \mathfrak{g}_-$ by  $\alpha_0(u) := R_{u}^*\alpha(u)$ and  $\beta_0(u) := R_{u}^*\beta(u)$. It follows from equation~\eqref{ca}, that for any $X\in\mathfrak{g}_+$,
$$
d_e\left(\pi(\alpha, \beta)\right)(X) = T_e\Pi_r(X)(\alpha_0(e), \beta_0(e)) + \Pi_r(e)(T_e\alpha_0(X), \beta_0(e)) + \Pi_r(e)(\alpha_0(g), T_g\beta_0(X_g)).
$$
By Corollary~\ref{vanishing_pi}, $\Pi_r(e) = 0$. Hence
\begin{equation}\label{de}
d_e\left(\pi(\alpha, \beta)\right)(X) = T_e\Pi_r(X)(\alpha_0(e), \beta_0(e)).
\end{equation}
By the first condition in the definition of a Poisson tensor, $d\left(\pi(\alpha, \beta)\right)$ is a local section of $\mathbb{F}$, therefore $d_e\left(\pi(\alpha, \beta)\right)$ belongs to $\mathbb{F}_e = \mathfrak{g}_-$.
%
It follows that the formula
$$[\alpha_1, \beta_1]_{\mathfrak{g}_-} := T_e\Pi_r(\cdot)(\alpha_1, \beta_1)$$
defines a bracket on $\mathfrak{g}_-$. The skew-symmetry of $[\cdot, \cdot]_{\mathfrak{g}_-}$ is clear.
\item Let us show that $[\cdot, \cdot]_{\mathfrak{g}_-}$ satisfies the Jacobi identity, hence is a Lie algebra structure on $\mathfrak{g}_-$. Consider any closed local sections $\alpha, \beta, \gamma$ of $\mathbb{F}$ defined in a neighborhood of $e\in G_+$. Since $\pi$ is a Poisson tensor, one has
$$\pi\left(\alpha, d\left(\pi(\beta, \gamma)\right)\right) + \pi\left(\beta, d\left(\pi(\gamma, \alpha)\right)\right) + \pi\left(\gamma, d\left(\pi(\alpha, \beta)\right) \right)= 0.$$
Differentiating the above identity at $e\in G_+$, 
one gets
\begin{equation}\label{ji}
d_e\left(\pi\left(\alpha, d\pi(\beta, \gamma)\right)\right) + d_e\left(\pi\left(\beta, d\pi(\gamma, \alpha)\right)\right)+ d_e\left(\pi\left(\gamma, d\pi(\alpha, \beta)\right)\right) = 0.
\end{equation}
Define $\alpha_0(u)~:= R_{u}^*\alpha(u)$, and $\delta_0(u)~:= R_{u}^*d_u\left(\pi(\beta, \gamma)\right)$. Note that $\alpha_0(e) = \alpha(e)$ and $\delta_0(e) = d_e\left(\pi(\beta, \gamma)\right)$. 
Hence, by equation~\eqref{de} and \eqref{def_bracket}, for any $X\in\mathfrak{g}_+$,
$$
\begin{array}{ll}
d_e\left(\pi\left(\alpha, d\pi(\beta, \gamma)\right)\right)(X) &= T_e\Pi_r(X)(\alpha_0(e), \delta_0(e)) = T_e\Pi_r(X)(\alpha(e), d_e\left(\pi(\beta, \gamma)\right))\\
& = T_e\Pi_r(X)\left(\alpha(e), T_e\Pi_r(\cdot)(\beta(e), \gamma(e)\right)\\
& = [\alpha(e), [\beta(e), \gamma(e)]_{\mathfrak{g}_-} ]_{\mathfrak{g}_-} (X).
\end{array}
$$
It follows that equation~\eqref{ji} can be rewritten as
$$
 [\alpha(e), [\beta(e), \gamma(e)]_{\mathfrak{g}_-} ]_{\mathfrak{g}_-}  +  [\beta(e), [\gamma(e), \alpha(e)]_{\mathfrak{g}_-} ]_{\mathfrak{g}_-}  +  [\gamma(e), [\alpha(e), \beta(e)]_{\mathfrak{g}_-} ]_{\mathfrak{g}_-}  = 0.
$$
To show that the bracket $[\cdot, \cdot]_{\mathfrak{g}_-} $ satisfies Jacobi identity, it remains to prove that any element $\alpha_1\in\mathfrak{g}_-$ can be extended to a closed local section $\alpha$ of $\mathbb{F}$ such that $\alpha(e) = \alpha_1$. For this, it suffices to find a scalar function $f$ defined in a neighborhood $\mathcal{V}_e$  of $e\in G_+$ such that $d_gf \in \mathbb{F}_g$ for any $g\in \mathcal{V}_e$ and $d_ef = \alpha_1.$ This can be done using a chart around $e\in G_+$ and a local trivialisation of $\mathbb{F}$. Then $\alpha~:= df$ is a closed local section of $\mathbb{F}$ such that $\alpha_1 = \alpha(e) $.
\item Let us show that $\mathfrak{g}_+$ acts continuously on $\mathfrak{g}_-\subset \mathfrak{g}_+^*$ by coadjoint action. By Proposition~\ref{cela2}, $\mathfrak{g}_-:=\mathbb{F}_e$ is invariant by the coadjoint action of $G_+$ on $\mathfrak{g}_+^*$. By differentiation, $\mathfrak{g}_-\subset \mathfrak{g}_+^*$ is stable by the coadjoint action of $\mathfrak{g}_+$ on $\mathfrak{g}_+^*$. This action is continuous when $\mathfrak{g}_-$ is endowed with the topology of $\mathfrak{g}_+^*$.
\item Let us show that the dual map of the bracket $[\cdot, \cdot]_{\mathfrak{g}_-}$ restricts to a $1$-cocycle $\theta~:\mathfrak{g}_+\rightarrow \Lambda^2\mathfrak{g}_-^*$ with respect to the adjoint action $\ad^{(2,0)}$ of $\mathfrak{g}_+$ on $\Lambda^2\mathfrak{g}_-^*$. By definition of the bracket \eqref{def_bracket}, $\theta = T_e\Pi_r$. By Proposition~\ref{cela}, $\Pi_r$ is a $1$-cocycle on $G_+$ with respect to the coadjoint action $\Ad^{**}$ of $G_+$ on $\Lambda^2\mathfrak{g}_-^*$. Hence $T_e\Pi_r$ is a $1$-cocycle on $\mathfrak{g}_+$ with respect to the adjoint action $\ad^{(2,0)}$ of $\mathfrak{g}_+$ on $\Lambda^2\mathfrak{g}_-^*$ (see Section~\ref{cocycle_section}). 
\end{itemize}
\end{proof}

\begin{example}{\rm
The tangent bialgebras of the Banach Poisson--Lie groups $\operatorname{B}_{p}^\pm(\mathcal{H})$ and $ \operatorname{U}_{p}(\mathcal{H})$ defined in Proposition~\ref{Bp} and Proposition~\ref{Up}, are the Banach Lie bialgebra $\mathfrak{b}_p^\pm(\mathcal{H})$ and $\mathfrak{u}_p(\mathcal{H})$ in duality, which combine into the Manin triple $(L_p(\mathcal{H}), \mathfrak{u}_{p}(\mathcal{H}), \mathfrak{b}^\pm_p(\mathcal{H}))$ given in Proposition~\ref{triples}.
Indeed, the derivative  at the unit element $e$  of $\Pi^{ \operatorname{U}_{p}^{\pm}}_r~:  \operatorname{U}_{p}(\mathcal{H})\rightarrow \Lambda^2\mathfrak{b}_p^\pm(\mathcal{H})^*$ defined by equation~\eqref{PiUp} reads~:
$$
\begin{array}{rl}
d_e\Pi^{ \operatorname{U}_{p}^{\pm}}_r(x)(b_1, b_2) = & \Im\Tr \left(p_{\mathfrak{u}_p,\pm}([x, b_1]) p_{\mathfrak{b}_p^\pm}(b_2)\right) +  \Im\Tr \left(p_{\mathfrak{u}_p,\pm}(b_1) p_{\mathfrak{b}_2^\pm}([x, b_2])\right),\\
= &  \Im\Tr \left(p_{\mathfrak{u}_p,\pm}([x, b_1]) b_2\right) = \Im \Tr [x, b_1] b_2 = \Im \Tr x[b_1, b_2]_{\mathfrak{b}_p^\pm},
\end{array}
$$
with $x\in \mathfrak{u}_{p}(\mathcal{H})$ and $b_1, b_2\in \mathfrak{b}^\pm_p(\mathcal{H})$, where we have used that $\mathfrak{b}_p^\pm(\mathcal{H})$ is an isotropic subspace.
It follows that $d_e\Pi^{ \operatorname{U}_{p}^{\pm}}_r(\cdot)(b_1, b_2) = [b_1, b_2]_{\mathfrak{b}_p^\pm} \in \mathfrak{b}_p^\pm(\mathcal{H}) \subset \mathfrak{u}_p(\mathcal{H})^*$.
Similarly, the derivative of $\Pi^{\operatorname{B}_{p}^\pm}_r$ defined by equation~\eqref{PiBp} is given by 
$$
\begin{array}{rl}
d_e\Pi^{\operatorname{B}_{p}^\pm}_r(b)(x_1, x_2) = \Im \Tr b[x_1, x_2]_{\mathfrak{u}_p},\quad b\in \mathfrak{b}^\pm_p(\mathcal{H}), x_1, x_2\in \mathfrak{u}_{p}(\mathcal{H}),
\end{array}
$$
and is the dual map of the bracket $[\cdot, \cdot]_{\mathfrak{u}_p}$.}
\end{example}

\begin{theorem}\label{PL_LP}
Let $(G_+, \mathbb{F}, \pi)$ be  a Banach Poisson--Lie group. If the map $\pi^{\sharp}~: \mathbb{F} \rightarrow \mathbb{F}^*$ defined by $\pi^{\sharp}(\alpha) := \pi(\alpha, \cdot)$ takes values in  $TG_+\subset \mathbb{F}^*$, then $\mathfrak{g}_+$ is a Banach Lie--Poisson space with respect to $\mathfrak{g}_-:=  \mathbb{F}_e$.
\end{theorem}

\begin{proof}
Let $\alpha_1 \in \mathfrak{g}_-$ and define $\alpha(g) = R^*_{g^{-1}}(\alpha_1)\in \mathbb{F}_g$. Then 
$\pi(R^*_{g^{-1}}\alpha_1, \cdot) = \pi(\alpha, \cdot) $ takes values in $T_gG_+\subset  \mathbb{F}_g^*$, and
$\Pi_r(g)(\alpha_1, \cdot) =  \pi(R^*_{g^{-1}}\alpha_1, R^*_{g^{-1}}\cdot)$ takes values in  $\mathfrak{g}_+\subset \mathfrak{g}_-^*$. It follows that $\Pi_r$ takes values in $\Lambda^2\mathfrak{g}_+\subset \Lambda^2\mathfrak{g}_-^*$. By differentiation, it follows that $T_e\Pi_r$ takes also values in $\Lambda^2\mathfrak{g}_+$. Using equation~\eqref{def_bracket} for the bracket in $\mathfrak{g}_-$, one has 
\begin{equation}
\langle \ad^*_{\alpha_1}X,  \beta_1\rangle_{\mathfrak{g}_+, \mathfrak{g}_-} := \langle X, [\alpha_1, \beta_1]_{\mathfrak{g}_-}\rangle_{\mathfrak{g}_+, \mathfrak{g}_-} = T_e\Pi_r(X)(\alpha_1, \beta_1).
\end{equation}
where $X \in \mathfrak{g}_+$ and $\alpha_1, \beta_1\in \mathfrak{g}_-$. Hence $\ad^*_{\alpha_1}X = T_e\Pi_r(X)(\alpha_1, \cdot)$, therefore $\ad^*_{\alpha_1}X$ belongs to $\mathfrak{g}_+$ for any $\alpha_1\in\mathfrak{g}_-$. Consequently $\mathfrak{g}_+$ is a Banach Lie--Poisson space with respect to $\mathfrak{g}_-$.
\end{proof}

\begin{remark}{\rm
In the finite-dimensional case, any Lie bialgebra can be integrated to a connected simply-connected Poisson--Lie group. The Banach situation is more complicated, since not every Banach Lie algebra can be integrated into a Banach Lie group (see \cite{Nee15} for a survey on the problem of integrability of Banach Lie algebras and on Lie theory in the more general framework of locally convex spaces). 
Even in the case when a Banach Lie bialgebra is a Lie algebra of a connected and simply-connect Banach Lie group, it is still an open problem to determine if the bialgebra structure can be integrated into a Poisson--Lie group structure on the group.
}
\end{remark}

\part{Poisson--Lie groups and the restricted Grassmannian}
In this Part we use the notions introduced in Part~1 in order to construct  Banach Poisson--Lie group structures on the restricted unitary group $\operatorname{U}_{\res}(\mathcal{H})$ and on the triangular group $\operatorname{B}_{\res}^{\pm}(\mathcal{H})$, and a generalized Banach Poisson manifold structure on the restricted Grassmannian such that both actions of $\operatorname{U}_{\res}(\mathcal{H})$ and $\operatorname{B}_{\res}^{\pm}(\mathcal{H})$ on the restricted Grassmannian are Poisson. 

In Section~6, we set the notation. In Section~7.1, we introduce  weak duality pairings between the Banach Lie algebras $\mathfrak{u}_{\res}(\mathcal{H})$ and $\mathfrak{b}_{1,2}^\pm(\mathcal{H})$, and between $\mathfrak{b}_{\res}^\pm(\mathcal{H})$ and $\mathfrak{u}_{1,2}(\mathcal{H})$. In Section~7.2 we use the unboundedness of the triangular truncation on the space of trace class operators to show that $\mathfrak{b}_{1,2}^\pm(\mathcal{H})$ is not a Banach Lie--Poisson space with respect to $\mathfrak{u}_{\res}(\mathcal{H})$. Similarly $\mathfrak{u}_{1,2}(\mathcal{H})$ is not a Banach Lie--Poisson space with respect to $\mathfrak{b}_{\res}^\pm(\mathcal{H})$. This implies in particular that there is no Banach Poisson--Lie group structure on $\operatorname{B}_{\res}^\pm(\mathcal{H})$ defined on the translation invariant subbundle whose fiber at the unit element is $\mathfrak{u}_{1,2}(\mathcal{H})\subset \mathfrak{b}_{\res}^\pm(\mathcal{H})^*$. In Section~7.3 we overcome this difficulty by replacing $\mathfrak{u}_{1,2}(\mathcal{H})$ by the quotient Banach space $\operatorname{L}_{1,2}(\mathcal{H})/\mathfrak{b}_{1,2}^\pm(\mathcal{H})$, and construct a Banach Poisson--Lie group structure on  $\operatorname{B}_{\res}^\pm(\mathcal{H})$. The Banach Poisson--Lie group structure of $\operatorname{U}_{\res}(\mathcal{H})$ can be constructed in a similar way. In Section~8, we show that the restricted Grassmannian is a quotient Poisson homogeneous space of  $\operatorname{U}_{\res}(\mathcal{H})$, the stabilizer $H$ of a point being a Banach Poisson--Lie subgroup of $\operatorname{U}_{\res}(\mathcal{H})$. In Section~9.1, we show that the actions of $\operatorname{U}_{\res}(\mathcal{H})$ and  $\operatorname{B}_{\res}^{\pm}(\mathcal{H})$ on the restricted Grassmannian are Poisson actions. In Section~9.2, we show that the symplectic leaves of the Poisson structure of the restricted Grassmannian are the orbits of $\operatorname{B}_{\res}^{\pm}(\mathcal{H})$ and coincides with Schubert cells. At last, we mention that the action of the subgroup $\Gamma^+$ of  $\operatorname{B}_{\res}^{\pm}(\mathcal{H})$ generated by the shift gives rise to the KdV hierachy.

\section{Preliminaries}\label{section1}
Let us introduce some notation. If not stated otherwise, the Banach Lie algebras and related notions are over the field of real numbers. Endow the infinite-dimensional separable complex Hilbert space $\mathcal{H}$ with  orthonormal basis $\{|n\rangle, n\in\mathbb{Z}\}$ ordered with respect to decreasing values of $n$, and consider the decomposition $\mathcal{H} = \mathcal{H}_+ \oplus \mathcal{H}_{-}$, where $\mathcal{H}_+ := \textrm{span}\{|n\rangle~: n\geq 0\}$ and $\mathcal{H}_- := \textrm{span}\{|n\rangle~: n < 0\}$. Denote by $p_+$ (resp. $p_-$) the orthogonal projection onto $\mathcal{H}_+$ (resp. $\mathcal{H}_-$), and set $d  = i(p_+ - p_-)\in \operatorname{L}_\infty(\mathcal{H})$. 

\subsection{Restricted Banach Lie algebra $\operatorname{L}_{\res}(\mathcal{H})$ and its predual $\operatorname{L}_{1,2}(\mathcal{H})$}\label{restricted_Banach_algebras} The restricted Banach Lie algebra is the Banach Lie algebra 
\begin{equation}\label{L_res}
\operatorname{L}_{\res}(\mathcal{H}) = \{ A \in \operatorname{L}_{\infty}(\mathcal{H})~: [d, A] \in \operatorname{L}_{2}(\mathcal{H})\}
\end{equation}
for the norm
$\|A\|_{\res} = \|A\|_{\infty} + \|[d, A]\|_2$ and the bracket given by the commutator of operators.   
A predual of $\operatorname{L}_{\res}$ is 
\begin{equation}\label{L_12}
\operatorname{L}_{1,2}(\mathcal{H}) := \{ A\in \operatorname{L}_{\infty}(\mathcal{H})~: [d, A]\in \operatorname{L}_{2}(\mathcal{H}), p_{\pm}A|_{\mathcal{H}_{\pm}}\in \operatorname{L}_1(\mathcal{H}_{\pm})\}.
\end{equation}
 It is a Banach Lie algebra for the norm given by
$$\| A\|_{1, 2} = \|p_{+}A|_{\mathcal{H}_{+}}\|_1 + \|p_{-}A|_{\mathcal{H}_{-}}\|_1+\|[d, A]\|_2.$$
The duality pairing between $\operatorname{L}_{1,2}(\mathcal{H})$ and $\operatorname{L}_{\res}(\mathcal{H})$ is given by 
$$
\begin{array}{llll}
\langle\cdot, \cdot\rangle_{\operatorname{L}_{\res}, \operatorname{L}_{1,2}}~:& \operatorname{L}_{\res}(\mathcal{H})\times  \operatorname{L}_{1,2}(\mathcal{H})& \rightarrow &\mathbb{C}\\
& (A, B) & \mapsto & \operatorname{Tr}_{\res}(AB),
\end{array}
$$
where the restricted trace $\operatorname{Tr}_{\res}$ (see \cite{GO10}))  is defined on $\operatorname{L}_{1,2}(\mathcal{H})$ by 
\begin{equation}\label{restrictedtrace}
\Tr_{\res} A = \Tr p_{+}A|_{\mathcal{H}_{+}} + \Tr p_{-}A|_{\mathcal{H}_{-}}.
\end{equation}
 According to Proposition~2.1 in \cite{GO10}, one has $\Tr_{\res} AB = \Tr_{\res} BA$ for any $A \in \operatorname{L}_{1,2}(\mathcal{H})$  and any $B \in \operatorname{L}_{\res}(\mathcal{H})$.

\subsection{Restricted general linear group $\operatorname{GL}_{\res}(\mathcal{H})$ and its ``predual'' $\operatorname{GL}_{1,2}(\mathcal{H})$}
The restricted general linear group, denoted by $\operatorname{GL}_{\res}(\mathcal{H})$ is defined as
\begin{equation}\label{GLres}
\operatorname{GL}_{\res}(\mathcal{H}) := \operatorname{GL}(\mathcal{H})\cap \operatorname{L}_{\res}(\mathcal{H}).
\end{equation}
It is an open subset of $\operatorname{L}_{\res}(\mathcal{H})$ hence carries a natural Banach Lie group structure with Banach Lie algebra $\operatorname{L}_{\res}(\mathcal{H})$. 
It is not difficult to show that $\operatorname{GL}_{\res}(\mathcal{H})$ is closed under the operation that takes an operator $A\in \operatorname{GL}_{\res}(\mathcal{H})$ to its inverse $A^{-1}\in \operatorname{GL}(\mathcal{H}) $.
It follows that $\operatorname{GL}_{\res}(\mathcal{H})$ is a Banach Lie group.

The Banach Lie algebra $\operatorname{L}_{1,2}(\mathcal{H})$, predual to $\operatorname{L}_{\res}(\mathcal{H})$, is the Banach Lie algebra of the following Banach Lie group
\begin{equation}\label{GL12}
\operatorname{GL}_{1,2}(\mathcal{H}) := \operatorname{GL}(\mathcal{H})\cap \{\textrm{Id} + A~: A \in \operatorname{L}_{1,2}(\mathcal{H})\}.
\end{equation}

\subsection{Unitary Banach Lie algebras $\mathfrak{u}(\mathcal{H})$,  $\mathfrak{u}_{\res}(\mathcal{H})$, $\mathfrak{u}_{1,2}(\mathcal{H})$} \label{touslesu}
The subspace  
\begin{equation}\label{u}
\mathfrak{u}(\mathcal{H}) := \{A\in \operatorname{L}_{\infty}(\mathcal H)~: A^* = -A\}
\end{equation} 
of skew-Hermitian bounded operators is a real Banach Lie subalgebra of $ \operatorname{L}_{\infty}(\mathcal H)$ considered as a real Banach space. 
The unitary restricted algebra $\mathfrak{u}_{\res}(\mathcal{H})$ is the real Banach Lie  subalgebra of $ \operatorname{L}_{\res}(\mathcal{H})$ consisting of skew-Hermitian operators~:
\begin{equation}\label{ures}
\mathfrak{u}_{\res}(\mathcal{H}) := \{A \in \mathfrak{u}(\mathcal{H})~: [d, A]\in \operatorname{L}_2(\mathcal{H})\} = \operatorname{L}_{\res}(\mathcal{H})\cap \mathfrak u(\mathcal{H}).
\end{equation}
By Proposition 2.1 in \cite{BRT07}, a predual of the unitary restricted algebra $\mathfrak{u}_{\res}(\mathcal{H})$ is the subalgebra $\mathfrak{u}_{1,2}(\mathcal{H})$ of $ \operatorname{L}_{\res}(\mathcal{H})$ consisting of skew-Hermitian operators (see also Remark~\ref{b_into_u} below)~:
\begin{equation}\label{u12}
\mathfrak{u}_{1,2}(\mathcal{H})~:= \{ A \in \operatorname{L}_{1,2}(\mathcal{H})~: A^* = -A\}.
\end{equation}
\begin{remark}{\rm
It follows from Proposition~2.5 in \cite{BRT07} with   $\gamma = 0$ that $\mathfrak{u}_{1,2}(\mathcal{H})$ is a Banach Lie--Poisson space with respect to $\mathfrak{u}_{\res}(\mathcal{H})$. A direct computation shows that $ \mathfrak{u}_{1,2}(\mathcal{H})$ is not a Banach Lie bialgebra with respect to $\mathfrak{u}_{\res}(\mathcal{H})$.}
\end{remark}

\subsection{Restricted unitary group $\operatorname{U}_{\res}(\mathcal{H})$ and its ``predual'' $\operatorname{U}_{1,2}(\mathcal{H})$}
The restricted unitary group is  defined as 
\begin{equation}\label{Ures}
\operatorname{U}_{\res}(\mathcal{H}) := \operatorname{GL}_{\res}(\mathcal{H}) \cap \operatorname{U}(\mathcal{H}).
\end{equation}
It has a natural structure of Banach Lie group with Banach Lie algebra $\mathfrak{u}_{\res}(\mathcal{H})$.
The Banach Lie algebra $\mathfrak{u}_{1,2}(\mathcal{H})$, predual to $\mathfrak{u}_{\res}(\mathcal{H})$, is the Banach Lie algebra of the following Banach Lie group
\begin{equation}\label{U12}
\operatorname{U}_{1,2}(\mathcal{H}) := \operatorname{U}(\mathcal{H})\cap \{\textrm{Id} + A~: A \in \operatorname{L}_{1,2}(\mathcal{H})\}.
\end{equation}

\subsection{The restricted Grassmannian $\operatorname{Gr}_{\res}(\mathcal{H})$}\label{Grres} In the present paper, the restricted Grassmannian $\operatorname{Gr}_{\res}(\mathcal{H})$ denotes the set of all closed subspaces $W$ of $\mathcal{H}$ such that the orthogonal projection $p_-: W \rightarrow \mathcal{H}_-$ is a Hilbert-Schmidt operator. 
The restricted Grassmannian is a homogeneous space under the restricted unitary group (see \cite{PS88}),
$$
\operatorname{Gr}_{\res}(\mathcal{H}) = \operatorname{U}_{\res}(\mathcal{H})/\left(\operatorname{U}(\mathcal{H}_+)\times \operatorname{U}(\mathcal{H}_-)\right),
$$
and under the restricted general linear group $\operatorname{GL}_{\res}(\mathcal{H})$,
$$
\operatorname{Gr}_{\res}(\mathcal{H}) = \operatorname{GL}_{\res}(\mathcal{H})/\operatorname{P}_{\res}(\mathcal{H}),
$$
where 
\begin{equation}\label{Parabolic}
\operatorname{P}_{\res}(\mathcal{H}) = \{ A\in \operatorname{GL}_{\res}(\mathcal{H})~: p_{-}A_{|\mathcal{H}_+} = 0\}.
\end{equation}
It follows that $\operatorname{Gr}_{\res}(\mathcal{H})$ is a homogeneous K\"ahler manifold (see \cite{Wu01}, \cite{BRT07}, \cite{Tum1}, \cite{Tum2} for more informations on the geometry of the restricted Grassmannian).

\subsection{Triangular Banach Lie subalgebras $\mathfrak{b}^\pm_{1,2}(\mathcal{H})$ and $\mathfrak{b}^\pm_{\res}(\mathcal{H})$}\label{triangular_subalgebras}
Let us define the following triangular subalgebras of $\operatorname{L}_{1,2}(\mathcal{H})$ and  $\operatorname{L}_{\res}(\mathcal{H})$~:
 $$
\mathfrak{b}^+_{1, 2}(\mathcal{H}) := \{\alpha\in \operatorname{L}_{1,2}(\mathcal{H})~: \alpha\left(|n\rangle\right) \in~ \textrm{span}\{|m\rangle, m\geq n\}~\textrm{and}~\langle n|\alpha|n\rangle\in\mathbb{R}, \textrm{for}~ n\in\mathbb{Z}\}.
$$  
$$
\mathfrak{b}^-_{1, 2}(\mathcal{H}) := \{\alpha\in \operatorname{L}_{1,2}(\mathcal{H})~: \alpha\left(|n\rangle\right) \in~ \textrm{span}\{|m\rangle, m\leq n\}~\textrm{and}~\langle n|\alpha|n\rangle\in\mathbb{R}, \textrm{for}~ n\in\mathbb{Z}\},
$$
$$
\mathfrak{b}^+_{\res}(\mathcal{H}) := \{\alpha\in \operatorname{L}_{\res}(\mathcal{H})~: \alpha\left(|n\rangle\right) \in~ \textrm{span}\{|m\rangle, m\geq n\}~\textrm{and}~\langle n|\alpha|n\rangle\in\mathbb{R}, \textrm{for}~ n\in\mathbb{Z}\}.
$$  
$$
\mathfrak{b}^-_{\res}(\mathcal{H}) := \{\alpha\in \operatorname{L}_{\res}(\mathcal{H})~: \alpha\left(|n\rangle\right) \in~ \textrm{span}\{|m\rangle, m\leq n\}~\textrm{and}~\langle n|\alpha|n\rangle\in\mathbb{R}, \textrm{for}~ n\in\mathbb{Z}\}.
$$

\subsection{Triangular Banach Lie groups  $\operatorname{B}_{1, 2}^{\pm}(\mathcal{H}) $, and  
$\operatorname{B}_{\res}^{\pm}(\mathcal{H})$}\label{Triangular_groups}

Consider
$$
\operatorname{B}_{1, 2}^{\pm}(\mathcal{H}):= \{\alpha\in \operatorname{GL}(\mathcal{H})\cap \left(\textrm{Id}+\mathfrak{b}^\pm_{1, 2}(\mathcal{H})\right)~:  \alpha^{-1}\in \textrm{Id} +\mathfrak{b}^\pm_{1, 2}(\mathcal{H}), \forall~ n\in\mathbb{Z},~\langle n|\alpha|n\rangle\in\mathbb{R}^{+*}\}.
$$  
For any $A\in \mathfrak{b}^\pm_{1, 2}(\mathcal{H})$ with $\|A\|_{1, 2}<1$, and any $\alpha\in \operatorname{B}_{1, 2}^{\pm}(\mathcal{H}) $, the operator $\alpha-\alpha A$ belongs to $\operatorname{B}_{1, 2}^{\pm}(\mathcal{H}) $, since
$$
(\alpha -\alpha A)^{-1} = (\textrm{Id} - A)^{-1}\alpha^{-1},
$$
and $ (\textrm{Id} - A)^{-1} = \sum_{n=0}^{\infty} A^{n}$ is a convergent series in $\left(\textrm{Id} + \mathfrak{b}^\pm_{1, 2}(\mathcal{H})\right)$, whose limit admits strictly positive diagonal coefficients. Hence $\operatorname{B}_{1, 2}^{\pm}(\mathcal{H}) $ is an open subset of $\left(\textrm{Id} + \mathfrak{b}^\pm_{1, 2}(\mathcal{H})\right)$, stable under group multiplication and inversion.
It follows that $\operatorname{B}_{1, 2}^{\pm}(\mathcal{H}) $ is a Banach Lie group with Banach Lie algebra  $\mathfrak{b}^\pm_{1, 2}(\mathcal{H})$.

Similarly define the following Banach Lie groups of triangular operators~:
$$
\operatorname{B}_{\res}^{\pm}(\mathcal{H}):= \{\alpha\in \operatorname{GL}_{\res}(\mathcal{H})\cap \mathfrak{b}^\pm_{\res}(\mathcal{H})~: \alpha^{-1}\in \operatorname{GL}_{\res}(\mathcal{H})\cap\mathfrak{b}^\pm_{\res}(\mathcal{H}) \\~\textrm{and}~\forall~ n\in\mathbb{Z}, \langle n|\alpha|n\rangle\in\mathbb{R}^{+*}\}.
$$

\begin{remark}{\rm
Remark that $\operatorname{B}_{\res}^{+}(\mathcal{H})$ does not contain the shift operator $S: \mathcal{H}\rightarrow \mathcal{H}$, $|n\rangle \mapsto |n+1\rangle$ since the diagonal coefficients of any element in $\operatorname{B}_{\res}^{+}(\mathcal{H})$ are non-zero. However $S$ belongs to the Lie algebra $\mathfrak{b}^+_{\res}(\mathcal{H})$, whereas $S^{-1}$ belongs to $\mathfrak{b}^-_{\res}(\mathcal{H})$.}
\end{remark}

\section{Example of  Banach Lie bialgebras and Banach Poisson--Lie groups related \\to the restricted Grassmannian}\label{section6}

\subsection{Iwasawa Banach Lie bialgebras}
Recall that   $\langle\cdot,\cdot\rangle_{\operatorname{L}_{\res}, \operatorname{L}_{1, 2}}$ denote the continuous bilinear map given by the imaginary part of the restricted trace (see equation~\eqref{restrictedtrace})~:
$$
\begin{array}{lcll}
\langle\cdot,\cdot\rangle_{\operatorname{L}_{\res}, \operatorname{L}_{1, 2}}~: & \operatorname{L}_{\res}(\mathcal{H})\times \operatorname{L}_{1,2}(\mathcal{H})& \longrightarrow &\mathbb{R}\\
& (x, y) & \longmapsto & \Im \Tr_{\res}\left(x y\right).
\end{array}
$$

\begin{proposition}\label{duality_U_b}
The continuous bilinear map $\langle\cdot,\cdot\rangle_{L_{\res}, \operatorname{L}_{1, 2}}$ restricts to 
  a weak duality pairing between $\mathfrak{u}_{\res}(\mathcal{H})$ and $\mathfrak{b}^\pm_{1, 2}(\mathcal{H})$ denoted by 
$$
\begin{array}{lcll}
\langle\cdot,\cdot\rangle_{\mathfrak{u}_{\res}, \mathfrak{b}_{1, 2}^{\pm}}~: & \mathfrak{u}_{\res}(\mathcal{H})\times \mathfrak{b}^{\pm}_{1,2}(\mathcal{H}) & \longrightarrow &\mathbb{R}\\
& (x, y) & \longmapsto & \Im \Tr_{\res}\left(x y\right).
\end{array}
$$
Similarly the continuous bilinear map $\langle\cdot,\cdot\rangle_{L_{\res}, \operatorname{L}_{1, 2}}$ restricts to  a weak duality pairing between $\mathfrak{b}^\pm_{\res}(\mathcal{H})$ and $\mathfrak{u}_{1,2}(\mathcal{H})$  denoted by 
$$
\begin{array}{lcll}
\langle\cdot,\cdot\rangle_{ \mathfrak{b}_{\res}^{\pm}, \mathfrak{u}_{1, 2}}~: & \mathfrak{b}^{\pm}_{\res}(\mathcal{H})\times \mathfrak{u}_{1, 2}(\mathcal{H})& \longrightarrow &\mathbb{R}\\
& (x, y) & \longmapsto & \Im \Tr_{\res}\left(x y\right).
\end{array}
$$
\end{proposition}

\begin{proof}
Let us show that the map $(a, b)\mapsto \Im \Tr_{\res} ab$ is non-degenerate for $a\in \mathfrak{u}_{\res}(\mathcal{H})$ and $b \in \mathfrak{b}^{+}_{1,2}(\mathcal{H})$.

 Suppose that $a\in \mathfrak{u}_{\res}(\mathcal{H})$ is such that $\Im \Tr_{\res} ab = 0$ for any $b \in \mathfrak{b}^{+}_{1,2}(\mathcal{H})$ and let us show that $a$ necessary vanishes. Since $\{|n\rangle\}_{n\in\mathbb{Z}}$ is an orthonormal basis of $\mathcal{H}$ and $a$ is bounded, it is sufficient to show that for any $n , m \in\mathbb{Z}$, $\langle m| a n\rangle = 0$. In fact, since $a$ is skew-symmetric, it is enough to show that $\langle m| a n\rangle = 0$ for $m\leq n$. For $n\geq m$, the operator $E_{nm} = |n\rangle\langle m| $ of rank one given by $x\mapsto \langle m, x\rangle |n\rangle$ belongs to $\mathfrak{b}^{+}_{1,2}(\mathcal{H})$.  Hence for $n\geq m$, one has
$$
\Im \Tr_{\res} a E_{nm} = \Im \left(\sum_{j\in\mathbb{Z}}  \langle j| m\rangle \langle j| a n\rangle\right) = \Im \langle m| a n\rangle = 0.
$$
In particular, for $m=n$, since $\langle n| a n\rangle$ is purely imaginary, one has $\langle n| a n\rangle = 0$, $ \forall n\in\mathbb{Z}$. For $n>m$, the operator $iE_{nm}$ belongs also to $\mathfrak{b}^{+}_{1,2}(\mathcal{H})$ and 
$$
\Im \Tr_{\res} a iE_{nm} = \Im \left(\sum_{j\in\mathbb{Z}}  i\langle j| m\rangle \langle j| a n\rangle\right) = \Re \langle m| a n\rangle = 0.
$$
This allows to conclude that $\langle m| a n\rangle = 0$ for any $n,m \in\mathbb{Z}$, hence $a = 0\in \mathfrak{u}_{\res}(\mathcal{H})$.
 
 On the other hand, consider an element $b\in \mathfrak{b}^{+}_{1,2}(\mathcal{H})$ such that $\Im \Tr ab =  0$ for any $a\in \mathfrak{u}_{\res}(\mathcal{H})$. We will show that $\langle n| b m \rangle = 0$ for any $n,m\in\mathbb{Z}$ such that $n\geq m$. For $n> m$, the operator $E_{mn}-E_{nm}$ belongs to $\mathfrak{u}_{\res}(\mathcal{H})$, and for $n\geq m$, $i E_{mn} + i E_{nm}\in \mathfrak{u}_{\res}(\mathcal{H})$.  Therefore for $n>m$, one has
 $$
 \Im \Tr_{\res} \left(E_{mn}-E_{nm}\right) b = \Im \left( \langle n| b m\rangle - \langle m| b n\rangle\right) = \Im \langle n| b m\rangle = 0,
 $$
 and for $n\geq m$, one has
 $$
 \Im \Tr_{\res} \left(i E_{mn}+ iE_{nm}\right) b = \Im \left(i \langle n| b m\rangle + i \langle m| b n\rangle\right) = \Re \langle n| b m\rangle = 0.
 $$
 It follows that $\langle n| b m\rangle = 0$ for all $n,m\in\mathbb{Z}$ such that $n>m$. Moreover, since $\langle n| b n\rangle \in\mathbb{R}$ for any $n\in\mathbb{Z}$, one also has $\langle n| b n\rangle = 0, \forall n\in\mathbb{Z}$. Consequently $b = 0$.
 
 It follows that $\langle\cdot,\cdot\rangle_{\mathfrak{u}_{\res}, \mathfrak{b}_{1, 2}^{\pm}}~:  \mathfrak{u}_{\res}(\mathcal{H})\times \mathfrak{b}^{+}_{1,2}(\mathcal{H})\rightarrow \mathbb{R}$, $(x, y) \mapsto \Im \Tr_{\res} x y$, is  non-degenerate and defines a duality pairing between $\mathfrak{u}_{\res}(\mathcal{H})$ and $\mathfrak{b}^{+}_{1,2}(\mathcal{H})$. One shows in a similar way that $\langle\cdot,\cdot\rangle_{\operatorname{L}_{\res}, \operatorname{L}_{1, 2}}$ induces a duality pairing between $\mathfrak{u}_{\res}(\mathcal{H})$ and $\mathfrak{b}^{-}_{1,2}(\mathcal{H})$, between $\mathfrak{u}_{1, 2}(\mathcal{H})$ and $\mathfrak{b}^{+}_{\res}(\mathcal{H})$, and between $\mathfrak{u}_{1, 2}(\mathcal{H})$ and $\mathfrak{b}^{-}_{\res}(\mathcal{H})$.
 \end{proof}

\begin{remark}\label{b_into_u}{\rm
Recall that by Proposition~2.1 in \cite{BRT07}, the dual space $\mathfrak{u}_{1, 2}(\mathcal{H})^*$ can be identified with $\mathfrak{u}_{\res}(\mathcal{H})$, the duality pairing being given by $(a, b)\mapsto \Tr_\res (ab)$. By previous Proposition, one has a continuous injection from $\mathfrak{b}^{+}_{\res}(\mathcal{H})$ into  $\mathfrak{u}_{1, 2}(\mathcal{H})^*$ by $a\mapsto (b\mapsto \Im\Tr_\res (a b))$. The corresponding injection from $\mathfrak{b}^{+}_{\res}(\mathcal{H})$ into $\mathfrak{u}_{\res}(\mathcal{H})\simeq \mathfrak{u}_{1, 2}(\mathcal{H})^*$ reads~:
$$
\begin{array}{llll}
\iota~:&\mathfrak{b}^{+}_{\res}(\mathcal{H}) & \hookrightarrow & \mathfrak{u}_{\res}(\mathcal{H})\\
& b & \mapsto & -\frac{i}{2}(b + b^*).
\end{array}
$$
The range of $\iota$ is the subspace of $\mathfrak{u}_{\res}(\mathcal{H})$ consisting of those $x\in \mathfrak{u}_{\res}(\mathcal{H})$ such that the triangular truncation $T_-(x)$ is bounded. Recall that $T_-$ is unbounded on $L_\infty(\mathcal{H})$, as well as on $L_1(\mathcal{H})$ (see \cite{M61}, \cite{KP70}, \cite{GK70}), and that there exists skew-symmetric bounded operators whose triangular truncation is not bounded (see \cite{D88}). Therefore $\iota$ is not surjective.
}
\end{remark}


\begin{theorem}\label{les_bigebres}
The Banach Lie algebra $\mathfrak{b}_{1, 2}^\pm(\mathcal{H})$ is a Banach Lie bialgebra with respect to $\mathfrak{u}_{\res}(\mathcal{H})$. Similarly the Banach Lie algebra $\mathfrak{u}_{1, 2}(\mathcal{H})$ is a Banach Lie bialgebra with respect to $\mathfrak{b}_{\res}^\pm(\mathcal{H})$.
\end{theorem}

\begin{proof}
Let us show that the Lie algebra structure $[\cdot, \cdot]_{\mathfrak{u}_{\res}}$ on $\mathfrak{u}_{\res}(\mathcal{H})$  is such that 
 \begin{enumerate}
 \item $\mathfrak{b}_{1, 2}^\pm(\mathcal{H})$ acts continuously by coadjoint action on $\mathfrak{u}_{\res}(\mathcal{H})$; 
 \item the dual map $[\cdot, \cdot]_{\mathfrak{u}_{\res}}^*~: \mathfrak{u}_{\res}^{*}(\mathcal{H})\rightarrow \Lambda^2\mathfrak{u}_{\res}^*(\mathcal{H})$  to the Lie bracket $[\cdot, \cdot]_{\mathfrak{u}_{\res}}~:\mathfrak{u}_{\res}(\mathcal{H})\times\mathfrak{u}_{\res}(\mathcal{H})\rightarrow \mathfrak{u}_{\res}(\mathcal{H})$ restricts to
  a $1$-cocycle $\theta~: \mathfrak{b}_{1, 2}^\pm(\mathcal{H})\rightarrow \Lambda^2\mathfrak{u}_{\res}^*(\mathcal{H})$ with respect to the adjoint action $\ad^{(2,0)}$ of $\mathfrak{b}_{1, 2}^\pm(\mathcal{H})$ on  $\Lambda^2\mathfrak{u}^*_{\res}(\mathcal{H}).$ 
   \end{enumerate} 
  \begin{itemize}
  \item 
  Let us first prove (1). Since by Proposition~\ref{duality_U_b}, $\langle\cdot,\cdot\rangle_{ \mathfrak{b}_{1, 2}^{\pm}, \mathfrak{u}_{\res}}$ is a duality pairing between $\mathfrak{u}_{\res}(\mathcal{H})$ and $\mathfrak{b}^\pm_{1, 2}(\mathcal{H})$, the Banach space $\mathfrak{u}_{\res}(\mathcal{H})$ is a subspace of the continuous dual of $\mathfrak{b}^\pm_{1, 2}(\mathcal{H})$. 
   Recall that the coadjoint action of $\mathfrak{b}_{1, 2}^\pm(\mathcal{H})$ on its dual reads
   $$
\begin{array}{llll}
-\ad^*~: &\mathfrak{b}_{1, 2}^\pm(\mathcal{H})\times\mathfrak{b}_{1, 2}^\pm(\mathcal{H})^*&\longrightarrow &\mathfrak{b}_{1, 2}^\pm(\mathcal{H})^*\\
& (x, \alpha) & \longmapsto & -\ad^*_x \alpha := -\alpha\circ \ad_x.
\end{array}
$$
Let us show that $\mathfrak{u}_{\res}(\mathcal{H})$ is invariant under coadjoint action. This means that when $\alpha$ is given by $\alpha(y) =  \Im \Tr_\res a y$ for some $a\in \mathfrak{u}_{\res}(\mathcal{H})$, then,  for any $x\in \mathfrak{b}_{1, 2}^\pm(\mathcal{H})$, the one form $\beta =  -\ad^*_x \alpha$ reads $\beta(y) = \Im \Tr_\res \tilde{a} y$ for some $\tilde{a}\in \mathfrak{u}_{\res}(\mathcal{H})$. One has
$$
\begin{array}{ll}
\beta(y) &= -\ad^*_x \alpha(y) = -\alpha(\ad_x y) = - \alpha([x, y]) \\&= -\Im \Tr_\res a [x, y] = -\Im \Tr_\res (a x y - a y x),
\end{array}
$$
where $a\in \mathfrak{u}_{\res}(\mathcal{H})$, $x, y \in \mathfrak{b}_{1, 2}^\pm(\mathcal{H})$.
Since $a y $ and $x$ belong to $L_2(\mathcal{H})$, $a y x$ and $x a y $ belong to $L_{1}(\mathcal{H})$ and $\Tr_\res (a y x) =  \Tr (a y x) = \Tr (x a y)$. Since $a x y$ belongs also to $L_{1}(\mathcal{H})$, one has
$$
\beta(y) =   -\Im \Tr (a x y) + \Im \Tr (a y x)  = -\Im \Tr (a x y) +\Im \Tr (x a y) = -\Im \Tr ([a, x] y).
$$
Note that $[a, x]$ belongs to $L_2(\mathcal{H})$. Recall that by Proposition~\ref{triples},  the triples of Hilbert Lie algebras $(L_2(\mathcal{H}), \mathfrak{u}_{2}(\mathcal{H}), \mathfrak{b}^+_2(\mathcal{H}))$ and 
$(L_2(\mathcal{H}), \mathfrak{u}_{2}(\mathcal{H}), \mathfrak{b}^-_2(\mathcal{H}))$ are real Hilbert Manin triples with respect to the pairing $\langle\cdot,\cdot\rangle_{\mathbb{R}}$ given by the imaginary part of the trace. Using the decomposition $L_2(\mathcal{H}) = \mathfrak{u}_2(\mathcal{H}) \oplus \mathfrak{b}_{2}^+(\mathcal{H})$, and the continuous projection $p_{\mathfrak{u}_2^\pm}~: L_2(\mathcal{H}) \rightarrow \mathfrak{u}_2(\mathcal{H})$  with kernal $\mathfrak{b}_{2}^\pm(\mathcal{H})$, one therefore has
$$
\beta(y) = -\Im \Tr p_{\mathfrak{u}_2^\pm}([a, x]) y,
$$
since $ y \in \mathfrak{b}_{1, 2}^\pm(\mathcal{H})\subset \mathfrak{b}_{2}^\pm(\mathcal{H})$ and $\mathfrak{b}_{2}^\pm(\mathcal{H})$ is isotropic. It follows that $\beta(y) = \Im \Tr \tilde{a} y$ with 
$$
\tilde{a} = -p_{\mathfrak{u}_2^\pm}([a, x]) \in \mathfrak{u}_2(\mathcal{H}) \subset \mathfrak{u}_{\res}(\mathcal{H}).
$$
In other words, the coadjoint action of $x\in \mathfrak{b}_{1, 2}^\pm(\mathcal{H})$ maps $a\in \mathfrak{u}_{\res}(\mathcal{H})$ to $-\ad^*_x a = -p_{\mathfrak{u}_2^\pm}([a, x])\in \mathfrak{u}_{\res}(\mathcal{H})$. The continuity of the map
$$
\begin{array}{lcll}
-\ad^*~: & \mathfrak{b}_{1, 2}^\pm(\mathcal{H}) \times \mathfrak{u}_{\res}(\mathcal{H}) & \rightarrow &\mathfrak{u}_{\res}(\mathcal{H})\\
& (x, a) & \mapsto & -\ad^*_x a = -p_{\mathfrak{u}_2^\pm}([a, x])
\end{array}
$$
follows from the continuity of the product 
$$
\begin{array}{cll}
 \mathfrak{b}_{1, 2}^\pm(\mathcal{H}) \times \mathfrak{u}_{\res}(\mathcal{H}) & \rightarrow &L_1(\mathcal{H})\\
 (x, a) & \mapsto & ax,
\end{array}
$$
from the continuity of the projection $p_{\mathfrak{u}_2^\pm}$ and from the continuity of the  injections $L_1(\mathcal{H})\subset L_2(\mathcal{H})$ and $\mathfrak{u}_{2}(\mathcal{H})\subset \mathfrak{u}_{\res}(\mathcal{H})$.

\item Let us now prove (2). The dual map of the bilinear map $[\cdot, \cdot]_{\mathfrak{u}_{\res}}$ is given by
$$
\begin{array}{lclcl}
[\cdot, \cdot]_{\mathfrak{u}^*_{\res}}~: & \mathfrak{u}^{*}_{\res}(\mathcal{H})& \longrightarrow & L(\mathfrak{u}_{\res}(\mathcal{H}), \mathfrak{u}_{\res}(\mathcal{H}); \mathbb{K}) & \simeq L(\mathfrak{u}_{\res}(\mathcal{H}); \mathfrak{u}^*_{\res}(\mathcal{H}))\\
& \mathcal{F}(\cdot)&  \longmapsto& \mathcal{F}\left([\cdot, \cdot]_{\mathfrak{u}_{\res}}\right) & \mapsto  \left(\alpha \mapsto \mathcal{F}\left([\alpha, \cdot]_{\mathfrak{u}_{\res}}\right) = \ad^*_{\alpha}\mathcal{F}(\cdot)\right),
\end{array}
$$
and takes values in $\Lambda^2\mathfrak{u}^*_{\res}(\mathcal{H})$. Since by (1), $\mathfrak{u}_{\res}(\mathcal{H})\subset  \mathfrak{b}_{1, 2}^\pm(\mathcal{H})^*$ is stable under the coadjoint action of $\mathfrak{b}_{1, 2}^\pm(\mathcal{H})$ and the coadjoint action $\ad^*~:\mathfrak{b}_{1, 2}^\pm(\mathcal{H})\times\mathfrak{u}_{\res}(\mathcal{H})\rightarrow \mathfrak{u}_{\res}(\mathcal{H})$ is continuous, one can consider the adjoint action of $\mathfrak{b}_{1, 2}^\pm(\mathcal{H})$ on $\Lambda^2\mathfrak{u}^*_{\res}(\mathcal{H})$ defined by \eqref{ad2}.
Denote by $\theta$ the restriction of $[\cdot, \cdot]_{\mathfrak{u}_{\res}}^*$ to the subspace $\mathfrak{b}_{1, 2}^\pm(\mathcal{H})\subset \mathfrak{u}_{\res}(\mathcal{H})^{*}$~:
$$
\begin{array}{lclcl}
\theta~: & \mathfrak{b}_{1, 2}^\pm(\mathcal{H})& \longrightarrow & L(\mathfrak{u}_{\res}(\mathcal{H}), \mathfrak{u}_{\res}(\mathcal{H}); \mathbb{K}) & \simeq L(\mathfrak{u}_{\res}(\mathcal{H}); \mathfrak{u}_{\res}(\mathcal{H})^{*})\\
&x&  \longmapsto& \langle  x, [\cdot, \cdot]_{\mathfrak{u}_{\res}}\rangle_{ \mathfrak{b}_{1,2}^{\pm}, \mathfrak{u}_{\res}}& \mapsto  \left(\alpha \mapsto \langle  x, [\alpha, \cdot]_{\mathfrak{u}_{\res}}\rangle_{\mathfrak{b}_{1,2}^{\pm}, \mathfrak{u}_{\res}} = \ad^*_{\alpha}x(\cdot)\right).
\end{array}
$$
The condition \eqref{cocycle} expressing that $\theta$ is a $1$-cocycle reads~:
\begin{equation}\label{ures-b-cocycle}
\begin{array}{ll}
\langle [\alpha, \beta], [x, y]\rangle_{ \mathfrak{b}_{1,2}^{\pm}, \mathfrak{u}_{\res}}& = +\langle y, [\ad^*_x\alpha, \beta]\rangle_{ \mathfrak{b}_{1,2}^{\pm}, \mathfrak{u}_{\res}} +\langle y, [\alpha, \ad^*_x\beta]\rangle_{ \mathfrak{b}_{1,2}^{\pm}, \mathfrak{u}_{\res}} \\ & \quad-\langle x, [\ad^*_y\alpha, \beta]\rangle_{ \mathfrak{b}_{1,2}^{\pm}, \mathfrak{u}_{\res}} - \langle x, [\alpha, \ad^*_y\beta]\rangle_{ \mathfrak{b}_{1,2}^{\pm}, \mathfrak{u}_{\res}}.  
\end{array}
\end{equation}
The first  term in the RHS reads
$$
+\langle y, [\ad^*_x\alpha, \beta]\rangle_{ \mathfrak{b}_{1,2}^{\pm}, \mathfrak{u}_{\res}} = \Im \Tr y [p_{\mathfrak{u}_2^\pm}([\alpha, x]), \beta] = \Im \Tr [\beta, y] p_{\mathfrak{u}_2^\pm}([\alpha, x]).
$$
Using the fact that $[\beta, y]\in L_2(\mathcal{H})$, and that $\mathfrak{u}_2(\mathcal{H})\subset L_2(\mathcal{H})$ and $\mathfrak{b}_2^\pm(\mathcal{H})\subset L_2(\mathcal{H})$ are isotropic subspaces with respect to the pairing given by the imaginary part of the trace, one has
$$
+\langle y, [\ad^*_x\alpha, \beta]\rangle_{ \mathfrak{b}_{1,2}^{\pm}, \mathfrak{u}_{\res}} = \Im \Tr p_{\mathfrak{b}_2^\pm}([\beta, y]) p_{\mathfrak{u}_2^\pm}([\alpha, x]).
$$
Similarly the second, third and last term in the RHS of equation \eqref{ures-b-cocycle} read respectively
$$
\begin{array}{l}
+\langle y, [\alpha, \ad^*_x\beta]\rangle_{ \mathfrak{b}_{1,2}^{\pm}, \mathfrak{u}_{\res}} = \Im \Tr p_{\mathfrak{b}_2^\pm}([y, \alpha]) p_{\mathfrak{u}_2^\pm}([\beta, x]),\\
-\langle x, [\ad^*_y\alpha, \beta]\rangle_{ \mathfrak{b}_{1,2}^{\pm}, \mathfrak{u}_{\res}} = -\Im \Tr p_{\mathfrak{b}_2^\pm}([\beta, x]) p_{\mathfrak{u}_2^\pm}([\alpha, y]),\\
- \langle x, [\alpha, \ad^*_y\beta]\rangle_{ \mathfrak{b}_{1,2}^{\pm}, \mathfrak{u}_{\res}} = -\Im \Tr p_{\mathfrak{b}_2^\pm}([x, \alpha]) p_{\mathfrak{u}_2^\pm}([\beta, y]).
\end{array}
$$
Using once more the fact that $\mathfrak{u}_2(\mathcal{H})\subset L_2(\mathcal{H})$ and $\mathfrak{b}_2^\pm(\mathcal{H})\subset L_2(\mathcal{H})$ are isotropic subspaces with respect to the pairing given by the imaginary part of the trace, it follows that the first and last term in the RHS of equation~\eqref{ures-b-cocycle} sum up to give
$$
+\langle y, [\ad^*_x\alpha, \beta]\rangle_{ \mathfrak{b}_{1,2}^{\pm}, \mathfrak{u}_{\res}} - \langle x, [\alpha, \ad^*_y\beta]\rangle_{ \mathfrak{b}_{1,2}^{\pm}, \mathfrak{u}_{\res}} =- \Im \Tr [\beta, y] [x, \alpha],
$$
and the second and third term in equation~\eqref{ures-b-cocycle} simplify to
$$
+\langle y, [\alpha, \ad^*_x\beta]\rangle_{ \mathfrak{b}_{1,2}^{\pm}, \mathfrak{u}_{\res}} -\langle x, [\ad^*_y\alpha, \beta]\rangle_{ \mathfrak{b}_{1,2}^{\pm}, \mathfrak{u}_{\res}} =- \Im \Tr [\beta, x] [\alpha, y].
$$
Developping the brackets and using that, for $A$ and $B$ bounded such that $A B$ and $BA$ are trace class, one has $\Tr AB = \Tr BA$, the RHS of equation~\eqref{ures-b-cocycle} becomes
$$
\begin{array}{ll}
\Im \Tr [\beta, y] [x, \alpha] + \Im \Tr [\beta, x] [\alpha, y] &= \Im \Tr (-\beta y x \alpha - y \beta \alpha x + \beta x y \alpha + x \beta \alpha y) \\&= \Im \Tr (x y \alpha \beta - x y \beta \alpha - y x \alpha \beta + y x \beta \alpha) \\&= \Im \Tr [x, y] [\alpha, \beta] \\ & = \langle [x, y], [\alpha, \beta]\rangle_{\mathfrak{b}_{1,2}^{\pm}, \mathfrak{u}_{\res}},
\end{array}
$$
hence $\theta$ satisfies the cocycle condition.

One can show in a similar way that 
the Lie algebra structure $[\cdot, \cdot]_{\mathfrak{b}^\pm_{\res}}$ on $\mathfrak{b}_{\res}^\pm(\mathcal{H})$  is such that 
 \begin{enumerate}
 \item $\mathfrak{u}_{1, 2}(\mathcal{H})$ acts continuously by coadjoint action on $\mathfrak{b}_{\res}^\pm(\mathcal{H})$; 
 \item the dual map $[\cdot, \cdot]_{\mathfrak{b}_{\res}^\pm}^*~: \mathfrak{b}_{\res}^\pm(\mathcal{H})^{*}\rightarrow \Lambda^2\mathfrak{b}_{\res}^\pm(\mathcal{H})^*$  to the Lie bracket $[\cdot, \cdot]_{\mathfrak{b}_{\res}^\pm}~:\mathfrak{b}_{\res}^\pm(\mathcal{H})\times\mathfrak{b}_{\res}^\pm(\mathcal{H})\rightarrow \mathfrak{b}_{\res}^\pm(\mathcal{H})$ restricts to
  a $1$-cocycle $\theta~: \mathfrak{u}_{1, 2}(\mathcal{H})\rightarrow \Lambda^2\mathfrak{b}_{\res}^\pm(\mathcal{H})^*$ with respect to the adjoint action $\ad^{(2,0)}$ of $\mathfrak{u}_{1, 2}(\mathcal{H})$ on  $\Lambda^2\mathfrak{b}_{\res}^\pm(\mathcal{H})^*.$ 
   \end{enumerate} 
   \end{itemize}
\end{proof}

\subsection{Unbounded coadjoint actions}\label{no_Manin}

Recall  that for $1<p<\infty$ and $q:= \frac{p}{p-1}$, $\mathfrak{u}_p(\mathcal{H})$ and $\mathfrak{b}_q^\pm(\mathcal{H})$ are dual Banach Lie--Poisson spaces (see Example~\ref{exIwasawa}), and that the coadjoint actions are given by $$\ad^*_{\alpha} x = p_{\mathfrak{u}_p,\pm}\left([x, \alpha]\right)\quad\textrm{ and   }\quad
\ad^*_{x}\alpha = p_{\mathfrak{b}_q^\pm}\left([\alpha, x]\right),
$$
where $x\in\mathfrak{u}_p(\mathcal{H})$ and  $\alpha\in \mathfrak{b}^\pm_q(\mathcal{H})$. 
In this example,  the continuity of the triangular truncation $T_+$ 
on $\operatorname{L}_p(\mathcal{H})$ and $\operatorname{L}_q(\mathcal{H})$ (see Section~\ref{Triangular truncations of operators}) is crucial in order to define the orthogonal projections $p_{\mathfrak{u}_p,\pm}$ and $p_{\mathfrak{b}_q^\pm}$ using equations~\eqref{projectionu+} and~~\eqref{projectionu-}. 

The situation is different for the Banach Lie algebras $\mathfrak{u}_{1, 2}(\mathcal{H})$ and $\mathfrak{b}_{\res}^\pm(\mathcal{H})$. We will show that $\mathfrak{u}_{1, 2}(\mathcal{H})$  is not a Banach Lie--Poisson space with respect to  $\mathfrak{b}_{\res}^\pm(\mathcal{H})$ since the coadjoint action of $\mathfrak{b}_{\res}^\pm(\mathcal{H})$ on $\mathfrak{u}_{1, 2}(\mathcal{H})$ is unbounded. 
To prove this result, we will use the fact that the triangular truncation is unbounded on the space of trace class operators.
 In a similar way, the coadjoint action of 
 $\mathfrak{u}_{\res}(\mathcal{H})$ on 
 $\mathfrak{b}_{1, 2}^+(\mathcal{H})$ is unbounded (see also \cite{ABT3})
 Using Theorem~\ref{bialgebra_to_manin}, we conclude that there is no Banach Manin triple associated to the pair $(\mathfrak{u}_{1, 2}(\mathcal{H}), \mathfrak{b}_{\res}^\pm(\mathcal{H}))$ nor to the pair $(\mathfrak{b}_{1, 2}^+(\mathcal{H}), \mathfrak{u}_{\res}(\mathcal{H}))$  for the duality pairing given by the imaginary part of the restricted trace (see Theorem~\ref{ccl} below).

%

\begin{proposition}\label{ex2}
There exist a bounded sequence of elements $x_n\in\mathfrak{u}_{1, 2}(\mathcal{H})$ and an element $y\in\mathfrak{b}_{\res}^\pm(\mathcal{H})$ such that 
 $$\|T_+([x_n, y]|_{\mathcal{H}_+})\|_1 \rightarrow +\infty.$$
\end{proposition}

\begin{proof}
Consider the Hilbert space $\mathcal{H} = \mathcal{H}_+\oplus\mathcal{H}_-$, with orthonormal basis $\{|n\rangle, n\in\mathbb{Z}\}$ ordered with respect to decreasing values of $n$,  where $\mathcal{H}_+ = \textrm{span}\{|n\rangle, n>0\}$ and $\mathcal{H}_- = \textrm{span}\{|n\rangle, n\leq 0\}$.
Furthermore decompose $\mathcal{H}_+$ into the Hilbert sum of $\mathcal{H}_+^{\textrm{even}} := \textrm{span}\{|2n+2\rangle, n\in\mathbb{N}\}$ and $\mathcal{H}_+^{\textrm{odd}} := \textrm{span}\{|2n+1\rangle, n\in\mathbb{N}\}$. We will denote by $u~:\mathcal{H}_+^{\textrm{odd}}\rightarrow\mathcal{H}_+^{\textrm{even}}$  the unitary operator defined by
$u|2n+1\rangle = |2n+2\rangle$.

Since the triangular truncation is not bounded on the Banach space of trace class operators, there exists a sequence $K_n\in L_1(\mathcal{H}_+^{\textrm{odd}})$ such that $\|K_n\|_1\leq 1$ and $\|T_+(K_n)\|_1 > n$ for all $n\in \mathbb{N}$. It follows that either $\|T_+(K_n+K_n^*)/2\|_1>n/2$ or $\|T_+(K_n-K_n^*)/2\|_1>n/2$. Modulo the extraction of a  subsequence, we can suppose that $K_n$ is either Hermitian $K_n = K_n^*$ or skew-Hermitian $K_n = - K_n^*$. Moreover, since the triangular truncation is complex linear, the existence of a sequence of skew-Hermitian operators such that $\|K_n\|_1\leq 1$ and $\|T_+(K_n)\|_1>n/2$ implies that the sequence $iK_n$ is a sequence of Hermitian operators such that $\|i K_n\|_1 \leq 1$ and $\|T_+(iK_n)\|_1>n/2$. Therefore without loss of generality we can suppose that $K_n$ are Hermitian.

Consider the bounded operators $x_n$ defined by $0$ on $\mathcal{H}_-$,  preserving $\mathcal{H}_+$ and whose expression with respect to the decomposition $\mathcal{H}_+ = \mathcal{H}_+^{\textrm{even}} \oplus \mathcal{H}_+^{\textrm{odd}}$ reads
\begin{equation}\label{x}
x_n|_{\mathcal{H}_+}=\left(\begin{array}{cc} 0 & u K_n\\ -K_n^* u^* & 0\end{array}\right).
\end{equation}
By construction, $x_n$ is skew-Hermitian.
The restriction of $x_n^*x_n$ to $\mathcal{H}_+$ decomposes as follows with respect to $\mathcal{H}_+ = \mathcal{H}_+^{\textrm{even}} \oplus \mathcal{H}_+^{\textrm{odd}}$,
$$
x_n^*x_n|_{\mathcal{H}_+} = \left(\begin{array}{cc} u K_n^* K_n  u^* & 0\\ 0 & K_n^* K_n\end{array}\right),
$$
therefore
$x_n$ belongs to $\mathfrak{u}_{1, 2}(\mathcal{H})$ since the singular values of $x_n$ are the singular values of $K_n$ but with doubled multiplicities. Moreover $\|x_n\|_1\leq 2$.

Now  let $y~:\mathcal{H}\rightarrow \mathcal{H}$ be the bounded linear  operator defined by $0$ on $\mathcal{H}_+^{\textrm{even}}$, by $0$  on $\mathcal{H}_-$, and by $y =u$ on $\mathcal{H}_+^{\textrm{odd}}$. Remark that $y$ belongs to $\mathfrak{b}_{\res}^+(\mathcal{H})$. Since $x_n$ and $y$ vanish on $\mathcal{H}_-$ and preserve $\mathcal{H}_+$, one has

$$[x_n, y] = \left(\begin{array}{cc} [x_n, y]|_{\mathcal{H}_+} & 0\\ 0 & 0\end{array}\right),$$
where the operators $[x_n, y]|_{\mathcal{H}_+}$ have the following expression with respect to the decomposition $\mathcal{H}_+ = \mathcal{H}_+^{\textrm{even}} \oplus \mathcal{H}_+^{\textrm{odd}}$,
$$
[x_n, y]|_{\mathcal{H}_+} =  \left(\begin{array}{cc} u K_n^* u^* & 0\\ 0 & -K_n^*\end{array}\right).
$$
It follows that 
\begin{equation}\label{Tn}
\|T_+([x_n, y]|_{\mathcal{H}_+})\|_1 = 2\|T_+(K_n)\|_1\geq n,
\end{equation}
 hence $\|T_+([x_n, y]|_{\mathcal{H}_+})\|_1 \rightarrow +\infty$.
\end{proof}

\begin{lemma}\label{lemma_coadjoint_ext} 
Let $x_n\in\mathfrak{u}_{1,2}(\mathcal{H})$ and $y\in \mathfrak{b}_{\res}^+(\mathcal{H})$ be as in the proof of Proposition~\ref{ex2}. Then $\|x_n\|_{\mathfrak{u}_{1,2}}\leq 2$ but $\|\ad^*_y x_n \|_{\mathfrak{u}_{1,2}}\rightarrow + \infty$. Consequently the coadjoint action of  $\mathfrak{b}_{\res}^+(\mathcal{H})$ on $\mathfrak{u}_{1, 2}(\mathcal{H})$ is unbounded.
\end{lemma}

\begin{proof}
Consider the linear forms $\alpha_n$ on $\mathfrak{b}_{\res}^+(\mathcal{H})$ given by $\alpha_n(A) =  \Im \Tr x_n A$ for  $x_n\in \mathfrak{u}_{1,2}(\mathcal{H})$ defined by \eqref{x}. Then
the linear forms $\beta_n =  -\ad^*_y \alpha_n$ read 
$$
\beta_n(A) = -\ad^*_y \alpha_n(A) = -\alpha_n(\ad_y A) = - \alpha_n([y, A]) = -\Im \Tr x [y, A] = -\Im \Tr (x y A -  x A y).
$$
According to Proposition~2.1 in \cite{GO10}, one has $\Tr x A y = \Tr y x A$, therefore $$\beta_n(A) = -\Im \Tr [x_n, y] A.$$
The unique skew-symmetric operator $T_n$ such that $-\Im \Tr T_n A = -\Im \Tr [x_n, y] A$ for any $A$ in the subspace $\mathfrak{b}_2^+(\mathcal{H})$ of $\mathfrak{b}_{\res}^+(\mathcal{H})$ is 
$$
\begin{array}{ll}
T_n & = p_{\mathfrak{u}_2^+}([x_n, y]) = T_{--}([x_n, y]) -T_{--}([x_n, y])^* + \frac{1}{2}\left(D([x_n, y]) - D([x_n, y])^*\right)\\
\end{array}
$$
Since $K_n$ are Hermitian, $[x_n, y]|_{\mathcal{H}_+}$ are Hermitian and we get
$$
T_n   = [x_n, y] - 2  T_{+}([x_n, y]) + D([x_n, y]).
$$
In particular,
$$
2  T_{+}([x_n, y]) = T_n - [x_n, y] - D([x_n, y]).
$$
By equation~\eqref{Tn}, $2  T_{+}([x_n, y]) \geq 2n.$
Therefore
$$
\| T_n\|_{\mathfrak{u}_{1,2}} + \|[x_n, y]\|_{\mathfrak{u}_{1,2}} + \|D([x_n, y])\|_{\mathfrak{u}_{1,2}} \geq \| T_n - [x_n, y] - D([x_n, y])\|_{\mathfrak{u}_{1,2}} \geq 2n,
$$
for all $n \in \mathbb{N}$, and
$$
\| T_n\|_{\mathfrak{u}_{1,2}} \geq 2n - 2 - \|D([x_n, y])\|_{\mathfrak{u}_{1,2}}.
$$
The operator $D$ consisting in taking the diagonal is bounded in $L_1(\mathcal{H})$ with operator norm less than $1$
(see Theorem~1.19 in \cite{Sim79} or \cite{GK70} page 134),
therefore
$$
\| T_n\|_{\mathfrak{u}_{1,2}} > 2n- 4.
$$
It follows that $\|-\ad^*_y \alpha_n\|_{\mathfrak{u}_{1,2}} = \| T_n\|_{\mathfrak{u}_{1,2}}  \rightarrow + \infty.$

\end{proof}

Using the same kind of arguments (see also \cite{ABT3}), we have~:
\begin{lemma}\label{lemma_coadjoint_ext2} 
The coadjoint action of 
 $\mathfrak{u}_{\res}(\mathcal{H})$ on 
 $\mathfrak{b}_{1, 2}^+(\mathcal{H})$ is unbounded.
 \end{lemma}


%
%
%
From the previous discussion, we obtain the following theorems.

\begin{theorem}\label{ccl}
The Banach Lie algebra $\mathfrak{u}_{1, 2}(\mathcal{H})$ is not a Banach Lie--Poisson space with respect to $\mathfrak{b}_{\res}^\pm(\mathcal{H})$. Consequently there is no  Banach Manin triple structure on the triple of Banach Lie algebras $\left(\mathfrak{b}_{\res}^\pm(\mathcal{H})\oplus\mathfrak{u}_{1, 2}(\mathcal{H}), \mathfrak{b}_{\res}^\pm(\mathcal{H}), \mathfrak{u}_{1, 2}(\mathcal{H})\right)$ for the duality pairing defined in Proposition~\ref{duality_U_b}.
\end{theorem}

\begin{proof}
The Banach space $\mathfrak{u}_{1, 2}(\mathcal{H})$ is not a Banach Lie--Poisson space with respect to  $\mathfrak{b}_{\res}^\pm(\mathcal{H})$ as a consequence of Lemma~\ref{lemma_coadjoint_ext}. 
By Theorem~\ref{bialgebra_to_manin}, there is no Banach Manin triple structure on the Banach Lie algebras $\left(\mathfrak{u}_{1, 2}(\mathcal{H})\oplus\mathfrak{b}^\pm_{\res}(\mathcal{H}), \mathfrak{u}_{1, 2}(\mathcal{H}), \mathfrak{b}_{\res}^\pm(\mathcal{H})\right)$ for the duality pairing given by the imaginary part of the restricted trace.
\end{proof}

Along the same lines, we have the analoguous Theorem~:
\begin{theorem}\label{cc2}
The Banach Lie algebra $\mathfrak{b}_{1, 2}^\pm(\mathcal{H})$ is not a Banach Lie--Poisson space with respect to  $\mathfrak{u}_{\res}(\mathcal{H})$. Consequently there is no Banach Manin triple structure on the triple of Banach Lie algebras $\left(\mathfrak{b}_{1, 2}^\pm(\mathcal{H})\oplus\mathfrak{u}_{\res}(\mathcal{H}), \mathfrak{b}_{1, 2}^\pm(\mathcal{H}), \mathfrak{u}_{\res}(\mathcal{H})\right)$ for the duality pairing defined  in Proposition~\ref{duality_U_b}.
\end{theorem}

\subsection{The Banach Poisson--Lie groups $\operatorname{B}_{\res}^{\pm}(\mathcal{H})$ and $\operatorname{U}_{\res}(\mathcal{H})$}\label{Bres_section}

In this Section we will construct a Banach Poisson--Lie group structure on the Banach Lie group $\operatorname{B}_{\res}^{+}(\mathcal{H})$. A similar construction can be of course carried out for the Banach Lie group $\operatorname{B}_{\res}^{-}(\mathcal{H})$ instead. Recall that the coajoint action of $\operatorname{B}_{\res}^+(\mathcal{H})$ is unbounded on  $\mathfrak{u}_{1,2}(\mathcal{H})$  (see Section~\ref{no_Manin}, in particular Lemma~\ref{lemma_coadjoint_ext}). Therefore, in order to construct a Poisson--Lie group structure on $\operatorname{B}_{\res}^{+}(\mathcal{H})$, we need a larger 
subspace of the dual $\mathfrak{b}_\res^+(\mathcal{H})^*$ which will play the role of $\mathfrak{g}_-:= \mathbb{F}_e$ (compare with Theorem~\ref{GPoisson}). Consider the following map~:
$$
\begin{array}{llll}
F~: &L_{1,2}(\mathcal{H}) &\rightarrow &\mathfrak{b}_\res^+(\mathcal{H})^*\\
& a & \mapsto & \left(b \mapsto \Im \Tr a b\right).
\end{array}
$$

\begin{proposition}
The kernel of $F$ equals $\mathfrak{b}_{1,2}^+(\mathcal{H})$, therefore $L_{1,2}(\mathcal{H})/\mathfrak{b}_{1,2}^+(\mathcal{H})$ injects into the dual space $\mathfrak{b}_\res^+(\mathcal{H})^*$. Moreover $L_{1,2}(\mathcal{H})/\mathfrak{b}_{1,2}^+(\mathcal{H})$  is preserved by the coadjoint action of $\operatorname{B}_{\res}^{+}(\mathcal{H})$ and strictly contains $\mathfrak{u}_{1,2}(\mathcal{H})$ as a dense subspace.
\end{proposition}

\begin{proof}
In order to show that the kernel of $F$ is $\mathfrak{b}_{1,2}^+(\mathcal{H})$, consider, for $n\geq m$, the operator $E_{nm}=|n\rangle\langle m| \in \mathfrak{b}_\res^+(\mathcal{H})$ given by $ x\mapsto \langle m| x\rangle |n\rangle$ and,  for $n>m$, the operator $iE_{nm}\in \mathfrak{b}_\res^+(\mathcal{H})$. 
As in the proof of Proposition~\ref{duality_U_b}, an element $a\in L_{1,2}(\mathcal{H})$ satisfying $F(a)(E_{nm}) = 0$ and $F(a)(i E_{nm}) = 0$ is such that $\langle m| an\rangle = 0$ for $n>m$ and $\langle n| an \rangle \in \mathbb{R}$ for $n\in \mathbb{Z}$, i.e. belongs to $\mathfrak{b}_{1,2}^+(\mathcal{H})$. Let us show that the range of $F$ is preserved by the coadjoint action of $\operatorname{B}_{\res}^{+}(\mathcal{H})$. Let $g\in \operatorname{B}_{\res}^{+}(\mathcal{H})$ and $a\in L_{1,2}(\mathcal{H})$. For any $b\in \mathfrak{b}_\res^+(\mathcal{H})$, one has~:
$$
\begin{array}{ll}
\textrm{Ad}^*(g)F(a)(b) &= F(a)(\textrm{Ad}(g)(b)) = F(a)(g b g^{-1}) \\&= \Im \Tr a g b g^{-1} = \Im \Tr g^{-1} a g b = F(g^{-1} a g)(b),
\end{array}
$$
where, in the fourth equality, we have used Proposition~2.1 in \cite{GO10} (since the product $a g b$ belongs to $L_{1,2}(\mathcal{H})$ and $b $ to $L_{\res}(\mathcal{H})$). In fact,  $\operatorname{B}_{\res}^{+}(\mathcal{H})$ acts continuously on the right on $L_{1,2}(\mathcal{H})$ by
$$
a\cdot g = g^{-1} a g.
$$
Then one has the equivariance property 
$$
F(a\cdot g) = \textrm{Ad}^*(g)F(a).
$$
Moreover the subalgebra $\mathfrak{b}_{1,2}^+(\mathcal{H})$ is preserved by the right action of $\operatorname{B}_{\res}^{+}(\mathcal{H})$ on $L_{1,2}(\mathcal{H})$. It follows that there is a well-defined right action of $\operatorname{B}_{\res}^{+}(\mathcal{H})$ on the quotient space $L_{1,2}(\mathcal{H})/\mathfrak{b}_{1,2}^+(\mathcal{H})$ defined by
$$
[a]\cdot g = [a\cdot g],
$$
where $[a]$ denotes the class of $a\in L_{1,2}(\mathcal{H})$ modulo $\mathfrak{b}_{1,2}^+(\mathcal{H})$. 

Let us show that $\mathfrak{u}_{1,2}(\mathcal{H})\oplus \mathfrak{b}_{1,2}^+(\mathcal{H})$ is dense in $L_{1,2}(\mathcal{H})$. To do this, we will show that any continuous linear form on $L_{1,2}(\mathcal{H})$ which vanishes on $\mathfrak{u}_{1,2}(\mathcal{H})\oplus \mathfrak{b}_{1,2}^+(\mathcal{H})$ is equal to the zero form. Recall that the dual space of $L_{1,2}(\mathcal{H})$ is $L_{\res}(\mathcal{H})$, the duality pairing being given by the restricted trace. Consider $X\in L_{\res}(\mathcal{H})$ such that $\Tr X a = 0$ and $\Tr Xb = 0$ for any $a\in \mathfrak{u}_{1,2}(\mathcal{H})$ and any $b\in \mathfrak{b}_{1,2}^+(\mathcal{H})$. Letting $b = E_{nm}$ with $n\geq m$, we get  $\langle m| Xn\rangle = 0$ for $n\geq m$. Letting $a = E_{nm}-E_{mn}\in \mathfrak{u}_{1,2}(\mathcal{H})$, we get $\langle m| Xn\rangle - \langle n| Xm\rangle= 0$ for $n\geq m$. It follows that $\langle m| Xn\rangle = 0$ for any $m, n \in \mathbb{Z}$, which implies that the bounded linear operator $X$ vanishes, hence $\mathfrak{u}_{1,2}(\mathcal{H})\oplus \mathfrak{b}_{1,2}^+(\mathcal{H})$ is dense in $L_{1,2}(\mathcal{H})$. It follows from Section~\ref{no_Manin},  that $\mathfrak{u}_{1,2}(\mathcal{H})\oplus \mathfrak{b}_{1,2}^+(\mathcal{H})$ is strictly contained in $L_{1,2}(\mathcal{H})$.

Let us show that $\mathfrak{u}_{1,2}(\mathcal{H})$ is dense in $L_{1,2}(\mathcal{H})/\mathfrak{b}_{1,2}^+(\mathcal{H})$. Consider a class $[a]\in L_{1,2}(\mathcal{H})/\mathfrak{b}_{1,2}^+(\mathcal{H})$, where $a$ is any element in $L_{1,2}(\mathcal{H})$. Since $\mathfrak{u}_{1,2}(\mathcal{H})\oplus \mathfrak{b}_{1,2}^+(\mathcal{H})$ is dense in $L_{1,2}(\mathcal{H})$, there is a sequence $u_i\in \mathfrak{u}_{1,2}(\mathcal{H})$ and a sequence $b_i\in \mathfrak{b}_{1,2}^+(\mathcal{H})$ such that $u_i+ b_i$ converge to $a$ in $L_{1,2}(\mathcal{H})$. It follows that $[u_i+b_i] = [u_i]$ converge to $[a]$ in $L_{1,2}(\mathcal{H})/\mathfrak{b}_{1,2}^+(\mathcal{H})$.
\end{proof}

Now we are able to state the following Theorem. The proof uses Lemma~\ref{lemmaJacobi}.
\begin{theorem}\label{BresPoisson}
Consider the Banach Lie group $\operatorname{B}_{\res}^{+}(\mathcal{H})$, and 
\begin{enumerate}
\item $\mathfrak{g}_- := \operatorname{L}_{1,2}(\mathcal{H})/\mathfrak{b}_{1,2}^+(\mathcal{H})\subset  \mathfrak{b}_\res^+(\mathcal{H})^*$,
\item $\mathbb{B}\subset T^*\operatorname{B}_{\res}^{+}(\mathcal{H})$, $\mathbb{B}_b:= R_{b^{-1}}^*\mathfrak{g}_-$,
\item $\Pi^{B^+_\res}_r~:\operatorname{B}_{\res}^{+}(\mathcal{H})\rightarrow \Lambda^2\mathfrak{g}_-^*$ defined by
$$
\Pi^{B^+_\res}_r(b)([x_1]_{\mathfrak{b}_{1,2}^+}, [x_2]_{\mathfrak{b}_{1,2}^+}) = \Im\Tr(b^{-1}\,p_{\mathfrak{u}^+_2}(x_1)\,b)\left[ p_{\mathfrak{b}_2^+}(b^{-1}\,p_{\mathfrak{u}^+_2}(x_2)\,b)\right],
$$
\item $\pi^{B^+_\res}(b) = R^{**}_b \Pi^{B^+_\res}_r(b)$.
\end{enumerate}
Then $\left(\operatorname{B}_{\res}^{+}(\mathcal{H}), \mathbb{B}, \pi^{B^+_\res}\right)$ is a Banach Poisson--Lie group.
\end{theorem}

\begin{proof}
\begin{itemize}
\item Let us show that $\Pi^{B^+_\res}_{r}$ satisfies the cocycle condition. 
\begin{equation*}
\begin{array}{l}
\Pi^{B^+_\res}_r(u)\left(\textrm{Ad}^*(g)[x_1]_{\mathfrak{b}_{1,2}^+}, \textrm{Ad}^*(g)[x_2]_{\mathfrak{b}_{1,2}^+}\right)  =  \Pi^{B^+_\res}_r(u)\left([g^{-1}\,x_1\,g]_{\mathfrak{b}_{1,2}^+},[g^{-1}x_2 \,g]_{\mathfrak{b}_{1,2}^+}\right) 
\\ \quad = \Im \Tr (u^{-1} p_{\mathfrak{u}_2^+}(g^{-1} x_1 \,g)\,u) \left[p_{\mathfrak{b}_2^+}(u^{-1} p_{\mathfrak{u}_2^+}(g^{-1} x_2\, g) u)\right]
\end{array}\end{equation*}
Using the decomposition $p_{\mathfrak{u}_2^+}(g^{-1} x_1 \,g) = g^{-1} x_1 \,g - p_{\mathfrak{b}_2^+}(g^{-1} x_1 \,g) $, the fact that $\mathfrak{b}_2^+$ is preserved by conjugation by elements in $\operatorname{B}_{\res}^{+}(\mathcal{H})$, and the fact that  $\mathfrak{b}_2^+$ is isotropic, one has~:
\begin{equation*}
\begin{array}{l}
\Pi^{B^+_\res}_r(u)\left(\textrm{Ad}^*(g)[x_1]_{\mathfrak{b}_{1,2}^+}, \textrm{Ad}^*(g)[x_2]_{\mathfrak{b}_{1,2}^+}\right)  
 = \Im \Tr (u^{-1} g^{-1} x_1 \,g \,u) \left[p_{\mathfrak{b}_2^+}(u^{-1} p_{\mathfrak{u}_2^+}(g^{-1} x_2\, g) u)\right]\\
 \quad =  \Im \Tr (u^{-1} g^{-1} x_1 \,g \,u) \left[p_{\mathfrak{b}_2^+}(u^{-1} g^{-1} x_2\, g u)\right] -  \Im \Tr (u^{-1} g^{-1} x_1 \,g \,u) \left[p_{\mathfrak{b}_2^+}(u^{-1} p_{\mathfrak{b}_2^+}(g^{-1} x_2\, g) u)\right]\\
\quad =  \Im \Tr (u^{-1} g^{-1} x_1 \,g \,u) \left[p_{\mathfrak{b}_2^+}(u^{-1} g^{-1} x_2\, g u)\right] -  \Im \Tr  g^{-1} x_1 \,g\, p_{\mathfrak{b}_2^+}(g^{-1} x_2\, g)\\
\end{array}\end{equation*}
Using the decompositions $x_1 =  p_{\mathfrak{u}_2^+}(x_1) +  p_{\mathfrak{b}_2^+}(x_1)$ and $x_2 =  p_{\mathfrak{u}_2^+}(x_2) +  p_{\mathfrak{b}_2^+}(x_2)$, one gets 8 terms but 4 of them vanish since $\mathfrak{b}_2^+$ is isotropic. The remaining terms are:
\begin{equation*}
\begin{array}{ll}
\Pi^{B^+_\res}_r(u)\left(\textrm{Ad}^*(g)[x_1]_{\mathfrak{b}_{1,2}^+}, \textrm{Ad}^*(g)[x_2]_{\mathfrak{b}_{1,2}^+}\right)  
& =  \Im \Tr (u^{-1} g^{-1} p_{\mathfrak{u}_2^+}(x_1) \,g \,u) \left[p_{\mathfrak{b}_2^+}(u^{-1} g^{-1} p_{\mathfrak{u}_2^+}(x_2)\, g u)\right] \\ & \,\,+  \Im \Tr (u^{-1} g^{-1} p_{\mathfrak{u}_2^+}(x_1) \,g \,u) \left[p_{\mathfrak{b}_2^+}(u^{-1} g^{-1} p_{\mathfrak{b}_2^+}(x_2)\, g u)\right] \\
& 
\,\,-\Im \Tr  g^{-1} p_{\mathfrak{u}_2^+}(x_1) \,g\, p_{\mathfrak{b}_2^+}(g^{-1} p_{\mathfrak{u}_2^+}(x_2)\, g) \\ &\,\, - \Im \Tr  g^{-1} p_{\mathfrak{u}_2^+}(x_1) \,g\, p_{\mathfrak{b}_2^+}(g^{-1} p_{\mathfrak{b}_2^+}(x_2)\, g)\\
\end{array}\end{equation*}
The first term in the right hand side equals $\Pi^{B^+_\res}_r(gu)([x_1]_{\mathfrak{b}_{1,2}^+}, [x_2]_{\mathfrak{b}_{1,2}^+})$, the third term equals $-\Pi^{B^+_\res}_r(g)([x_1]_{\mathfrak{b}_{1,2}^+}, [x_2]_{\mathfrak{b}_{1,2}^+})$, whereas the second term equals $+\Im \Tr(p_{\mathfrak{u}_2^+}(x_1) p_{\mathfrak{b}_2^+}(x_2))$, and the last terms equals $-\Im \Tr(p_{\mathfrak{u}_2^+}(x_1) p_{\mathfrak{b}_2^+}(x_2))$.

\item It remains to check that $\pi^{B^+_\res}$ satisfies the Jacobi identity \eqref{Jacobi_Poisson}. We will use Lemma~\ref{lemmaJacobi}. Using the cocycle identity, one has for any $X$ in $\mathfrak{b}_{\res}^+(\mathcal{H})$ and $g\in \operatorname{B}_{\res}^+$,
\begin{equation*}
T_g\Pi^{B^+_\res}_r(L_{g}X)([x_1], [x_2]) = T_e\Pi^{B^+_\res}_r(X)(\textrm{Ad}^*(g)[x_1], \textrm{Ad}^*(g)[x_2]),
\end{equation*}
in particular, 
\begin{equation*}
\begin{array}{ll}
T_g\Pi^{B^+_\res}_r(R_{g}X)([x_1], [x_2]) &=  T_g\Pi^{B^+_\res}_r(L_{g}\textrm{Ad}(g^{-1})(X))([x_1], [x_2]) \\ & =T_e\Pi^{B^+_\res}_r(\textrm{Ad}(g^{-1})(X))(\textrm{Ad}^*(g)[x_1], \textrm{Ad}^*(g)[x_2])\\& = 
T_e\Pi^{B^+_\res}_r(\textrm{Ad}(g^{-1})(X))([g^{-1}\,x_1\,g], [g^{-1}\,x_2\,g])
\end{array}
\end{equation*}
On the other hand
\begin{equation*}
\begin{array}{ll}
T_e\Pi^{B^+_\res}_r(Y)([x_1], [x_2]) & = -\Im \Tr [Y, p_{\mathfrak{u}_2^+}(x_1)] p_{\mathfrak{b}_2^+}(p_{\mathfrak{u}_2^+}(x_2)) - \Im \Tr p_{\mathfrak{u}_2^+}(x_1) p_{\mathfrak{b}_2^+}([Y, p_{\mathfrak{u}_2^+}(x_2)])\\
& = - \Im \Tr p_{\mathfrak{u}_2^+}(x_1) [Y, p_{\mathfrak{u}_2^+}(x_2)] = \Im\Tr Y[p_{\mathfrak{u}_2^+}(x_1), p_{\mathfrak{u}_2^+}(x_2)].
\end{array}
\end{equation*}
It follows that 
\begin{equation}\label{Rgpi}
T_g\Pi^{B^+_\res}_r(R_gX)([x_1], [x_2]) = \Im \Tr g^{-1}\,X\, g[p_{\mathfrak{u}_2^+}(g^{-1}\,x_1\,g), p_{\mathfrak{u}_2^+}(g^{-1}\,x_2\,g)].
\end{equation}
In particular, for any $x_1$ and $x_2$ in $\operatorname{L}_{1,2}(\mathcal{H})$, the $1$-form on $\mathfrak{b}_{\res}^+$ given by 
$$X\mapsto T_g\Pi^{B^+_\res}_r(L_gX)([x_1], [x_2])$$
 belongs to $\mathfrak{u}_{1,2}(\mathcal{H})$ and is given by
\begin{equation*}
T_g\Pi^{B^+_\res}_r(L_g(\cdot))([x_1], [x_2])= [p_{\mathfrak{u}_2^+}(g^{-1}\,x_1\,g), p_{\mathfrak{u}_2^+}(g^{-1}\,x_2\,g)]
\end{equation*}
Moreover for $g\in \operatorname{B}_{\res}^{+}(\mathcal{H})$, $x_3\in \operatorname{L}_{1,2}(\mathcal{H})$ and $y\in  \operatorname{L}_{1,2}(\mathcal{H})$, one has
\begin{equation*}
\begin{array}{ll}
\Pi^{B^+_\res}_r(g)([x_3], [y])&  =  \Im\Tr (g^{-1}\,p_{\mathfrak{u}_2^+}(x_3)\, g) p_{\mathfrak{b}_2^+}(g^{-1}\,p_{\mathfrak{u}_2^+}(y)\,g)\\
& =  \Im\Tr p_{\mathfrak{u}_2^+}(g^{-1}\,p_{\mathfrak{u}_2^+}(x_3)\, g) p_{\mathfrak{b}_2^+}(g^{-1}\,p_{\mathfrak{u}_2^+}(y)\,g)\\
& =  -\Im\Tr p_{\mathfrak{b}_2^+}(g^{-1}\,p_{\mathfrak{u}_2^+}(x_3)\, g) p_{\mathfrak{u}_2^+}(g^{-1}\,p_{\mathfrak{u}_2^+}(y)\,g)\\
& =  -\Im\Tr g\,p_{\mathfrak{b}_2^+}(g^{-1}\,p_{\mathfrak{u}_2^+}(x_3)\, g)\,g^{-1}\,p_{\mathfrak{u}_2^+}(y)\\
& =  -\Im\Tr g\,p_{\mathfrak{b}_2^+}(g^{-1}\,p_{\mathfrak{u}_2^+}(x_3)\, g)\,g^{-1}(y)\\
%
\end{array}
\end{equation*}
%
In particular 
$i_{[x_3]}\Pi^{B^+_\res}_r(g) =- g\,p_{\mathfrak{b}_2^+}(g^{-1}\,p_{\mathfrak{u}_2^+}(x_3)\, g)\,g^{-1}$ 
belongs to $\mathfrak{b}_2^+(\mathcal{H})\subset\mathfrak{b}_{\res}^+(\mathcal{H})$. 
Using \eqref{Rgpi}, it follows that 
\begin{equation}\label{Teq}
\begin{array}{ll}
T_g\Pi^{B^+_\res}_r(R_g i_{[x_3]}\Pi^{B^+_\res}_r(g))([x_1], [x_2]) &=
-\Im \Tr \,p_{\mathfrak{b}_2^+}(g^{-1}\,p_{\mathfrak{u}_2^+}(x_3)\, g)[p_{\mathfrak{u}_2^+}(g^{-1}\,x_1\,g), p_{\mathfrak{u}_2^+}(g^{-1}\,x_2\,g)]\\
&= -\Im \Tr \,p_{\mathfrak{b}_2^+}(g^{-1}\,p_{\mathfrak{u}_2^+}(x_3)\, g)[p_{\mathfrak{u}_2^+}(g^{-1}\,p_{\mathfrak{u}_2^+}(x_1)\,g), p_{\mathfrak{u}_2^+}(g^{-1}\,p_{\mathfrak{u}_2^+}(x_2)\,g)],\\
\end{array}
\end{equation}
where we have used that $g^{-1}\,p_{b_2^+}(x_i)\,g\in \mathfrak{b}_2^+$ for any $x_i\in \operatorname{L}_{1,2}(\mathcal{H})$ and any $g\in \operatorname{B}_{\res}^+(\mathcal{H})$. Moreover
\begin{equation}\label{lr}
\begin{array}{ll}
\langle x_1, [i_{[x_3]}\Pi^{B^+_\res}_r(g), i_{[x_2]}\Pi^{B^+_\res}_r(g)]\rangle & = \Im \Tr x_1 [g\,p_{\mathfrak{b}_2^+}(g^{-1}\,p_{\mathfrak{u}_2^+}(x_3)\, g)\,g^{-1}, g\,p_{\mathfrak{b}_2^+}(g^{-1}\,p_{\mathfrak{u}_2^+}(x_2)\, g)\,g^{-1}]\\
&= \Im \Tr p_{\mathfrak{u}_2^+}(x_1) [g\,p_{\mathfrak{b}_2^+}(g^{-1}\,p_{\mathfrak{u}_2^+}(x_3)\, g)\,g^{-1}, g\,p_{\mathfrak{b}_2^+}(g^{-1}\,p_{\mathfrak{u}_2^+}(x_2)\, g)\,g^{-1}]\\
& = \Im \Tr g^{-1} p_{\mathfrak{u}_2^+}(x_1)  g [p_{\mathfrak{b}_2^+}(g^{-1}\,p_{\mathfrak{u}_2^+}(x_3)\, g), p_{\mathfrak{b}_2^+}(g^{-1}\,p_{\mathfrak{u}_2^+}(x_2)\, g)]\\
& = \Im \Tr p_{\mathfrak{u}_2^+}(g^{-1} p_{\mathfrak{u}_2^+}(x_1)  g) [p_{\mathfrak{b}_2^+}(g^{-1}\,p_{\mathfrak{u}_2^+}(x_3)\, g), p_{\mathfrak{b}_2^+}(g^{-1}\,p_{\mathfrak{u}_2^+}(x_2)\, g)]\\
& = -\Im \Tr p_{\mathfrak{u}_2^+}(g^{-1} p_{\mathfrak{u}_2^+}(x_1)  g) [p_{\mathfrak{b}_2^+}(g^{-1}\,p_{\mathfrak{u}_2^+}(x_2)\, g), p_{\mathfrak{b}_2^+}(g^{-1}\,p_{\mathfrak{u}_2^+}(x_3)\, g)]
\end{array}
\end{equation}
Consider $\alpha = R_{g^{-1}}^*[x_1]\in (T_g\operatorname{B}_{\res}^+)^*$, $\beta = R_{g^{-1}}^*[x_2]\in (T_g\operatorname{B}_{\res}^+)^*$ and $\gamma = R_{g^{-1}}^*[x_3]\in (T_g\operatorname{B}_{\res}^+)^*$, for $x_1$, $x_2$ and $x_3$ in $\operatorname{L}_{1,2}(\mathcal{H})$. Injecting \eqref{Teq} and \eqref{lr} into \eqref{Jacobi_tensor} and using the fact that the left hand side of \eqref{Jacobi_tensor}  defines a tensor, one gets :
\begin{equation*}
\begin{array}{l}
\!\!\!\!\!\pi\left(\alpha, d\left(\pi(\beta, \gamma)\right)\right) + \pi\left(\beta, d\left(\pi(\gamma, \alpha)\right)\right) + \pi\left(\gamma, d\left(\pi(\alpha, \beta)\right) \right) \\
= -\Im \Tr \,p_{\mathfrak{b}_2^+}(g^{-1}\,p_{\mathfrak{u}_2^+}(x_3)\, g)[p_{\mathfrak{u}_2^+}(g^{-1}\,p_{\mathfrak{u}_2^+}(x_1)\,g), p_{\mathfrak{u}_2^+}(g^{-1}\,p_{\mathfrak{u}_2^+}(x_2)\,g)]\\
\quad- \Im \Tr p_{\mathfrak{u}_2^+}(g^{-1} p_{\mathfrak{u}_2^+}(x_1) g) [p_{\mathfrak{b}_2^+}(g^{-1}\,p_{\mathfrak{u}_2^+}(x_2)\, g), p_{\mathfrak{b}_2^+}(g^{-1}\,p_{\mathfrak{u}_2^+}(x_3)\, g)]\\
\quad-\Im \Tr \,p_{\mathfrak{b}_2^+}(g^{-1}\,p_{\mathfrak{u}_2^+}(x_1)\, g)[p_{\mathfrak{u}_2^+}(g^{-1}\,p_{\mathfrak{u}_2^+}(x_2)\,g), p_{\mathfrak{u}_2^+}(g^{-1}\,p_{\mathfrak{u}_2^+}(x_3)\,g)]\\
\quad- \Im \Tr p_{\mathfrak{u}_2^+}(g^{-1} p_{\mathfrak{u}_2^+}(x_2) g) [p_{\mathfrak{b}_2^+}(g^{-1}\,p_{\mathfrak{u}_2^+}(x_3)\, g), p_{\mathfrak{b}_2^+}(g^{-1}\,p_{\mathfrak{u}_2^+}(x_1)\, g)]\\
\quad-\Im \Tr \,p_{\mathfrak{b}_2^+}(g^{-1}\,p_{\mathfrak{u}_2^+}(x_2)\, g)[p_{\mathfrak{u}_2^+}(g^{-1}\,p_{\mathfrak{u}_2^+}(x_3)\,g), p_{\mathfrak{u}_2^+}(g^{-1}\,p_{\mathfrak{u}_2^+}(x_1)\,g)]\\
\quad- \Im \Tr p_{\mathfrak{u}_2^+}(g^{-1} p_{\mathfrak{u}_2^+}(x_3) g) [p_{\mathfrak{b}_2^+}(g^{-1}\,p_{\mathfrak{u}_2^+}(x_1)\, g), p_{\mathfrak{b}_2^+}(g^{-1}\,p_{\mathfrak{u}_2^+}(x_2)\, g)]\\
\\
= -\Im \Tr \,p_{\mathfrak{b}_2^+}(g^{-1}\,p_{\mathfrak{u}_2^+}(x_3)\, g)[p_{\mathfrak{u}_2^+}(g^{-1}\,p_{\mathfrak{u}_2^+}(x_1)\,g), p_{\mathfrak{u}_2^+}(g^{-1}\,p_{\mathfrak{u}_2^+}(x_2)\,g)]\\
\quad- \Im \Tr  p_{\mathfrak{b}_2^+}(g^{-1}\,p_{\mathfrak{u}_2^+}(x_3)\, g)[p_{\mathfrak{u}_2^+}(g^{-1} p_{\mathfrak{u}_2^+}(x_1) g) , p_{\mathfrak{b}_2^+}(g^{-1}\,p_{\mathfrak{u}_2^+}(x_2)\, g)]\\
\quad-\Im \Tr \,p_{\mathfrak{u}_2^+}(g^{-1}\,p_{\mathfrak{u}_2^+}(x_3)\,g)[p_{\mathfrak{b}_2^+}(g^{-1}\,p_{\mathfrak{u}_2^+}(x_1)\, g), p_{\mathfrak{u}_2^+}(g^{-1}\,p_{\mathfrak{u}_2^+}(x_2)\,g)] \\
\quad- \Im \Tr  p_{\mathfrak{b}_2^+}(g^{-1}\,p_{\mathfrak{u}_2^+}(x_3)\, g)[ p_{\mathfrak{b}_2^+}(g^{-1}\,p_{\mathfrak{u}_2^+}(x_1)\, g), p_{\mathfrak{u}_2^+}(g^{-1} p_{\mathfrak{u}_2^+}(x_2) g)]\\
\quad-\Im \Tr \,p_{\mathfrak{u}_2^+}(g^{-1}\,p_{\mathfrak{u}_2^+}(x_3)\,g) [p_{\mathfrak{u}_2^+}(g^{-1}\,p_{\mathfrak{u}_2^+}(x_1)\,g), p_{\mathfrak{b}_2^+}(g^{-1}\,p_{\mathfrak{u}_2^+}(x_2)\, g)]\\
\quad- \Im \Tr p_{\mathfrak{u}_2^+}(g^{-1} p_{\mathfrak{u}_2^+}(x_3) g) [p_{\mathfrak{b}_2^+}(g^{-1}\,p_{\mathfrak{u}_2^+}(x_1)\, g), p_{\mathfrak{b}_2^+}(g^{-1}\,p_{\mathfrak{u}_2^+}(x_2)\, g)]\\
\\
= -\Im \Tr \, g^{-1} p_{\mathfrak{u}_2^+}(x_3) g) [ g^{-1}\,p_{\mathfrak{u}_2^+}(x_1)\,g, g^{-1}\,p_{\mathfrak{u}_2^+}(x_2)\, g]\\
= -\Im \Tr g^{-1} p_{\mathfrak{u}_2^+}(x_3)[p_{\mathfrak{u}_2^+}(x_1), p_{\mathfrak{u}_2^+}(x_2)] g\\
= -\Im \Tr p_{\mathfrak{u}_2^+}(x_3)[p_{\mathfrak{u}_2^+}(x_1), p_{\mathfrak{u}_2^+}(x_2)] \\
= 0,
\end{array}
\end{equation*}
hence $\pi$ is a Poisson tensor.
\end{itemize}
\end{proof}

\begin{remark}{\rm
In the proof of the previous Theorem, we have established that
\begin{equation*}
\begin{array}{ll}
T_e\Pi^{B^+_\res}_r(Y)([x_1]_{\mathfrak{b}_{1,2}^+}, [x_2]_{\mathfrak{b}_{1,2}^+}) & =  \Im\Tr Y[p_{\mathfrak{u}_2^+}(x_1), p_{\mathfrak{u}_2^+}(x_2)],
\end{array}
\end{equation*}
where $x_1$, $x_2\in \operatorname{L}_{1,2}(\mathcal{H})$ and $Y\in \mathfrak{b}_{\res}^+(\mathcal{H})$. It follows  that $T_e\Pi^{B^+_\res}_r$ is the dual map of
\begin{equation}\label{bracketb}
\begin{array}{cll}
\operatorname{L}_{1,2}(\mathcal{H})/\mathfrak{b}_{1,2}^+(\mathcal{H}) \times \operatorname{L}_{1,2}(\mathcal{H})/\mathfrak{b}_{1,2}^+(\mathcal{H}) & \rightarrow & \operatorname{L}_{1,2}(\mathcal{H})/\mathfrak{b}_{1,2}^+(\mathcal{H})\\
([x_1]_{\mathfrak{b}_{1,2}^+}, [x_2]_{\mathfrak{b}_{1,2}^+}) & \mapsto & [p_{\mathfrak{u}_2^+}(x_1), p_{\mathfrak{u}_2^+}(x_2)],
\end{array}
\end{equation}
which is well defined  on $\operatorname{L}_{1,2}(\mathcal{H})/\mathfrak{b}_{1,2}^+(\mathcal{H})$ since $[p_{\mathfrak{u}_2^+}(x_1), p_{\mathfrak{u}_2^+}(x_2)]\in \operatorname{L}_1(\mathcal{H})$ for any  $x_1, x_2\in \operatorname{L}_{1,2}(\mathcal{H})$. Note that this bracket is continuous and extends the natural bracket of $\mathfrak{u}_{1,2}(\mathcal{H})$.
}
\end{remark}


Along the same lines  (see also \cite{ABT3}), we obtain the following Theorem~:
\begin{theorem}\label{UresPoisson}
Consider the Banach Lie group $\operatorname{U}_{\res}(\mathcal{H})$, and 
\begin{enumerate}
\item $\mathfrak{g}_+ := \operatorname{L}_{1,2}(\mathcal{H})/\mathfrak{u}_{1,2}(\mathcal{H})\subset  \mathfrak{u}_\res^*(\mathcal{H})$,
\item $\mathbb{U}\subset T^*\operatorname{U}_{\res}(\mathcal{H})$, $\mathbb{U}_g = R_{g^{-1}}^*\mathfrak{g}_+,$
\item $\Pi^{\operatorname{U}_\res}_r~:\operatorname{U}_{\res}(\mathcal{H})\rightarrow \Lambda^2\mathfrak{g}_+^*$ defined by
$$
\Pi^{\operatorname{U}_\res}_r(g)([x_1]_{\mathfrak{u}_{1,2}}, [x_2]_{\mathfrak{u}_{1,2}}) = \Im\Tr(g^{-1}\,p_{\mathfrak{b}^+_2}(x_1)\,g)\left[ p_{\mathfrak{u}_2}(g^{-1}\,p_{\mathfrak{b}^+_2}(x_2)\,g)\right],
$$
\item $\pi^{\operatorname{U}_\res}(g) = R^{**}_g \Pi^{\operatorname{U}_\res}_r(g)$.
\end{enumerate}
Then $\left(\operatorname{U}_{\res}(\mathcal{H}), \mathbb{U}, \pi^{\operatorname{U}_\res}\right)$ is a Banach Poisson--Lie group.
\end{theorem}

%

\section{Bruhat-Poisson structure of the restricted Grassmannian}\label{section7}

In this Section, we construct a generalized Banach Poisson structure on the restricted Grassmannian, and called it Bruhat-Poisson structure by reference to the finite-dimensional picture developped in \cite{LW90}. 
\subsection{A Poisson--Lie subgroup of $\operatorname{U}_{\res}(\mathcal{H})$}

The following definition is identical to the definition in the finite-dimensional case.
\begin{definition}
A Banach Lie subgroup $H$ of a Banach Poisson--Lie group $G$ is called a \textbf{Banach Poisson--Lie subgroup} if it is a Banach Poisson submanifold of $G$, i.e. if it carries a Poisson structure such that the inclusion map $\iota~: H \hookrightarrow G$ is a Poisson map.
\end{definition}

Let us show the following Proposition.
\begin{proposition}
The Banach Lie group $\operatorname{H} := \operatorname{U}(\mathcal{H}_+)\times \operatorname{U}(\mathcal{H}_-)$ is a Poisson--Lie subgroup of $\operatorname{U}_{\res}(\mathcal{H})$.
\end{proposition}

\begin{proof}
Denote by $\iota~: \operatorname{H}  \hookrightarrow \operatorname{U}_{\res}(\mathcal{H})$ the inclusion map.
It is clear that $H$ is a Banach submanifold of $\operatorname{U}_{\res}(\mathcal{H})$. Denote by $\mathfrak{h}$ its Lie algebra.
Recall that $\mathbb{U}$ is the subbundle of $T^*\operatorname{U}_{\res}(\mathcal{H})$ given by $\mathbb{U}_g = R_{g^{-1}}^*\mathfrak{g}_+$ where $\mathfrak{g}_+ : = L_{1,2}(\mathcal{H})/\mathfrak{u}_{1,2}(\mathcal{H})$. Denote by $\langle \cdot, \cdot \rangle_{\mathfrak{u}_{\res}}$ the  duality pairing between $\mathfrak{g}_+$ and $\mathfrak{u}_{\res}(\mathcal{H})$, and  by $\mathfrak{h}^0$ the closed subspace of $\mathfrak{g}_+$ consisting of those covectors in $\mathfrak{g}_+$ which vanish on the closed subspace $\mathfrak{h}$ of $\mathfrak{u}_{\res}(\mathcal{H})$. For any covector $\alpha\in\iota^*\mathfrak{g}_+$ acting on $\mathfrak{h}$, and any vector $X\in\mathfrak{h}$, denote by  $[\alpha]_{\mathfrak{h}^0}$ the class of $\alpha\in i^*\mathfrak{g}_+$ in $i^*\mathfrak{g}_+/\mathfrak{h}^0$. Then the formula
$$
\langle [\alpha]_{\mathfrak{h}^0}, X \rangle_{\mathfrak{h}} := \langle \alpha, X\rangle_{\mathfrak{u}_{\res}},
$$
defines a duality pairing between  $\mathbb{H}_e :=i^*\mathfrak{g}_+/\mathfrak{h}^0$ and $\mathfrak{h}$. It follows that $\mathbb{H} := i^*\mathbb{U}/(T H)^0$ is a subbundle of $T^*H$ in duality with $TH$.
Recall that the Poisson tensor on $\operatorname{U}_{\res}(\mathcal{H})$ is defined as follows
$$
  \Pi^{\operatorname{U}_\res}_r(h)(\alpha, \beta) = \Im \Tr(h^{-1} p_{\mathfrak{b}_2^+}(x_1) h) \left[p_{\mathfrak{u}_2^+}(h^{-1}p_{\mathfrak{b}_2^+}(x_2) h)\right]
$$
where $\alpha, \beta \in\mathfrak{g}_+ = L_{1,2}(\mathcal{H})/\mathfrak{u}_{1,2}$ and $x_1$, $x_2\in L_{1,2}(\mathcal{H})$ are such that $\alpha = [x_1]_{\mathfrak{u}_{1,2}}$ and $\beta = [x_2]_{\mathfrak{u}_{1,2}}$. Note that an element $x_2 = \left(\begin{smallmatrix} A & B\\ C & D\end{smallmatrix}\right)\in L_{1,2}(\mathcal{H})$ belongs to $\mathfrak{h}^0$ if and only if $A \in \mathfrak{u}_{1}(\mathcal{H})$ and $D\in  \mathfrak{u}_{1}(\mathcal{H})$. In that case, one has
$$
x_2 = \left(\begin{smallmatrix} A & -C^*\\ C & D\end{smallmatrix}\right) + \left(\begin{smallmatrix} 0 & B+C^*\\ 0 & 0\end{smallmatrix}\right),
$$
with $p_{\mathfrak{u}_2^+}(x_2) = \left(\begin{smallmatrix} A & -C^*\\ C & D\end{smallmatrix}\right)$ and $p_{\mathfrak{b}_2^+}(x_2) = \left(\begin{smallmatrix} 0 & B+C^*\\ 0 & 0\end{smallmatrix}\right)$. Note also that for any $h = \left(\begin{smallmatrix} h_1 & 0\\ 0 & h_2\end{smallmatrix}\right)\in \operatorname{U}(\mathcal{H}_+)\times\operatorname{U}(\mathcal{H}_-)$, one has
$$
h^{-1} p_{\mathfrak{b}_2^+}(x_2) h = \left(\begin{smallmatrix} 0 & h^{-1}_1(B+C^*)h_2\\ 0 & 0\end{smallmatrix}\right)\in \mathfrak{b}_{2}^+(\mathcal{H}).
$$
It follows that $\Pi^{\operatorname{U}_\res}_r(h)(\cdot, \beta) = 0$ whenever $\beta \in \mathfrak{h}^0$. By skew-symmetry of $\Pi^{\operatorname{U}_\res}_r$, one also  has $\Pi^{\operatorname{U}_\res}_r(h)(\alpha, \cdot) = 0$ whenever $\alpha\in \mathfrak{h}^0$. This allows to define the following map
$$
\Pi^H_r~: H \rightarrow \Lambda^2\mathbb{H}_e^*
$$
by 
$$
\Pi^H_r(h)([\alpha]_{\mathfrak{h}^0}, [\beta]_{\mathfrak{h}^0}) := \Pi^{\operatorname{U}_\res}_r(h)(\alpha, \beta)
$$
for $\alpha, \beta\in\mathfrak{g}_+ = L_{1,2}(\mathcal{H})/\mathfrak{u}_{1,2}$. Set $\pi^H_g := R_g^{**}\Pi^H_r$.  The Jacobi identity for $\pi^H$ follows from the Jacobi identity for $\pi^{\operatorname{U}_\res}$.
By construction, the injection $\iota : H \hookrightarrow \operatorname{U}_{\res}(\mathcal{H})$ is a Poisson map.
\end{proof}
\subsection{The restricted Grassmannian as a quotient Poisson homogeneous space}

\begin{theorem}\label{Poisson}
The restricted Grassmannian $\operatorname{Gr}_{\res}(\mathcal{H}) =  \operatorname{U}_{\res}(\mathcal{H})/\operatorname{U}(\mathcal{H}_+)\times \operatorname{U}(\mathcal{H}_-)$ carries a natural Poisson structure $(\operatorname{Gr}_{\res}(\mathcal{H}), T^*\operatorname{Gr}_{\res}(\mathcal{H}), \pi^{Gr_\res})$ such that~:
\begin{enumerate}
\item the canonical projection $p~: \operatorname{U}_{\res}(\mathcal{H}) \rightarrow \operatorname{Gr}_{\res}(\mathcal{H})$ is a Poisson map,
\item the natural action $\operatorname{U}_{\res}(\mathcal{H})\times\operatorname{Gr}_{\res}(\mathcal{H})\rightarrow \operatorname{Gr}_{\res}(\mathcal{H})$ by left translations is a Poisson map.
\end{enumerate}
\end{theorem}

\begin{proof}\begin{enumerate}
\item
The tangent space at $eH\in \operatorname{Gr}_{\res}(\mathcal{H}) =  \operatorname{U}_{\res}(\mathcal{H})/\operatorname{U}(\mathcal{H}_+)\times \operatorname{U}(\mathcal{H}_-)$ can be identified with the quotient Banach space $\mathfrak{u}_{\res}(\mathcal{H})/\left(\mathfrak{u}(\mathcal{H}_+)\oplus \mathfrak{u}(\mathcal{H}_-)\right)$ which is isomorphic to the Hilbert space $$\mathfrak{m} := \{ \left(\begin{smallmatrix} 0 & A\\ -A^*& 0\end{smallmatrix}\right)\in \mathfrak{u}_2(\mathcal{H})\}.$$
The duality pairing between $\mathfrak{u}_{\res}(\mathcal{H})$ and $\mathfrak{g}_+ = L_{1,2}(\mathcal{H})/\mathfrak{u}_{1, 2}(\mathcal{H})$ induces a strong duality pairing between the quotient space $\mathfrak{u}_{\res}(\mathcal{H})/\left(\mathfrak{u}(\mathcal{H}_+)\oplus \mathfrak{u}(\mathcal{H}_-)\right) = \mathfrak{m}$ and $\mathfrak{h}^0\subset \mathfrak{g}_+$.
 For $\alpha, \beta \in T^*_{gH} \operatorname{Gr}_{\res}(\mathcal{H})$, identify $p^*\alpha\in T^*_g\operatorname{U}_{\res}(\mathcal{H})$ with an element $L_{g^{-1}}^*x_1$ in $ L_{g^{-1}}^*\mathfrak{h}^0$, and $p^*\beta$ with $L_{g^{-1}}^*x_2\in L_{g^{-1}}^*\mathfrak{h}^0$. Define
$$
\begin{array}{ll}
\pi^{Gr_\res}_{gH}(\alpha, \beta) &=\pi^{\operatorname{U}_\res}_g(p^*\alpha, p^*\beta).\end{array}
$$
We have to check that the right hand side is invariant by the natural right action of $H$ on $ \operatorname{U}_{\res}(\mathcal{H})$, which induces an action of $H$ on  forms in $T^*_g\operatorname{U}_{\res}(\mathcal{H})$ by $\gamma\rightarrow R_{h^{-1}}^*\gamma  \in T^*_{gh}\operatorname{U}_{\res}(\mathcal{H}).$ In other words, we have to check that 
\begin{equation}\label{invariance}
\begin{array}{l}
 \pi^{\operatorname{U}_\res}_g((p^*\alpha)_g, (p^*\beta)_g) =  \pi^{\operatorname{U}_\res}_{gh}(R^*_{h^{-1}}(p^*\alpha)_{g}, R^*_{h^{-1}}(p^*\beta)_{g})
 \end{array}
 \end{equation}
 $$
  \Leftrightarrow  \pi^{\operatorname{U}_\res}_g(L_{g^{-1}}^*x_1, L_{g^{-1}}^*x_2) =  \pi^{\operatorname{U}_\res}_{gh}(R^*_{h^{-1}}L_{g^{-1}}^*x_1, R^*_{h^{-1}}L_{g^{-1}}^*x_2)
  $$
  $$  \Leftrightarrow
  \begin{array}{l}
\Pi^{\operatorname{U}_\res}_r(g)(\textrm{Ad}_{g^{-1}}^*x_1, \textrm{Ad}_{g^{-1}}^*x_2) = \Pi^{\operatorname{U}_\res}_r(gh)(R^*_{gh}R^*_{h^{-1}}L_{g^{-1}}^*x_1, R^*_{gh}R^*_{h^{-1}}L_{g^{-1}}^*x_2)\end{array}
$$
Note that $R^*_{gh}\gamma(X) = \gamma(R_{gh}X) = \gamma(Xgh) = R_h^*\gamma(Xg) = R_g^* R_h^* \gamma(X)$. Therefore $R^*_{gh} = R_g^* R_h^*$. It follows that \eqref{invariance} is equivalent to
$$
 \begin{array}{l}
\Pi^{\operatorname{U}_\res}_r(g)(\textrm{Ad}_{g^{-1}}^*x_1, \textrm{Ad}_{g^{-1}}^*x_2) = \Pi^{\operatorname{U}_\res}_r(gh)(\textrm{Ad}_{g^{-1}}^*x_1, \textrm{Ad}_{g^{-1}}^*x_2)
\end{array}
$$
By the cocycle identity
$
 \Pi^{\operatorname{U}_\res}_r(gh) = \textrm{Ad}(g)^{**}\Pi^{\operatorname{U}_\res}_r(h) + \Pi^{\operatorname{U}_\res}_r(g),
$
one has
$$
\begin{array}{ll}
 \Pi^{\operatorname{U}_\res}_r(gh)(\textrm{Ad}_{g^{-1}}^*x_1, \textrm{Ad}_{g^{-1}}^*x_2)
 =&  \Pi^{\operatorname{U}_\res}_r(h)(\textrm{Ad}_{g}^*\textrm{Ad}_{g^{-1}}^*x_1, \textrm{Ad}_{g}^*\textrm{Ad}_{g^{-1}}^*x_2)  \\& + \Pi^{\operatorname{U}_\res}_r(g)(\textrm{Ad}_{g^{-1}}^*x_1, \textrm{Ad}_{g^{-1}}^*x_2)  \\
\end{array}
$$
Since $ \Pi^{\operatorname{U}_\res}_r(h)$ vanishes on $\mathfrak{h}^0$, one has 
$$ \Pi^{\operatorname{U}_\res}_r(h)(\textrm{Ad}_{h^{-1}}^*x_1, \textrm{Ad}_{h^{-1}}^*x_2) =0,$$
therefore equation \eqref{invariance} is satisfied.
The Jacobi identity for $\pi^{Gr_\res}$ follows from the Jacobi identity for $\pi^{\operatorname{U}_\res}$. Moreover $p$ is a Poisson map by construction.

\item 
Consider the action
$$
\begin{array}{lcll}
a_U~: & \operatorname{U}_{\res}(\mathcal{H})\times\operatorname{Gr}_{\res}(\mathcal{H})& \rightarrow & \operatorname{Gr}_{\res}(\mathcal{H})\\
& (g_1, gH) & \mapsto & g_1 g H
\end{array}
$$ by left translations.
Note that the tangent map to $a_U$ 
is given by
$$
\begin{array}{lcll}
T_{(g_1, gH)} a_U~:& T_{g_1}\operatorname{U}_{\res}(\mathcal{H}) \oplus T_{gH}\operatorname{Gr}_{\res}(\mathcal{H})& \rightarrow &T_{g_1gH}\operatorname{Gr}^0_{\res}(\mathcal{H})\\
& (X_{g_1}, X_{gH}) & \mapsto & p_*[(R_g)_*X_{g_1}] + (L_{g_1})_*X_{gH}.
\end{array}
$$
Therefore, for any $\alpha \in  T_{g_1gH}^*\operatorname{Gr}_{\res}(\mathcal{H})$,
$$
\begin{array}{ll}
\alpha \circ T_{(g_1, gH)} a_U(X_{g_1}, X_{gH}) &= \alpha(p_*[(R_g)_*X_{g_1}]) + \alpha((L_{g_1})_*X_{gH}) \\ &= R_g^*p^*\alpha(X_{g_{1}}) + L_{g_1}^*\alpha(X_{gH}).
\end{array}
$$
In other words 
$$
\alpha \circ T_{(g_1, gH)} a_U = R_g^*p^*\alpha + L_{g_1}^*\alpha,
$$ 
where $R_g^*p^*\alpha\in T_{g_1}\operatorname{U}_{\res}(\mathcal{H}) $ and $ L_{g_1}^*\alpha\in T_{gH}\operatorname{Gr}_{\res}(\mathcal{H})$.
In order to show that $a_U$ is a Poisson map, we have to show that
\begin{enumerate}
\item for any $\alpha \in  T_{g_1gH}^*\operatorname{Gr}_{\res}(\mathcal{H})$, the covector $R_g^*p^*\alpha$ belongs to  $$\mathbb{U}_{g_1} = R_{(g_1)^{-1}}^* L_{1,2}(\mathcal{H})/\mathfrak{u}_{1,2}(\mathcal{H}),$$ 
\item the Poisson tensors $\pi^{\operatorname{U}_\res}$ and $\pi^{Gr_\res}$ are related by
$$
\pi^{Gr_\res}_{g_1 gH}(\alpha, \beta) = \pi^{\operatorname{U}_\res}_{g_1}(R_g^*p^*\alpha, R_g^*p^*\beta) + \pi^{Gr_\res}_{gH}(L_{g_1}^*\alpha, L_{g_1}^*\beta).
$$
\end{enumerate}
For point (a), let us show that for $\alpha \in T^*_{g_1gH}\operatorname{Gr}_{\res}(\mathcal{H})$, and $g_1, g\in \operatorname{U}_{\res}(\mathcal{H})$, one has $R_{g_1}^*R_{g}^*p^*\alpha\in L_{1,2}(\mathcal{H})/\mathfrak{u}_{1,2}(\mathcal{H})$. Recall that $p^*\alpha$ can by identified with an element $L_{(g_1g)^{-1}}^*x_1$ where $x_1\in \mathfrak{h}^0$. Therefore 
$R_{g_1}^*R_{g}^*p^*\alpha  = \textrm{Ad}_{(g_1 g)^{-1}}^* x_1$. 
For $X\in T_{e}\operatorname{U}_{\res}(\mathcal{H}) $, one has
$$
\begin{array}{ll}
R_{g_1}^*R_{g}^*p^*\alpha(X) &= \Im \Tr x_1 \textrm{Ad}_{(g_1 g)^{-1}}(X) = \Im \Tr x_1 (g_1 g)^{-1} X g_1 g\\ &= \Im \Tr g_1 g x_1 (g_1 g)^{-1} X.
\end{array}
$$
Since $g_1 g x_1 (g_1 g)^{-1} \in L_{1,2}(\mathcal{H})$ for any $g_1, g\in \operatorname{U}_{\res}(\mathcal{H})$ and $x_1 \in L_{1,2}(\mathcal{H})$, it follows that  $R_{g_1}^*R_{g}^*p^*\alpha\in L_{1,2}(\mathcal{H})/\mathfrak{u}_{1,2}(\mathcal{H}).$

In order to prove (b), we will the cocycle identity. Note that for  $\alpha$, $\beta\in T^*_{g_1gH}\operatorname{Gr}_{\res}(\mathcal{H})$, 
one has
$$
\begin{array}{ll}
\pi^{Gr_\res}_{g_1gH}(\alpha, \beta) & = 
\pi^{\operatorname{U}_\res}_{g_1g}(p^*\alpha, p^*\beta) 
= \Pi^{\operatorname{U}_\res}_r(g_1 g)(R_{g_1 g}^*p^*\alpha, R_{g_1 g}^*p^*\beta)\\
& =\textrm{Ad}(g_1)^*\Pi^{\operatorname{U}_\res}_r(g)(R_{g_1g}^*p^*\alpha, R_{g_1g}^*p^*\beta) + \Pi^{\operatorname{U}_\res}_r(g_1)(R_{g_1g}^*p^*\alpha, R_{g_1g}^*p^*\beta) \\
&= \Pi^{\operatorname{U}_\res}_r(g)(L_{g_1}^*R_g^*p^*\alpha, L_{g_1}^*R_g^*p^*\beta) + \Pi^{\operatorname{U}_\res}_r(g_1)(R_{g_1}^*R_{g}^*p^*\alpha, R_{g_1}^*R_{g}^*p^*\beta) \\
& = \pi^{\operatorname{U}_\res}_g(L_{g_1}^*p^*\alpha, L_{g_1}^*p^*\beta) + \pi^{\operatorname{U}_\res}_{g_1}(R_{g}^*p^*\alpha, R_{g}^*p^*\beta) \\
 & =\pi^{\operatorname{U}_\res}_g(p^*L_{g_1}^*\alpha, p^*L_{g_1}^*\beta)+  \pi^{\operatorname{U}_\res}_{g_1}(R_{g}^*p^*\alpha, R_{g}^*p^*\beta)
\\ & = \pi^{Gr_\res}_{gH}(L_{g_1}^*\alpha, L_{g_1}^*\beta) +  \pi^{\operatorname{U}_\res}_{g_1}(R_{g}^*p^*\alpha, R_{g}^*p^*\beta).
\end{array}
$$
Hence the left action of $\operatorname{U}_{\res}(\mathcal{H})$ on$\operatorname{Gr}_{\res}(\mathcal{H})$ is a Poisson map.
\end{enumerate}

\end{proof}

\section{Poisson action of $\operatorname{B}_{\res}^{+}(\mathcal{H})$ on $\operatorname{Gr}_{\res}(\mathcal{H})$ and Schubert cells}\label{section8}

\subsection{Poisson action of $\operatorname{B}_{\res}^{\pm}(\mathcal{H})$ on $\operatorname{Gr}_{\res}(\mathcal{H})$}
The next Theorem shows that the action of $\operatorname{B}_{\res}^{\pm}(\mathcal{H})$ on $\operatorname{Gr}_{\res}(\mathcal{H})$ is a Poisson map, where $\operatorname{B}_{\res}^{\pm}(\mathcal{H})$ is endowed with the Banach Poisson--Lie group structure defined in Section~\ref{section6}, and where $\operatorname{Gr}_{\res}(\mathcal{H})$ is endowed with the Bruhat-Poisson structure defined in Section~\ref{section7}.

\begin{theorem}\label{actionB}
The following right action of $\operatorname{B}_{\res}^{\pm}(\mathcal{H})$ on $\operatorname{Gr}_{\res}(\mathcal{H}) = \operatorname{GL}_{\res}(\mathcal{H})/\operatorname{P}_{\res}(\mathcal{H})$  is a  Poisson map~:
$$
\begin{array}{lcll}
a_B~:& \operatorname{Gr}_{\res}(\mathcal{H})\times \operatorname{B}_{\res}^{\pm}(\mathcal{H})& \rightarrow & \operatorname{Gr}_{\res}(\mathcal{H})\\
&(g \operatorname{P}_{\res}(\mathcal{H}), b) & \mapsto & (b^{-1}g) \operatorname{P}_{\res}(\mathcal{H}).
\end{array}
$$
\end{theorem}

\begin{proof}
The tangent map to the action $a_B$ reads
$$
\begin{array}{lcll}
T_{(gH, b)}a_B~: &T_{gH}\operatorname{Gr}_{\res}(\mathcal{H})\oplus T_b\operatorname{B}_{\res}^{\pm}(\mathcal{H})& \rightarrow & T_{b^{-1}gP_{\res}}\operatorname{Gr}_{\res}(\mathcal{H})\\
&(X_{gH}, X_b) & \mapsto & (L_{(b^{-1})})_*X_{gH} - p_*(R_{g})_*(b^{-1} X_b b^{-1}).
\end{array}
$$
Therefore, for any $\alpha\in T^*_{b^{-1}gP_{\res}}\operatorname{Gr}_{\res}(\mathcal{H})$, 
$$
\begin{array}{ll}
\alpha \circ T_{(gH, b)}a_B(X_{gH}, X_b) &=  \alpha((L_{(b^{-1})})_*X_{gH}) - \alpha(p_*(R_{g})_*b^{-1} X_b b^{-1})\\
& = L_{b^{-1}}^*\alpha(X_{gH}) - R_{b^{-1}}^*L_{b^{-1}}^*R_g^* p^*\alpha(X_b),
\end{array}
$$
and $$\alpha \circ T_{(gH, b)}a_B = L_{b^{-1}}^*\alpha - R_{b^{-1}}^*L_{b^{-1}}^*R_g^* p^*\alpha,$$ where $L_{b^{-1}}^*\alpha\in T^*_{gH}\operatorname{Gr}_{\res}(\mathcal{H})$ and $R_{b^{-1}}^*L_{b^{-1}}^*R_g^* p^*\alpha \in T^*_b\operatorname{B}_{\res}^{\pm}(\mathcal{H})$. 
\begin{enumerate}
\item[(a)]
Let us show that for any  $\alpha\in T^*_{b^{-1}gP_{\res}}\operatorname{Gr}_{\res}(\mathcal{H})$ and any $b\in \operatorname{B}_{\res}^{\pm}(\mathcal{H})$, the form $R_{b^{-1}}^*L_{b^{-1}}^*R_g^* p^*\alpha$ belongs to $\mathbb{B}_b = R_{b^{-1}}^*L_{1,2}(\mathcal{H})/\mathfrak{b}_{1,2}(\mathcal{H})$. Recall that $\alpha$ can be identified with an element $L_{(b^{-1}g)^{-1}}^*x_1$ where $x_1\in \mathfrak{h}^0$. For $X\in T_e \operatorname{B}_{\res}^{\pm}(\mathcal{H})$, one has
$$
\begin{array}{ll}
L_{b^{-1}}^*R_g^* p^*\alpha(X) &= \alpha(p_* R_{g^*} (L_{b^{-1}})_* X) = \Im \Tr x_1 (L_{g^{-1}b})_*p_* R_{g^*} (L_{b^{-1}})_* X)
\\ &=  \Im \Tr x_1 p_*(\textrm{Ad}(g^{-1}) X)
= \Im \Tr p_{\mathfrak{b}_2^+}(x_1) g^{-1}X g \\ &= \Im \Tr g p_{\mathfrak{b}_2^+}(x_1) g^{-1} X.
\end{array}
$$
Recall that for $x_1 = \left(\begin{smallmatrix} A & B\\ C & D\end{smallmatrix}\right) \in \mathfrak{h}^0$, $p_{\mathfrak{b}_2^+}(x_1) =  \left(\begin{smallmatrix} 0& B+C^*\\ 0 & 0\end{smallmatrix}\right)$.
Since for any $g\in GL_{\res}(\mathcal{H})$ and any $x_1\in \mathfrak{h}^0$, $g p_{\mathfrak{b}_2^+}(x_1) g^{-1}\in L_{1,2}(\mathcal{H})$, the form $R_{b^{-1}}^*L_{b^{-1}}^*R_g^* p^*\alpha$ belongs to $\mathbb{B}_b$.
\item[(b)]
Let us show that the Poisson tensors $\pi^{B^+_\res}$ and $\pi^{Gr_\res}$ are related by
$$
\pi^{Gr_\res}_{b^{-1} gP_{\res}}(\alpha, \beta) =  \pi^{Gr_\res}_{gH}(L_{b^{-1}}^*\alpha, L_{b^{-1}}^*\beta)+ \pi^{B^+_\res}_b(R_{b^{-1}}^*L_{b^{-1}}^*R_g^* p^*\alpha, R_{b^{-1}}^*L_{b^{-1}}^*R_g^* p^*\beta ).
$$
One has
$$
\begin{array}{ll}
 \pi^{B^+_\res}_b(R_{b^{-1}}^*L_{b^{-1}}^*R_g^* p^*\alpha, R_{b^{-1}}^*L_{b^{-1}}^*R_g^* p^*\beta)
& = \Pi^{B^+_\res}_{r}(b)([g p_{\mathfrak{b}_2^+}(x_1) g^{-1}]_{\mathfrak{b}_{1,2}^+}, [g p_{\mathfrak{b}_2^+}(x_2) g^{-1}]_{\mathfrak{b}_{1,2}^+})\\
& = \Im \Tr \left(b^{-1} p_{\mathfrak{u}_{2}^+}(g p_{\mathfrak{b}_2^+}(x_1) g^{-1}) b\right) \left[p_{\mathfrak{b}_{2}^+}(b^{-1}p_{\mathfrak{u}_2^+}(g p_{\mathfrak{b}_2^+}(x_2) g^{-1})b)\right]\\
& = \Im \Tr p_{\mathfrak{u}_{2}^+}(g p_{\mathfrak{b}_2^+}(x_1) g^{-1})  b  \left[p_{\mathfrak{b}_{2}^+}(b^{-1}p_{\mathfrak{u}_2^+}(g p_{\mathfrak{b}_2^+}(x_2) g^{-1})b)\right] b^{-1}\\
& = \Im \Tr (b^{-1}g p_{\mathfrak{b}_2^+}(x_1) g^{-1}b) \left[p_{\mathfrak{b}_{2}^+}(b^{-1}p_{\mathfrak{u}_2^+}(g p_{\mathfrak{b}_2^+}(x_2) g^{-1})b)\right] \\
& = \Im \Tr (b^{-1}g p_{\mathfrak{b}_2^+}(x_1) g^{-1}b) \left[p_{\mathfrak{b}_{2}^+}(b^{-1}g p_{\mathfrak{b}_2^+}(x_2) g^{-1}b)\right] \\
&\quad  - \Im \Tr (b^{-1}g p_{\mathfrak{b}_2^+}(x_1) g^{-1}b) \left[p_{\mathfrak{b}_{2}^+}(b^{-1}p_{\mathfrak{b}_2^+}(g p_{\mathfrak{b}_2^+}(x_2) g^{-1})b)\right] \\

\end{array}
$$
Therefore
\begin{equation}\label{sommedesdeux}
\begin{array}{l}
 \pi^{B^+_\res}_b(R_{b^{-1}}^*L_{b^{-1}}^*R_g^* p^*\alpha, R_{b^{-1}}^*L_{b^{-1}}^*R_g^* p^*\beta) = \\
\quad  \Im \Tr (b^{-1}g p_{\mathfrak{b}_2^+}(x_1) g^{-1}b) \left[p_{\mathfrak{b}_{2}^+}(b^{-1}g p_{\mathfrak{b}_2^+}(x_2) g^{-1}b)\right]  - \Im \Tr (g p_{\mathfrak{b}_2^+}(x_1) g^{-1}) \left[p_{\mathfrak{b}_2^+}(g p_{\mathfrak{b}_2^+}(x_2) g^{-1})\right].
\end{array}
\end{equation}
On the other hand
$$
\begin{array}{ll}
 \pi^{Gr_\res}_{gH}(L_{b^{-1}}^*\alpha, L_{b^{-1}}^*\beta) & = \Pi^{\operatorname{U}_{\res}}_r(g)([gp_{\mathfrak{b}_2^+}(x_1) g^{-1}], [gp_{\mathfrak{b}_2^+}(x_2) g^{-1}])\\
 & = \Im \Tr (g^{-1} p_{\mathfrak{b}_2^+}(g p_{\mathfrak{b}_2^+}(x_1) g^{-1}) g) \left[p_{\mathfrak{u}_2^+} (g^{-1} p_{\mathfrak{b}_2^+}(g p_{\mathfrak{b}_2^+}(x_1) g^{-1}) g)\right]\\
  & = \Im \Tr  p_{\mathfrak{b}_2^+}(x_1) \left[p_{\mathfrak{u}_2^+} (g^{-1} p_{\mathfrak{b}_2^+}(g p_{\mathfrak{b}_2^+}(x_1) g^{-1}) g)\right]\\
  & = \Im \Tr  p_{\mathfrak{b}_2^+}(x_1)  (g^{-1} p_{\mathfrak{b}_2^+}(g p_{\mathfrak{b}_2^+}(x_1) g^{-1}) g\\
  & = \Im \Tr g p_{\mathfrak{b}_2^+}(x_1) g^{-1} p_{\mathfrak{b}_2^+}(g p_{\mathfrak{b}_2^+}(x_1) g^{-1})
\end{array}
$$
which is the second term in the right hand side of equation \eqref{sommedesdeux} with the opposite sign. Moreover, since 
$$\operatorname{Gr}_{\res}(\mathcal{H}) = \operatorname{GL}_{\res}(\mathcal{H})/\operatorname{P}_{\res}(\mathcal{H}) =  \operatorname{U}_{\res}(\mathcal{H})/\left(\operatorname{U}(\mathcal{H}_+)\times\operatorname{U}(\mathcal{H}_-)\right)
$$
there exist $g_1\in \operatorname{U}_{\res}(\mathcal{H})$ and $p_1 \in P_{\res}(\mathcal{H})$ such that $b^{-1} g = g_1 p_1$. In fact, the pair $(g_1, p_1)$ is defined modulo the right action by $H$ given by $ (g_1, p_1) \cdot h = (g_1 h, h^{-1} p_1)$. It follows that the first term in the right hand side of equation \eqref{sommedesdeux} reads
$$
\begin{array}{l}
\Im \Tr (b^{-1}g p_{\mathfrak{b}_2^+}(x_1) g^{-1}b) \left[p_{\mathfrak{b}_{2}^+}(b^{-1}g p_{\mathfrak{b}_2^+}(x_2) g^{-1}b)\right]\\
\quad = \Im \Tr (g_1 p_1 p_{\mathfrak{b}_2^+}(x_1) p_1^{-1} g_1^{-1}) \left[p_{\mathfrak{b}_{2}^+}(g_1 p_1 p_{\mathfrak{b}_2^+}(x_2) p_1^{-1} g_1^{-1})\right]\\
\end{array}
$$
Recall that for any $x_1 = \left(\begin{smallmatrix} A & B\\ C & D\end{smallmatrix}\right)\in \mathfrak{h}^0$, one has 
$$
x_1 = \left(\begin{smallmatrix} A & -C^*\\ C & D\end{smallmatrix}\right) + \left(\begin{smallmatrix} 0 & B+C^*\\ 0 & 0\end{smallmatrix}\right),
$$
with $p_{\mathfrak{u}_2^+}(x_1) = \left(\begin{smallmatrix} A & -C^*\\ C & D\end{smallmatrix}\right)$ and $p_{\mathfrak{b}_2^+}(x_1) = \left(\begin{smallmatrix} 0 & B+C^*\\ 0 & 0\end{smallmatrix}\right)$. 
Note that for any $p_1 = \left(\begin{smallmatrix} P_1 & P_2\\ 0 & P_3\end{smallmatrix}\right)\in P_{\res}(\mathcal{H})$, one has 
$$
p_1^{-1} = \left(\begin{smallmatrix} P_1^{-1} & -P_1^{-1} P_2 P_3^{-1}\\ 0 & P_3^{-1}\end{smallmatrix}\right)\in P_{\res}(\mathcal{H}), 
$$
and 
$$
p_1 p_{\mathfrak{b}_2^+}(x_1) p_1^{-1} = \left(\begin{smallmatrix} 0 & P_1 (B+C^*) P_3^{-1}\\ 0 & 0\end{smallmatrix}\right)\in \mathfrak{b}_{2}^+(\mathcal{H}).
$$
Therefore
$$
\begin{array}{l}
\Im \Tr (b^{-1}g p_{\mathfrak{b}_2^+}(x_1) g^{-1}b) \left[p_{\mathfrak{b}_{2}^+}(b^{-1}g p_{\mathfrak{b}_2^+}(x_2) g^{-1}b)\right]\\
\quad = \Im \Tr (g_1 p_1 p_{\mathfrak{b}_2^+}(x_1) p_1^{-1} g_1^{-1}) \left[p_{\mathfrak{b}_{2}^+}(g_1 p_1 p_{\mathfrak{b}_2^+}(x_2) p_1^{-1} g_1^{-1})\right]\\
\quad = \Pi^{\operatorname{U}_{\res}}_r(g_1)([g_1 p_1p_{\mathfrak{b}_2^+}(x_1) p_1^{-1}g_1^{-1}], [g_1 p_1 p_{\mathfrak{b}_2^+}(x_2) p_1^{-1}g_1^{-1}])\\
\quad = \Pi^{\operatorname{U}_{\res}}_r(g_1)([b^{-1}gp_{\mathfrak{b}_2^+}(x_1) g^{-1}b], [b^{-1}g p_{\mathfrak{b}_2^+}(x_2) g^{-1}b])\\
\quad = \pi^{Gr_\res}_{g_1 H}(\alpha, \beta) = \pi^{Gr_\res}_{b^{-1} g P_{\res}}(\alpha, \beta).
\end{array}
$$
It follows that the right action of $\operatorname{B}_{\res}^{+}(\mathcal{H})$ on $\operatorname{Gr}_{\res}(\mathcal{H}) $ is a Poisson map.
\end{enumerate}

\end{proof}
\subsection{Schubert cells of the restricted Grassmannian}

In this Section, $\mathcal{H}$ will be specified to be the space $L^{2}(\mathbb{S}^1, \mathbb{C})$ of complex square-integrable functions defined almost everywhere on the unit circle $\mathbb{S}^1 = \{z\in\mathbb{C}, |z| = 1\}$ modulo the equivalence relation that identifies two functions that are equal almost everywhere. In that case, the inner product of two elements $f$ and $g$ in $L^2(\mathbb{S}^1, \mathbb{C})$ reads $\langle f , g \rangle = \int_{\mathbb{S}^1} \overline{f(z)} g(z) d\mu(z)$, where $d\mu(z)$ denotes the Lebesgue mesure on the circle.
Let us recall some geometric facts about the restricted Grassmannian that were established in \cite{PS88}, Chapter~7.
Set $\mathcal{H}_+ = \textrm{span}\{z^n, n\geq 0\}$ and $\mathcal{H}_- = \textrm{span}\{z^n , n<0\}$.

The restricted Grassmannian admits  a stratification $\{\Sigma_S, S\in\mathcal{S}\}$ as well as a decomposition into Schubert cells $\{\mathcal{C}_S, S\in \mathcal{S}\}$, which are dual to each other in the following sense~:
\begin{enumerate}
\item[(i)] the same set $\mathcal{S}$ indexes the cells $\{\mathcal{C}_S\}$ and the strata $\{\Sigma_S\}$;
\item[(ii)] the dimension of $\mathcal{C}_S$ is the codimension of $\Sigma_S$;
\item[(iii)] $\mathcal{C}_S$ meets $\Sigma_S$ transversally in a single point, and meets no other stratum of the same codimension.
\end{enumerate}
A element $S$ of the set $\mathcal{S}$ is a subset of $\mathbb{Z}$, which is bounded from below and contains all sufficiently large integers. Given $S\in\mathcal{S}$, define the subspace $\mathcal{H}_S$ of the restricted Grassmannian $\operatorname{Gr}_{\res}(\mathcal{H})$ by~:
$$
\mathcal{H}_S = \overline{\textrm{span}\{z^s, s\in S\}}.
$$
Recall the following Proposition~:
\begin{proposition}[Proposition~7.1.6 in \cite{PS88}]
For any $W\in \operatorname{Gr}_{\res}(\mathcal{H})$ there is a set $S\in \mathcal{S}$ such that the orthogonal projection $W\rightarrow \mathcal{H}_S$ is an isomorphism. In other words the sets $\{\mathcal{U}_S, S\in\mathcal{S}\}$, where 
$$
\mathcal{U}_S = \{ W \in  \operatorname{Gr}_{\res}(\mathcal{H}), \textrm{the orthogonal projection  } W\rightarrow \mathcal{H}_S \textrm{ is an isomorphism}\},
$$
form an open covering of $\operatorname{Gr}_{\res}(\mathcal{H})$.
\end{proposition}
Following \cite{PS88}, let us introduce the Banach Lie groups  $\operatorname{N}^+_\res(\mathcal{H})$ and $\operatorname{N}^-_\res(\mathcal{H})$~:
$$
\operatorname{N}^+_\res(\mathcal{H}) =  \{A \in \operatorname{GL}_\res(\mathcal{H}), A(z^k\mathcal{H}_+) = z^{k}\mathcal{H}_+ \textrm{ and } (A-\textrm{Id})(z^k\mathcal{H}_+)\subset z^{k+1}\mathcal{H}_+, ~\forall k\in\mathbb{Z}\},
$$
$$
\operatorname{N}^-_\res(\mathcal{H}) =  \{A \in \operatorname{GL}_\res(\mathcal{H}), A(z^k\mathcal{H}_-) = z^{k}\mathcal{H}_- \textrm{ and } (A-\textrm{Id})(z^k\mathcal{H}_-)\subset z^{k+1}\mathcal{H}_-, ~\forall k\in\mathbb{Z}\}.
$$
In other words, the group $\operatorname{N}^\pm_\res(\mathcal{H})$ is the subgroup of $\operatorname{B}_{\res}^{\pm}(\mathcal{H})$ consisting of the triangular operators with respect to the basis $\{|n\rangle := z^n, n\in\mathbb{Z}\}$ which have only $1$'s on the diagonal.

\begin{proposition}
The Banach Lie group $\operatorname{N}^\pm_\res(\mathcal{H})$ is a normal subgroup of $\operatorname{B}_{\res}^{\pm}(\mathcal{H})$ and the quotient group $\operatorname{B}_{\res}^{\pm}(\mathcal{H})/\operatorname{N}^\pm_\res(\mathcal{H})$ is isomorphic to the group of bounded linear positive definite operators which are diagonal with respect to the orthonormal basis $\{|z^{k}\rangle, k\in\mathbb{Z}\}$.
\end{proposition}

\begin{proof}
For a triangular operator $g\in \operatorname{B}_{\res}^{\pm}(\mathcal{H})$, the diagonal coefficients of $g$ and $g^{-1}$ are inverses of each other~:  $\langle n| g^{-1}n\rangle = \langle n| gn\rangle^{-1}$, $\forall n\in\mathbb{Z}$. Therefore, for any element $h\in \operatorname{N}^\pm_\res(\mathcal{H})$, the composed operator $g h g^{-1}$ has only 1's on it's diagonal and belong to 
$\operatorname{N}^\pm_\res(\mathcal{H})$.
This implies that $\operatorname{N}^\pm_\res(\mathcal{H})$ is a normal subgroup of $\operatorname{B}_{\res}^{\pm}(\mathcal{H})$.
Recall that $D$ denotes the  linear transformation consisting in taking the diagonal part of a linear operator (see equation \eqref{diagonal}).
Since $|\langle n| D(A) m\rangle|\leq \|A\|$ and $D(A)$ is diagonal, the linear transformation $D$ maps bounded operators to bounded operators. By the definition of $\operatorname{B}_{\res}^{\pm}(\mathcal{H}) $,  the range of $D~: \operatorname{B}_{\res}^{\pm}(\mathcal{H}) \rightarrow L_{\infty}(\mathcal{H})$ is the group of bounded linear positive definite operators which are diagonal with respect to the orthonormal basis $\{|z^{k}\rangle~: k\in\mathbb{Z}\}$. Moreover, the kernel of $D~: \operatorname{B}_{\res}^{\pm}(\mathcal{H}) \rightarrow L_{\infty}(\mathcal{H})$ is exactly $\operatorname{N}^\pm_\res(\mathcal{H})$.
\end{proof}

\begin{proposition}\label{CsB}
\begin{enumerate}
\item[(i)] The cell $\mathcal{C}_S$ is the orbit of $\mathcal{H}_S$ under $\operatorname{B}_{\res}^{+}(\mathcal{H})$.
\item[(ii)] The stratum $\Sigma_S$ is the orbit of $\mathcal{H}_S$ under $\operatorname{B}_{\res}^{-}(\mathcal{H})$.
\end{enumerate}
\end{proposition}

\begin{proof}
It follows from Proposition~7.4.1 in \cite{PS88}, that the cell $\mathcal{C}_S$ is the orbit of $\mathcal{H}_S$ under $\operatorname{N}^+_\res(\mathcal{H})$. Symmetrically, it follows from Proposition~7.3.3 in \cite{PS88}, that the stratum $\Sigma_S$ is the orbit of $\mathcal{H}_S$ under $\operatorname{N}^-_\res(\mathcal{H})$.
Since the diagonal part of an operator in $\operatorname{B}_{\res}^{\pm}(\mathcal{H})$ acts trivially, one gets the same result replacing $\operatorname{N}^\pm_\res(\mathcal{H})$ by $\operatorname{B}_{\res}^{\pm}(\mathcal{H})$.
\end{proof}



Recall that the restricted Grassmannian is a Hilbert manifold endowed with the Poisson structure constructed in Theorem~\ref{Poisson}. In this Hilbert context, the Poisson tensor 
$\pi^{Gr_{\res}}$ defines a bundle map $\left(\pi^{Gr_{\res}}\right)^\sharp~:T^*\operatorname{Gr}_{\res}(\mathcal{H}) \rightarrow T\operatorname{Gr}_{\res}(\mathcal{H})$. The range of this map is called the \textbf{characteristic distribution} of the Poisson structure, and the maximal integral submanifolds of this distribution are called \textbf{symplectic leaves} (see \cite{OR03} Section~7 for a general discussion on characteristic distributions and symplectic leaves in the Banach context).
\begin{theorem}\label{cells}
The Schubert cells $\{\mathcal{C}_S, S\in \mathcal{S}\}$ are the symplectic leaves of $\operatorname{Gr}_{\res}(\mathcal{H})$.
\end{theorem}

\begin{proof}
The integrability of the characteristic distribution follows from Theorem~6 in \cite{P12}, since $\operatorname{Gr}_{\res}(\mathcal{H})$ is a Hilbert manifold. 
The fact that the symplectic leaves are the orbits of $\operatorname{B}^+_{\res}(\mathcal{H})$ follows from the construction as in the finite-dimensional case (see Theroem~4.6 (3) in \cite{LW90}). It follows from Proposition~\ref{CsB} that the orbits of $\operatorname{B}^+_{\res}(\mathcal{H})$ coincide with the Schubert cells $\{\mathcal{C}_S, S\in \mathcal{S}\}$.
\end{proof}


%
\begin{remark}{\rm
Let $\Gamma^+$ be the group of real-analytic functions $g~:\mathbb{S}^1 \rightarrow \mathbb{C}^*$, which extend to holomorphic functions $g$ from the unit disc $\mathbb{D} = \{z\in\mathbb{C}~:|z|\leq 1\}$ to $\mathbb{C}^*$, satisfying $g(0) = 1$.  Any such function $g\in\Gamma^+$ can be written $g = e^f$, where $f$ is a holomorphic function on $\mathbb{D}$ such that $f(0) = 0$.
The group $\Gamma^+$ acts by multiplication operators on $\mathcal{H}$ and therefore also on $\operatorname{Gr}_{\res}(\mathcal{H})$. As explained in \cite{SW85} (see  Proposition~5.13), the action of $\Gamma^+$  on (some subgrassmannians of) $\operatorname{Gr}_{\res}(\mathcal{H})$ generates the KdV hierarchy. It is easy to see that $\Gamma^+\subset \operatorname{B}^+_{\res}(\mathcal{H})$.
Indeed, by Proposition~2.3 in \cite{SW85}, $\Gamma^+ \subset \operatorname{GL}_{\res}(\mathcal{H}) := \operatorname{GL}(\mathcal{H})\cap \operatorname{L}_{\res}(\mathcal{H})$. Since $g\in\Gamma^+$ is holomorphic in $\mathbb{D}$ and satisfies $g(0) = 1$, the Fourier decomposition of $g$ reads $g(z) = 1 + \sum_{k>0} g_k z^k$. Therefore $g(z)\cdot z^n = z^n + \sum_{k>0} g_k z^{k+n}$. It follows that the multiplication operator by $g$ is a upper triangular operator $M_g\in \operatorname{B}^+_{\res}(\mathcal{H})$, with diagonal elements equal to $1$. Therefore $\Gamma^+\subset \operatorname{B}^+_{\res}(\mathcal{H})$ and, by Theorem~\ref{actionB}, the action of  $\operatorname{B}^+_{\res}(\mathcal{H})$ on the restricted Grassmannian is a Poisson action. }
\end{remark}

\begin{acknowledgment}{\footnotesize \rm
I would first and foremost like to thank T.S.~Ratiu  who brought the theory of Poisson--Lie groups to my attention when I was a post-graduate in Lausanne. It was only much later that I understood it to be the key to understanding   \cite{SW85} via the Poisson theory. I could not have written my paper without the help of D.~Belti\c t\u a who not only gave me the references  indicating where the problem of triangulating an operator was studied, but who also gave me the electronic version of the documents when I was unable to go to any library. My paper was finally produced thanks to the support of the CNRS, of the University of Lille (France), in particular thanks to the CEMPI Labex  (ANR-11-LABX-0007-01), as well as to the Pauli Institute in Vienna (Austria) offering very nice working conditions. The special assistance given by the CNRS for the promotion of women scientists was crucial for the inception of this paper  which also benefited from the exchange in the early months of 2015 with other researchers during the Shape Analysis programme at the Erwin Schrodinger Institute  in Vienna. The discussions with C. Vizman, K.-H. Neeb and F. Gay-Balmaz were particularly fruitful as were the WGMP lectures (exchanges with D.~Belti\c t\u a, T.~Golinski,  F.\,Pelletier, C.~Roger and, last but not least, G.~Larotonda). I~am grateful to the anonymous referees for their pertinent comments which made me improve the quality of this paper. Finally I heartily thank Y.~Kosmann-Schwarzbach and D.~Bennequin for their  valuable appreciations of a presentation of my paper.

}
\end{acknowledgment}

\end{document}